\documentclass[pdflatex,sn-mathphys-num]{sn-jnl}% 

\usepackage{stmaryrd}

\usepackage{graphicx}%
\usepackage{multirow}%
\usepackage{amsmath,amssymb,amsfonts}%
\usepackage{amsthm}%
\usepackage{mathrsfs}%
\usepackage[title]{appendix}%
\usepackage{xcolor}%
\usepackage{textcomp}%
\usepackage{manyfoot}%
\usepackage{booktabs}%
\usepackage{algorithm}%
\usepackage{algorithmicx}%
\usepackage{algpseudocode}%
\usepackage{listings}%
\usepackage{comment}
\usepackage[framemethod=tikz]{mdframed}
\newenvironment{proofline}[1][\unskip]
  {\par\noindent\normalfont\textit{Proof #1.}\par\nopagebreak%
  \begin{mdframed}[
     linewidth=1pt,
     linecolor=black,
     bottomline=false,topline=false,rightline=false,
     innerrightmargin=0pt,innertopmargin=0pt,innerbottommargin=0pt,
     innerleftmargin=1em,% Distance between vertical rule & proof content
     skipabove=.5\baselineskip
   ]
   }
  {%\end{small}
  \end{mdframed}
  
  \phantom{a}
  
  }
  
\theoremstyle{plain}%
\newtheorem{lemma}{Lemma}[section]
\newtheorem{theorem}{Theorem}[section]

\newtheorem{corollary}{Corollary}[section]
\newtheorem{proposition}[theorem]{Proposition}% 

\theoremstyle{remark}%
\newtheorem{example}{Example}[section]%
\newtheorem{remark}{Remark}[section]%

\theoremstyle{definition}%
\newtheorem{definition}{Definition}[section]%

\newcommand{\anb}[1][]{V_{#1}}

\newcommand{\str}[1][]{\Sigma_{#1}}
\newcommand{\strp}[1][]{\Sigma_{#1}^{(2\pi)}}

\numberwithin{equation}{section}

\begin{document}

\title%[Article Title]
{Asymptotic behaviour of determinants through the expansion of the Moyal star product}

\author*[1]{\fnm{Maurizio} \sur{Fagotti}}\email{maurizio.fagotti@universite-paris-saclay.fr}
\equalcont{The authors contributed equally to this work.}

\author[2]{\fnm{Vanja} \sur{Mari\'c}}
% \email{vanja.maric@universite-paris-saclay.fr}
\equalcont{The authors contributed equally to this work.}

\affil*[1]{\orgdiv{Universit\'e Paris-Saclay}, \orgname{CNRS, LPTMS},  
\city{Orsay}, \postcode{91405}, 
\country{France}}

\abstract{We work out a generalization of the Szeg\"o limit theorems on the determinant of large matrices.  
We focus on matrices with nonzero leading principal minors and elements that decay to zero  exponentially fast with the distance from the main diagonal, but we relax the constraint of the Toeplitz structure. We obtain an expression for the asymptotic behaviour of the determinant written in terms of the factors of a left and right Wiener-Hopf type factorization of an appropriately defined symbol. For matrices with elements varying slowly along the diagonals (e.g., in locally Toeplitz sequences), we propose to apply the analogue of the semiclassical expansion of the Moyal star product in phase-space quantum mechanics. This is a systematic method that provides approximations up to any order in the typical scale of the inhomogeneity and allows us to obtain explicit asymptotic formulas.}

\keywords{determinant, Toeplitz, Moyal star product, Wiener-Hopf factorization}

\maketitle

\tableofcontents

\section{Introduction}\label{sec1}

Determinants of Toeplitz matrices are ubiquitous in physical problems with translational symmetry: every time that a quantity can be associated with a matrix whose indices are identified with the position, translational invariance manifests itself in the matrix being (block-)Toeplitz. 
On the other hand, when translational symmetry is broken,  the Toeplitz structure breaks down and only few cases can be addressed as powerfully and effectively (for example, when the matrices exhibit particular Toeplitz+Hankel structures~\cite{Deift2011Asymptotics,Gharakhloo2020A,Fagotti2011Universal,Stephan2014,Fagotti2016Local,Bonsignori2021}). 
More generally, moving from translationally invariant---homogeneous---settings to inhomogeneous ones is like moving from Toeplitz to generic matrices. 

A very common physical situation is when the typical scale of the inhomogeneity is controlled by a parameter, in the sense that the inhomogeneity becomes milder and milder as the parameter is increased. Such settings are often described by local density approximations or hydrodynamic theories, which generally become effective in the limit of low inhomogeneity. There are however systems in which (generalised) hydrodynamics turns out to be an exact alternative description: we mention, for example,  spin-$\frac{1}{2}$ chain systems that are dual to free fermions~\cite{Fagotti2020Locally}. The absence of interactions manifests itself in the applicability of the Wick's theorem, which allows one to  fully characterize the system by means of a single matrix filled with the two-point fermionic correlations. In the presence of translational symmetry, the correlation matrix is a \mbox{(block-)Laurent} operator, and many quantities of interest reduce to the calculation of determinants of (block-)Toeplitz matrices;  estimating the same quantities in the presence of inhomogeneity, on the other hand, presents itself as a much more difficult problem. Here is where  generalised hydrodynamics takes the lead. Such a theory provides an alternative exact description, which defines a phase-space formulation of quantum mechanics where matrices are replaced by functions~\cite{Zachos2005Quantum}. The latter are nothing but the symbols of the matrices and live in a space where multiplication is represented by the so-called Moyal star product~\cite{Groenewold1946On,Moyal1949,Zachos2005Quantum,Rieffel1993,Waldmann2019}. Notwithstanding the phase-space formulation being completely equivalent to the standard description in terms of the correlation matrix, the former is superior when the symbol varies over large scales, as it allows for asymptotic expansions. 

This was the motivation underlying our interest in the behaviour of determinants of matrices when the Toeplitz structure breaks down. In this respect, we point out Tilli's work~\cite{Tilli1998Locally} on locally Toeplitz sequences. 
He considered indeed matrices that resemble Toeplitz ones when inspecting a neighbourhood of elements, which is analogous to the low inhomogeity limit mentioned above. Imposing a particular structure (later relaxed to some extent by Serra-Capizzano and collaborators~\cite{Capizzano2003Generalized,Garoni2017Generalized}), Tilli obtained a result that can be read as the analogue of the weak Szeg\"o limit theorem. In this respect, we also mention Refs~\cite{Roch2007Szego,Ehrhardt2011A}, which focussed on the generalisation of the Szeg\"o limit theorems to operators with almost periodic diagonals. 

In this paper,
the analogy with the theory of Toeplitz operators 
will be used to develop a systematic method to approximate the determinant of a class of matrices that arise in the study of physical systems with inhomogeneities.  
We start  
by deriving a Borodin-Okounkov-Case–Geronimo type formula  as well as  Szeg\"o limit type theorems for a class of matrices with elements that decay to zero exponentially fast with the distance from the main diagonal. We then specialise the generalization to matrices with local smoothness properties along the diagonals and exhibit asymptotic expressions for their determinants. Finally, we comment on the natural generalizations to matrices with an emergent block structure. 

\subsection{Notations}
Even if the notations will be explained the first time they are used, the main ones are collected here for easy reference. 
\begin{itemize}
\item $\mathbb Z$ denotes the set of all integers; $\mathbb{N}$ is the set of natural numbers, and $\mathbb{N}_0=\mathbb{N} \cup \{0\}$.
\item $\mathbb T$ denotes the unit circle $\{z\in\mathbb C: |z|=1\}$. 
\item $L^p$ and $\ell^p$ are the function and sequence spaces respectively.

\item $(a)_k$ ($k\in \mathbb Z$) are the Fourier coefficients of a function $a\in L^\infty(\mathbb T)$.

\item
$\mathbb{R}/(2\pi \mathbb{Z})$ and $\mathbb{C}/(2\pi \mathbb{Z})$ denote the quotient spaces of $\mathbb R$ and $\mathbb C$, respectively, by $2\pi \mathbb{Z}$ in which $x\sim x+2\pi$ (we will use this notation to define $2\pi$-periodic functions).

\item 
$
g^+:(x,p) \mapsto \sum_{j=0}^\infty\frac{1}{2\pi}\int_{-\pi}^\pi e^{i j (p-q)}g(x,q) dq
$, 

$
g^-:(x,p)\mapsto \sum_{j=-\infty}^{-1}\frac{1}{2\pi}\int_{-\pi}^\pi e^{i j (p-q)}g(x,q)dq
$, 

$\tilde g=g^+-g^-$.
\item A  $\pm$ symbol is indicated with a $\pm$ subscript, and we use the notation
$a_{\pm}:(x,z)\mapsto a_{\pm,x}(z)$.
\item $\star$ denotes the Moyal (Weyl-Groenewold) star product.
\item $a^{-1_\star}$ is the star inverse of $a$. 
\item $\log_\star a$ is the star logarithm of $a$.
\item $(a^{\overset{n}{\sim}})_x(z)=a_{n+1-x}(z^{-1})$, 

$a^\sim=a^{\overset{0}{\sim}}$.
\item $\texttt T^t: a_x(z)\mapsto a_{x+t}(z)$.
\item $[A,B]$ denotes the commutator $A B- BA$.
\item $\{a,b\}_M$ denotes the Moyal brackets $-i(a\star b-b\star a)$.
\item $\{a,b\}$ denotes the Poisson brackets $\partial_1 a\partial_2 b-\partial_2 a\partial_1 b$, where $a:(x,p)\mapsto a(x,p)$ and $b:(x,p)\mapsto b(x,p)$.
\end{itemize}

\subsection{Overview of the Szeg\"o limit theorems}
Here we review some results on Toeplitz determinants directly connected with this work. More details can be found, e.g., in textbook~\cite{Bottcher1999}.

Let  $a$ be a function in $L^\infty(\mathbb T)$. A Laurent operator $L(a)$ is a bounded linear operator on $\ell^2(\mathbb Z)$ and is represented by the doubly infinite matrix with elements $[L(a)]_{ij}=(a)_{i-j}$, where $i,j\in\mathbb Z$. Such matrices provide the matrix representation of multiplication operators on $L^\infty(\mathbb T)$ with respect to the orthonormal basis $\{\frac{e^{inp}}{\sqrt{2\pi}}\}_{n\in \mathbb Z}$ and, in particular, $L(a)L(b)=L(ab)$ for any $a,b\in L^\infty(\mathbb T)$. The function $a(z)=\sum_{j\in \mathbb Z}(a)_j z^{j}$, with $z\in\mathbb T$, is known as the symbol of $L(a)$. 

Closely related to Laurent operators are Toeplitz operators. The Toeplitz operator $T(a)$ induced by the symbol $a$ is the bounded linear operator  on $\ell^2(\mathbb N)$  represented by the infinite matrix with elements $[T(a)]_{ij}=(a)_{i-j}$, where $i,j\in\mathbb N$. Such matrices provide the matrix representation of multiplication operators on the Hardy spaces in which all the negative or all the positive Fourier coefficients vanish.

Toeplitz operators play a crucial role in the calculation of determinants of large Toeplitz matrices $T_n(a)$, which are finite sections of Toeplitz operators:
\begin{equation}
T_n(a)=\begin{pmatrix}
(a)_0&(a)_{-1}&(a)_{-2}&\hdots&(a)_{1-n}\\
(a)_1&(a)_0&(a)_{-1}&\hdots&(a)_{2-n}\\
(a)_2&(a)_1&(a)_{0}&\hdots&(a)_{3-n}\\
\vdots&\vdots&\vdots&\ddots&\vdots\\
(a)_{n-1}&(a)_{n-2}&(a)_{n-3}&\hdots&(a)_{0}
\end{pmatrix}\label{eq:Toeplitz}\, .
\end{equation}
One of the basic results on the large $n$ behaviours of such determinants is referred to as the strong Szeg\"o limit theorem: 
\begin{theorem}[Strong Szeg\"o Limit Theorem]
Let $a\in L^{\infty}(\mathbb T)$ satisfy
\begin{itemize}
\item[-] $\sum_{n\in\mathbb Z} |(a)_n|<\infty$
\item[-] $\sum_{n\in\mathbb Z}(|n|+1)(a)_n^2<\infty$
\item[-] $a(z)\neq 0$ for every $z\in\mathbb T$
\item[-] $a$ has zero winding number with respect to the origin.
\end{itemize}
Then
\begin{equation}\label{eq:strongS}
\lim_{n\rightarrow\infty}\frac{\det T_n(a)}{\exp(n(\log a)_0)}=e^{E(a)}, \quad E(a):=\sum_{k=1}^\infty k(\log a)_k(\log a)_{-k} \, .
\end{equation}
\end{theorem}
The theorem was proven initially by Szeg\"o~\cite{Szego1952} under stronger conditions on $a$ and was subsequently generalized by many authors
(see e.g. Refs~\cite{Simon2005book,Deift2013} for details and a historical account). Note that when the symbol is analytic and nonzero in an annulus around the unit circle, a case often relevant for physical applications, the limit is reached exponentially fast in $n$. A consequence of the strong Szeg\"o limit theorem is
\begin{equation}\label{eq:Szego1}
\lim_{n\rightarrow\infty}\frac{\det T_n(a)}{\det T_{n-1}(a)}=\exp((\log a)_0)\, ,
\end{equation}
which is referred to as the weak Szeg\"o limit theorem and historically predates the strong version.

These results have been generalised in different directions. Particularly important in physics are the following two generalisations. First, Fisher and Hartwig~\cite{Fisher1969,Hartwig1969} derived some formulas covering the case of a non-zero winding number and conjectured a formula dealing with singular and vanishing symbols, which has been proven afterwards (see e.g. Ref.~\cite{Krasovsky2011} for a review).  Second, some of the results were generalised to block-Toeplitz matrices~\cite{Widom1975,Widom1976,Bottcher2006}. 
Historically, the development on Toeplitz determinants had a direct connection with the investigations into the two-dimensional classical Ising model~\cite{Deift2013}. More recently, they have been  applied in the calculation of two-point correlation functions as well as R\'enyi/von Neumann entanglement entropies of subsystems in translationally invariant quantum many-body systems both in and out of equilibrium. 
Just to provide a list of recent works,  largely incomplete and manifestly biased but at least varied for applications and theorems  involved, we refer the reader to Refs~\cite{Jin2004,Keating2004, Franchini2005Asymptotics,Its2007Entropy, Ovchinnikov2007Fisher,Calabrese2010Universal,Gutman2011, Calabrese2012Quantum,Bernigau2015,Groha2018,Ares2019Sublogarithmic,Basor2019Modified,Bonsignori2019,Jones2019,Maric2021,Jones2023, Bocini2024No,Fagotti2024Nonequilibrium};  we cannot however fail to mention also applications in random matrix theory~\cite{Baik2016book}.

Nowadays, on the other hand, much attention is being paid to the behaviour of systems prepared in traps or in inhomogeneous states~\cite{Bastianello2022Introduction}. In those situations the quantities that, in homogeneous settings, could be described in terms of determinants of (block-)Toeplitz matrices are now related to determinants of matrices that could still exhibit some nice smoothness properties but that are far from being (block-)Toeplitz. 
Can one predict their asymptotic behaviour notwithstanding the breakdown of translational symmetry? We present a series of results that go in the desired direction. 
Specifically, we consider the fundamental case where, broadly speaking, the matrix elements decay rapidly as one moves away from the main diagonal. Our primary focus is on deriving asymptotic formulas suitable for physical applications, rather than on identifying the most general conditions under which those formulas hold. Indeed, some of the assumptions underlying our proofs could likely be weakened to some extent.
Importantly, any physical scenario in which our asymptotic results fail to capture the limiting behavior of the underlying sequence of matrices warrants particular attention, as it may reveal qualitatively different phenomena.

\subsection{Statement of the problem  and organization of the article}

We are interested in the large $n$ behavior of determinants of $n\times n$ matrices $T_n(a)$, with elements
\begin{equation}\label{eq:Tstar}
    (T_n(a))_{j,k}=\frac{1}{2\pi}\int_{-\pi}^\pi a_{\frac{j+k}{2}}(e^{ip})e^{-ip(j-k)}dp \; , \quad j,k=1,2,\ldots ,n,
\end{equation}
where $a_x$, with $x\in\frac{1}{2}\mathbb Z$, is an integrable function that determines the elements on the antidiagonal associated with $x$ as $a$ does in the Toeplitz matrix~\eqref{eq:Toeplitz}.
Such a parametrisation does not constrain the matrix elements for given $n$ but still 
provides a natural  generalisation of Toeplitz matrices; we keep the analogy alive by referring to $T_n$ as the ``star-Toeplitz'' matrix associated with $a$ (we will be more precise in a moment).

We can separate our results in two categories. First, we accommodate some standard theorems on Toeplitz and Hankel operators to the less structured operators we are interested in. Our underlying goal is to emphasize the analogy with the theory of Toeplitz and Hankel operators, and indeed we will not do more than adapting standard proofs to our situation. Proposition~\ref{t:BOCG}, its (weak Szeg\"o limit type) corollary, and Theorem~\ref{c:strong} will be the most general results obtained, but such a generality is arguably accompanied by impracticality: the problem of computing the determinant is simply moved to the problem of working out two factorisations. The second category of results deals with this latter problem. We restrict ourselves to matrices whose elements are slowly varying along the diagonals.  This allows us to work out the factorisations and obtain the results stated in Theorem~\ref{c:strong}, which we also use to infer the asymptotics of determinants in locally Toeplitz sequences. 

In this respect, Section~\ref{s:main} collects the main determinant formulas and its first two subsections  follow the aforementioned division in two categories.
Some definitions are required to read the statement of the theorems. We report first those that are needed in both of the subsections. The remaining, more specific, definitions will be provided when necessary. We point out that, even if the first subsection of Section~\ref{s:main} provides an overview of the framework, it is not preparatory to the second subsection, thus we encourage the reader most interested in application-oriented results to read Section~\ref{ss:asymp} first. %Finally, 
Section~\ref{ss:blockToeplitz} sketches the minor changes needed to generalise the results to matrices with an emergent block structure. Finally, Section~\ref{ss:examples} collects a list of examples.
The reader can find the lemmas and proofs underlying the results reported in Section~\ref{ss:readapt} and Section~\ref{ss:asymp} in Section~\ref{s:limit} and  Section~\ref{s:locally}, respectively.  
 
\section{Main results}\label{s:main}

Just as Laurent matrices represent Laurent operators  through the Fourier coefficients of their symbol, so we can use the same parametrisation as in \eqref{eq:Tstar}
to construct the following doubly-infinite matrix.

\phantom{a}

\begin{definition}
Let $a_x$, for $x\in\frac{1}{2}\mathbb{Z}$, be functions in $L^\infty(\mathbb{T})$ 
and denote by $a$ the function $a:\frac{1}{2}\mathbb{Z}\times \mathbb{T}\to \mathbb{C}$, $(x,z)\mapsto a_x(z)$. By $L(a)$ we denote the operator represented by the doubly infinite matrix
\begin{equation}
    L_{j,k}(a)=\frac{1}{2\pi}\int_{-\pi}^\pi a_{\frac{j+k}{2}}(e^{ip})e^{-ip(j-k)}dp \; , \quad j,k\in \mathbb{Z}\, . \label{matrix elements doubly infinite}
\end{equation}
\end{definition}

\begin{definition}\label{a:boundness}
For $\rho\in(0,1)$, we denote by $\anb[\rho]$ the set of all $a$ with the following properties:
\begin{enumerate}
\item For every $x\in\frac{1}{2}\mathbb Z$, $a_x(z)$ has an analytic extension to the annulus $\rho<|z|<\rho^{-1} $.
\item $a_x(z)$ is uniformly bounded  in $x$ for $\rho<|z|<\rho^{-1} $.
\end{enumerate}
When we do not need to be specific, we say that a symbol is in $\anb\equiv \cup_{\rho\in(0,1)}V_\rho$.
\end{definition}

\begin{quotation}
\begin{remark}
These properties provide a sufficient condition for the boundedness of the operator represented by \eqref{matrix elements doubly infinite}. Namely, if $a\in\anb[]$ then $L(a)$ is a bounded linear operator on $\ell^2(\mathbb Z)$ (Lemma~\ref{lemma bounded linear operators} in Appendix~\ref{a:lemmas}).
\end{remark}
\end{quotation}

\phantom{a}

\begin{definition}[Symbol]  
The symbol of $L(a)$ is the equivalence class of all functions generating the same $L(a)$
\begin{equation}
a\sim b\Leftrightarrow  \{(a_{\frac{n}{2}})_{2m-\frac{1-(-1)^n}{2}}=(b_{\frac{n}{2}})_{2m-\frac{1-(-1)^n}{2}}, \quad \forall m,n\in \mathbb Z\}\, .
\end{equation}
%In the following
For reason that will be apparent later, we will usually represent the symbol with a function $f$ defined on an extended domain,  $f:\mathbb R\times (\mathbb R/(2\pi\mathbb Z))\rightarrow
\mathbb C$, and we will write $a=[f]$ as a way to say
\begin{equation}
(a_{\frac{n}{2}})_{2m-\frac{1-(-1)^n}{2}}=\frac{1}{2\pi}\int_{-\pi}^\pi f(\tfrac{n}{2},p)e^{-i (2m-\frac{1-(-1)^n}{2}) p}d p
\qquad \forall n,m\in\mathbb Z\, .
\end{equation}
\begin{quotation}
\begin{remark}
If $a=[f]$ and $f:\mathbb R\times (\mathbb R/(2\pi\mathbb Z))\rightarrow
\mathbb C$ is independent of its first argument,  $L(a)$ is a Laurent operator and $T(a)$ is a Toeplitz operator.  
\end{remark}
\end{quotation}

\phantom{a}

\end{definition}

\phantom{a}

\subsection{Readaptation of standard results}
\label{ss:readapt}

Similarly to the relation between Laurent matrices and operators on $L^\infty(\mathbb T)$, \eqref{matrix elements doubly infinite} provides the matrix representation of the following star product:
\begin{definition}[Moyal-Weyl–Groenewold product] The  Moyal star product $a\star b$ of two symbols $a,b$ is defined by the double series
\begin{equation}\label{eq:Moyal_star}
(a\star b)_x(e^{ip})=\frac{1}{(2\pi)^2}\sum_{m\in \mathbb{Z}}\sum_{n\in \mathbb{Z}} e^{i(m+n)p}\int\limits_{[-\pi,\pi]^2} a_{x+\frac{m}{2}}(e^{iq_1})b_{x-\frac{n}{2}}(e^{iq_2})  e^{-i(nq_1+mq_2)} d^2 q \, .
\end{equation}
\end{definition}

\begin{quotation}
\begin{remark}\label{lemma star product analiticity}
If $a,b\in\anb$ then the double series converges absolutely and the Moyal product is well defined.  
In particular, the following properties  hold (Appendix~\ref{a:lemmas})
\begin{enumerate}
    \item $a\star b \in \anb$.
  \item $L(a\star b)=L(a)L(b)$.
    \item $(a\star b)\star c =a\star (b\star c)$.
\end{enumerate}
\end{remark}

\end{quotation}

\phantom{a}

\begin{definition}[Star inverse]
A symbol $a$ is star invertible if there is $a^{-1_\star}$ such that
$$
a\star a^{-1_\star}=a^{-1_\star}\star a=1\, .
$$
We then say that $a^{-1_\star}$ is the star inverse of $a$.
\end{definition}

Since $L(a)$ is a linear bounded operator on $\ell^2(\mathbb Z)$, a holomorphic function $f$ of $L(a)$ is defined by the holomorphic functional calculus
\begin{equation*}
f(L(a))=\frac{1}{2\pi i}\oint_C f(\zeta)(\zeta-L(a))^{-1}\, \mathrm d \zeta\, ,
\end{equation*}
where $C$, traversed anticlockwise, is a closed simple curve of the complex plane strictly surrounding the spectrum of $L(a)$. This induces, through $L(f_{\star}(a))=f(L(a))$, the following definition of star function
\begin{definition}[Star function] \label{d:starf}
Let $f(z)$ be analytic in a an open region including %inside 
a closed simple curve $C$ of the complex plane strictly surrounding the spectrum of $L(a)$; the star function $f_\star(a)$ reads
\begin{equation}\label{star function definition}
f_\star(a)=\frac{1}{2\pi i}\oint_C f(\zeta)(\zeta-a)^{-1_\star}\, \mathrm d \zeta\, ,
\end{equation}
where $C$ is traversed anticlockwise.
\end{definition}

%\begin{quotation}
%\begin{remark}
%This definition is consistent with the definition of operator valued functions in the holomorphic functional calculus
%\begin{equation*}
%f(L(a))=\frac{1}{2\pi i}\oint_C f(\zeta)(\zeta-L(a))^{-1}\, \mathrm d \zeta\, ,
%\end{equation*}
%in the sense that $f(L(a))=L(f_{\star}(a))$.
%\end{remark}
%\end{quotation}

\begin{quotation}
\begin{remark}
If the curve $C$ in definition~\ref{d:starf} can be chosen to be a circle centered at $0$, the star-function can also be defined through its series expansion around $0$
\begin{equation*}
f_\star(a)=\sum_{n=0}^\infty \frac{f^{(n)}(0)}{n!}a^{n_\star}
\end{equation*}
where $a^{n_\star}$ can be defined resursively as $a^{n_\star}=a\star a^{(n-1)_\star}$, $a^{0_\star}=1$.
\end{remark}
\end{quotation}

\phantom{a}

\begin{definition}[Star operators] We call star-Toeplitz the operator represented by
the semi-infinite matrix $T(a)$ with elements
\begin{equation}
    T_{j,k}(a)= L_{j,k}[a], \quad j,k=1,2,3,\ldots
\end{equation}
We call star-Hankel the operator represented by
the semi-infinite matrix $H(a)$ with elements
\begin{equation}
    H_{j,k}(a)=\frac{1}{2\pi}\int_{-\pi}^\pi a_{\frac{j-k+1}{2}}(e^{ip})e^{-ip(j+k-1)} dp \; , \quad j,k=1,2,3,\ldots \label{locally Hankel}
\end{equation}
By extension, we will refer to $L(a)$ as the star-Laurent operator induced by $a$.
\end{definition}

\begin{definition}[Star-Toeplitz matrix]
The $n$-by-$n$ star-Toeplitz matrix $T_n(a)$ is the finite section~\eqref{eq:Tstar} of the star-Toeplitz operator $T(a)$. 
\end{definition}

\begin{definition}[Reflection]
We define the reflected symbol $a^{\overset{n}{\sim}}$ as follows
\begin{equation}\label{eq:tilde}
(a^{\overset{n}{\sim}})_x(z)=a_{n+1-x}(z^{-1})\, .
\end{equation}
and, for the sake of compactness, we omit the number when $n=0$ : $a^\sim\equiv a^{\overset{0}{\sim}}$. 
\end{definition}

\begin{definition}[$\pm$ symbols]
We say that a symbol $a$ is a $+$ symbol if $m<0\Rightarrow (a_{\frac{n}{2}})_{2m+\frac{1-(-1)^n}{2}}=0$ for all $n\in\mathbb Z$. Similarly, we say that $a$ is a $-$ symbol if $m>0\Rightarrow (a_{\frac{n}{2}})_{2m-\frac{1-(-1)^n}{2}}=0$  for all $n\in\mathbb Z$.
Equivalently, $a$ is a $+$ or $-$ symbol if $L(a)$ is lower or upper triangular, respectively. We will typically add a subscript $+$ (or $-$) to a symbol if it is a $+$ (or $-$) symbol. 
\end{definition}

\begin{definition}[$\pm$ functions] If $g:\mathbb R\times (\mathbb R/(2\pi\mathbb Z)) $ is an integrable function of the second argument at any fixed value of the first argument, we define the $\pm$ functions $g^{\pm}:\mathbb R\times (\mathbb R/(2\pi\mathbb Z))$ as follows:
\begin{equation}
\begin{aligned}
g^+(x,p)&=\sum_{n=0}^\infty \frac{1}{2\pi}\int_{-\pi}^{\pi} g(x,q) e^{i n (p-q)} d q\\
g^-(x,p)&=\sum_{n=-\infty}^{-1} \frac{1}{2\pi}\int_{-\pi}^{\pi} g(x,q) e^{i n (p-q)} d q\, .
\end{aligned}
\end{equation}
We also use $\tilde g$ to indicate $g^+-g^-$.
\end{definition}

\begin{quotation}
\begin{remark}
A $+$ symbol can be represented by a $+$ function, but it's not always possible to represent a $-$ symbol by a $-$ function (the diagonal of the star-Laurent matrix induced by a $-$ symbol is generally nonzero).
\end{remark}
\end{quotation}

\phantom{a}

Finite matrices whose leading principal submatrices are all invertible admit an $LU$ decomposition, where $L$ and $U$ are lower and upper triangular matrices, respectively. An analogous concept at the operator level is captured by the following Wiener–Hopf factorization.

\phantom{a}

\begin{definition}[Wiener-Hopf star factorization with zero winding number]\label{a:wnumber}
We say that a symbol $a\in\anb$ has zero star winding number if it can be decomposed as follows
$$
a=a_{+}^L\star a_-^L=a_{-}^R\star a_+^R\, ,
$$
where $a_{\pm}^{L/R}\in\anb$ are $\pm$ symbols for which $\log_\star a_{\pm}^{L/R}$ exists and belongs to $\anb$. We call this decomposition Wiener-Hopf star factorization.

\end{definition}

\begin{quotation}
\begin{remark}\label{remark uniqueness Wiener Hopf star factorization}
The Wiener-Hopf star factorization is not unique. It becomes unique (in the quotient space of symbols) if we adopt e.g. the convention $(a_{-,x}^{R/L})_0=1$ for every $x\in\mathbb{Z}$, where $(a_{\pm,x}^{R/L})_0$ is a shorthand for $((a_{\pm }^{R/L})_x)_0$.
\end{remark}
\end{quotation}

\begin{quotation}
\begin{remark}
When $a=[f]$ and $f:\mathbb R\times (\mathbb R/(2\pi\mathbb Z))\rightarrow
\mathbb C$ can be chosen to be independent of its first argument (i.e., when the star Laurent matrix is a Laurent matrix), this definition reduces to the requirement that the symbol has zero winding number. The Wiener-Hopf star factorization is then the standard Wiener-Hopf factorization, an example of which is provided by $a_\pm^{L/R}=[e^{(\log f)^{\pm}}]$. 

\end{remark}
\end{quotation}
\begin{quotation}
\begin{remark}
The fact that  $\log_\star a_{\pm}^{L/R}$ exist and belong to $V$ implies that also $(a_{\pm}^{L/R})^{-1_\star}=\exp_{\star}(-\log_\star a_{\pm}^{L/R})$ exist and belong to $V$: the exponential is an entire function and the holomorphic functional calculus provides an explicit expression for the star exponential in the annulus in which $\log_\star a_{\pm}^{L/R}$  are bounded.
\end{remark}
\end{quotation}

\phantom{a}

We refer the reader to Examples~\ref{ex:tridiagonal} and \ref{example nonzero winding number} for explicit Wiener-Hopf star factorizations of symbols generating tridiagonal matrices.
Example~\eqref{example nonzero winding number}, in particular, points out that the star winding number of a symbol $a$ is generally different from the  winding number of the  symbol $a_x:\mathbb T\rightarrow \mathbb C$ of the Laurent matrices $L(a_x)$ for given $x\in\frac{1}{2}\mathbb Z$.

\phantom{a}

\begin{definition}
For $n\in\mathbb{Z}$ we use the self-explanatory notation $z^n$ for the symbol represented by the function that maps $(x,z)\mapsto z^n$. 
\end{definition}

\phantom{a}

The following Borodin-Okounkov-Case–Geronimo (BOCG) type formula holds: 
\begin{proposition}\label{t:BOCG}
For a symbol $a\in\anb$ with zero star winding number the operator \mbox{$T(a_-^L\star(a_+^R)^{-1_\star})T((a_-^R)^{-1_\star}\star a_+^L)$} is invertible and  the determinant of the star-Toeplitz matrix $T_n(a)$ satisfies the identity
\begin{equation}\label{eq:BOCG}
\det T_n(a)=\tfrac{\det\left(I- H\left(z^{-n}\star a_-^L\star(a_+^R)^{-1_\star}\right)H\left(((a_-^R)^{-1_{\star}} \star a_+^L)^{\sim}\star z^{-n}\right)\right)}{\det(T(a_-^L\star(a_+^R)^{-1_\star})T((a_-^R)^{-1_\star}\star a_+^L))}
\exp\left(\sum_{j=1}^n\log\Bigl(\big( a_{+,j}^L\big)_0 \big(a_{-,j}^L\big)_0\Bigr)\right) \, .
\end{equation}
\end{proposition}

\begin{quotation}
\begin{remark}
A similar expression is known for block-Toeplitz matrices, which suffer from the same lack of commutativity of the factors of the Wiener-Hopf factorisation---see, e.g., Ref.~\cite{Basor2000On}. 
\end{remark}
\end{quotation}

\begin{quotation}
\begin{remark}\label{rem:boundHH}
If $a\in\anb[\rho]$, the numerator in the formula satisfies the bound
\begin{equation*}
    \det\left(I- H\left(z^{-n}\star a_-^L\star(a_+^R)^{-1_\star}\right)H\left(((a_-^R)^{-1_{\star}} \star a_+^L)^{\sim}\star z^{-n}\right)\right)=1+O(\rho_1^{2n})
\end{equation*}
for any $\rho_1\in(\rho,1)$---Lemma~\ref{l:boundHH}.
\end{remark}
\end{quotation}

\phantom{a}

A weak Szeg\"o limit type theorem readily follows: 
\begin{corollary}\label{c:weakS}
From \eqref{eq:BOCG} and Remark~\ref{rem:boundHH} it follows
\begin{equation}\label{eq:weakS}
\lim_{n\rightarrow\infty} \frac{\det T_{n}(a)}{\det T_{n-1}(a)(a_{+,n}^L)_0(a_{-,n}^L)_0}=1\, .
\end{equation}
\end{corollary}
\begin{quotation}
\begin{remark}\label{c:leftincrease}
There is an intuitive reason why this asymptotic behaviour is expressed only in terms of the factors of the left factorization: the formula compares the determinant of a star-Toeplitz matrix with the one of the submatrix obtained by removing the last row and column. A different result is obtained by comparing it with the submatrix in which the first row and column are removed. Specifically, 
since $\det T_{n}(a^{\overset{n}{\sim}})=\det T_n(a)$ and $\det T_{n-1}( a^{\overset{n}{\sim}})=\det T_{n-1}(\texttt T^{1}a)$, where $\texttt T$ is the shift operator $\texttt T^t: a_x(z)\mapsto a_{x+t}(z)$, Corollary~\ref{c:weakS} can be alternatively expressed as
\begin{equation}
\lim_{n\rightarrow\infty} \frac{\det T_{n}(a)}{\det T_{n-1}(\texttt T^1 a)(a_{-,n}^R)_0(a_{+,n}^R)_0}=1\, .
\end{equation}
Clearly, this comment applies also to the full formula~\eqref{eq:BOCG}.

\end{remark}
\end{quotation}

\phantom{a}

The denominator in \eqref{eq:BOCG} captures the leading correction to \eqref{eq:weakS} and is the main character of the star generalization of the strong Szeg\"o limit theorem. In the theory of Toeplitz and Hankel operators, such contribution is usually reduced to the calculation of the trace of a commutator of Toeplitz operators. In the generalisation that we are considering, a similar approach would involve additional assumptions to ensure that some operators are trace class\footnote{We have followed that route in a preliminary version of this work in which we used the generalisation of the  Helton-Howe-Pincus formula \cite{Helton1973Integral} worked out by Ehrhardt in Ref.~\cite{Ehrhardt2003A}}. Star-Laurent matrices, however, live in a larger space than Laurent ones, and this allows us to ease some of the complications underlying operator algebra. This is where we depart somewhat from the standard Toeplitz formalism. 

For stating the result, we need two additional definitions. 
\begin{definition}\label{d:PhiBCH}
Given two symbols $c,d\in V$,   $\Phi_{BCH}(c,d)$ stands for
\begin{multline}\label{eq:BCH}
\Phi_{BCH}(c,d)=\left(\int_0^1\int_0^1 e^{-i s \mathcal M(c)}\psi_1(e^{i \mathcal M(c)}e^{i t \mathcal M(d)})\, ds \, dt\right) d\\
-\left(\int_0^1\int_0^t e^{i s \mathcal M(c)}e^{i \mathcal M(d)}\psi_1(e^{-i \mathcal M(d)}e^{-i t \mathcal M(c)})\, ds \, dt\right) c\, ,
\end{multline}
where 
$
\psi_1(x)=2x\frac{x\log x+1-x}{(x-1)^2}
$ and $\mathcal M(c) d=\{c,d\}_M= -i(c\star d-d\star c)$. 
\end{definition}
\begin{quotation}
\begin{remark}
If $b_\pm\in \anb$ are $\pm$ symbols and $a=\exp_\star(b_-)\star \exp_\star(b_+)$,  the Baker-Campbell-Hausdorff formula for $\log_\star a$ gives---Lemma~\ref{lemma BCH}
\begin{equation}\label{eq:BCHeq}
\log_\star a-b_--b_+=\frac{i}{4}\{b_-,\Phi_{BCH}(b_-,b_+)\}_M-\frac{i}{4}\{\Phi_{BCH}(-b_+,-b_-),b_+\}_M\, .
\end{equation}
\end{remark}
\end{quotation}
\begin{quotation}
\begin{remark}
The first terms of the expansion of  $\nu^{-1}\Phi_{BCH}(\nu c,\nu d)$ as $\nu\rightarrow 0$ read
\begin{equation}\label{eq:Phi}
\nu^{-1}\Phi_{BCH}(\nu c,\nu d)= d+\frac{1}{3}i\nu\{c,d\}_M+\frac{1}{12}\nu^2\{\{c,d\}_M,d\}_M+\dots
\end{equation}
\end{remark}
\end{quotation}

\phantom{a}

\begin{definition}
We denote by $S_\Phi$ the subset of $V\times V$ for which, if $(c,d)\in S_\Phi$ then the power series in $\nu$ of $\nu^{-1}\Phi_{BCH}(\nu c,\nu d)$ is absolutely convergent to a symbol in $V$ (with respect to the metric induced by the norm $\|L(\cdot) \|$) in an open interval including $\nu=1$.  
\end{definition}
\begin{quotation}
\begin{remark}
Toeplitz symbols $c,d\in V$ trivially satisfy $(c,d)\in S_\Phi$ due to commutation.
\end{remark}
\end{quotation}
\begin{quotation}
\begin{remark}
If there is $\rho\in(0,1)$ such that  $(M(c)+M(d))\frac{1+\rho}{1-\rho}< \log 2$, with $M:a\mapsto\sup\{{|a_x(z)|}: \ x\in\frac{1}{2}\mathbb{Z}, {\rho\leq}|z|\leq \rho^{{-1}} \} $, then $(c,d)\in S_\Phi$---Remarks~\ref{remark convergence BCH} and~\ref{remark bound L sup annulus}.
%If $c,d\in V$ have a sufficiently small bound on the annulus of analyticity then $(c,d)\in S_\Phi$ (it follows from remarks~\ref{remark convergence BCH} and~\ref{remark bound L sup annulus}). 
\end{remark}
\end{quotation}

\phantom{a}

\begin{theorem}\label{c:strong}
Let $a\in\anb$ be a symbol with zero star winding number %\textcolor{red}{, for which regularization exists (see remark~\ref{remark regularisation existence}),} 
and $b_\pm^{L/R}=\log_\star a_\pm^{L/R}$ be such that $(b_+^L, b_-^L),(b_-^R, b_+^R)\in S_\Phi$. If the symbol $(x,z)\mapsto a_x(z)$ behaves nicely as $x\rightarrow\pm \infty$, i.e., it can be regularized as in Definition~\ref{definition regularisation}, then the following formula holds
\begin{equation}\label{eq:strongSzego}
\log \det T_n(a)= \sum_{j=1}^n ((\log_\star  a)_j)_0+\frac{E^R_{\frac{1}{2}}+E^L_{n+\frac{1}{2}}}{2}+O(\rho^{n})
\end{equation}
for some $0<\rho<1$, where
\begin{equation}\label{eq:ERL}
E^{R/L}_x=
\frac{1}{2}\sum_{m=1}^\infty \sum_{j=1}^m ( b^{R/L}_{+,x+j-\frac{m+1}{2}})_m(d^{R/L}_{(-),{x+j-\frac{m+1}{2}}})_{-m}+(d^{R/L}_{(+),x+j-\frac{m+1}{2}})_m( b^{R/L}_{-,{x+j-\frac{m+1}{2}}})_{-m}
\end{equation}
\begin{equation}\label{eq:dPhiBCH}
\begin{aligned}
d^L_{(-)}=&\Phi_{BCH}(b_+^L,b_-^L)&d^R_{(+)}=&\Phi_{BCH}(b_-^R,b_+^R)\\
d^L_{(+)}=&-\Phi_{BCH}(-b_-^L,-b_+^L)&
d^R_{(-)}=&-\Phi_{BCH}(-b_+^R,-b_-^R)\, .
\end{aligned}
\end{equation}
\end{theorem}

\begin{quotation}

\begin{remark}\label{remark regularisation existence}
We have not undertaken a systematic investigation of the general conditions under which a symbol  can be regularized. However, since we have not been able to conceive any counterexample, we suspect that any symbol in $V$ with zero star winding number could be regularized. 
%if  the limit $x\rightarrow \infty$ exists in $(x,z)\mapsto a_x(z)$, it is clear that a necessary condition is for the resulting Toeplitz symbol $z\mapsto \lim_{x\rightarrow\infty} a_x(z)$ to have zero winding number---see also Example~\ref{example nonzero winding number}.
An example of regularization is provided in Remark~\ref{remark regularisation tridiagonal} for the symbol examined in Example~\ref{ex:tridiagonal}.
\end{remark}
%\begin{remark}
%\textcolor{red}{The meaning of regularisation of a symbol is given in definition~\ref{definition regularisation}. In remark~\ref{remark regularisation tridiagonal} we construct it for a symbol examined in example~\ref{ex:tridiagonal}. Existence of regularisation is a technical condition, whose proof for general symbols goes beyond the scope of this work.}
%\end{remark}
\end{quotation}

\begin{quotation}
\begin{remark}
For a symbol inducing a Laurent operator the formula is reduced to \eqref{eq:strongS} and, in particular, $E^R_x=E^L_x=E(a)$.  
\end{remark}
\end{quotation}

\begin{quotation}
\begin{remark}
Since $b_\pm^{R/L},d_\pm^{R/L}\in V$, the dependence of $E_x$ 
on  $a_y(z)$ is exponentially suppressed with $|y-x|$, making it clear that $E^R_{\frac{1}{2}}$
and $E^L_{n+\frac{1}{2}}$
are determined by the matrix elements close to the upper-left and bottom-right corners of the matrix.
\end{remark}
\end{quotation}

%\phantom{a} 

\subsection{Asymptotic expansion}\label{ss:asymp}
Similarly to the standard block-Toeplitz case, the generality of Theorem~\ref{c:strong} is paid by the difficulty in computing the Wiener-Hopf star factorizations. In addition, even in cases in which the star factorization is known---cf.~Example~\ref{ex:tridiagonal}---the expression for $E_x^{E/L}$ in \eqref{eq:ERL} relies on the computation of $\Phi_{BCH}$, which is expected to pose significant challenges. 
%This problem is actually manageable 
These problems, however, become manageable when the symbol $a=[f]$ can be represented by a function $(x,p)\mapsto f(x,p)$ whose typical scale of variation with respect to $x$ becomes large.
Equation~\eqref{eq:strongSzego} shows that, ignoring  exponentially small corrections $O(
\rho^{n})$, $\log\det T_n(a)$ is determined by two contributions: corner terms, which we denote by $C_{ul}$ and $C_{br}$,  do not change when the matrix is perturbed in some region far from the upper-left or  bottom-right corner of the matrix, respectively; the remaining contributions are bulk terms, which we denote by~$D$
\begin{equation}\label{eq:Ansatz_st}
\log\det T_n(a)=
D+C_{ul}+C_{br}
+O(\rho^{n})\, .
\end{equation}
Our ultimate goal is to  express   $C_{up}$, $C_{br}$ and $D$ as functionals of  $g=\log f$ and its derivatives as follows
\begin{equation}\label{eq:Ansatz_st1}
\begin{gathered}
C_{ul}=\frac{1}{2\pi}\int_{-\pi}^\pi \mathcal C_-\Bigl(g(\tfrac{1}{2},p),\partial g(\tfrac{1}{2},p),\dots\Bigr)dp\qquad
C_{br}=\frac{1}{2\pi}\int_{-\pi}^\pi \mathcal C_+\Bigl(g(n+\tfrac{1}{2},p),\partial g(n+\tfrac{1}{2},p),\dots\Bigr)dp\\
D[g]=\frac{1}{2\pi}\int_{\frac{1}{2}}^{n+\frac{1}{2}}\int_{-\pi}^\pi \mathcal D\Bigl(g(x,p),\partial g(x,p),\dots\Bigr) dp dx\, .
\end{gathered}
\end{equation}

\begin{quotation}
\begin{remark}
In the definition of $\mathcal D$ and $\mathcal C_{\pm}$ there are redundant (gauge) degrees of freedom. On the one hand, they are defined up to derivatives with respect to $p$; on the other hand, \eqref{eq:Ansatz_st} is invariant under the following gauge transformation
\begin{equation*}
\begin{aligned}
\mathcal D(g,\partial g,\dots)\rightarrow &\mathcal D(g,\partial g,\dots)+\partial_x \mathcal G(g,\partial g,\dots)\\
\mathcal C_\pm(g,\partial g,\dots)\rightarrow &\mathcal C_\pm(g,\partial g,\dots)\mp \mathcal G(g,\partial g,\dots)
\end{aligned}
\end{equation*}
for any function $\mathcal G$ of $g$ and derivatives.
The latter transformation, however, affects the values of $D$, $C_{ul}$, and $C_{br}$ themselves, as $D\rightarrow D+(\mathcal G|_{x=n+\frac{1}{2}})_0-(\mathcal G|_{x=\frac{1}{2}})_0$, $C_{br}\rightarrow C_{br}-(\mathcal G|_{x=n+\frac{1}{2}})_0$ and $C_{br}\rightarrow C_{ul}+(\mathcal G|_{x=\frac{1}{2}})_0$.
In the following we fix the gauge by expressing $D[g]$ as follows:
\begin{equation}\label{eq:gauge}
 D[g]=\int_{-\infty}^0\delta_g  D[e^s g](g)d s\, ,
\end{equation}
where $\delta_g  D[e^s g]:\phi\mapsto \lim_{m\rightarrow\infty}\lim_{\epsilon\rightarrow 0}\frac{D[e^s (g+\epsilon\phi_m)]-D[e^s g]}{\epsilon}$ for any sequence  $\{\phi_m\}_{\mathbb N}$ of functions $\phi_m:\mathbb R\times (\mathbb R/(2\pi\mathbb Z))$ of class $C^\infty$ converging to $\phi$ almost everywhere in $(\frac{1}{2},n+\frac{1}{2})$ but vanishing at both $\frac{1}{2}$ and $n+\frac{1}{2}$ together with all their derivatives---see Lemma~\ref{l:gauge}. This induces a natural representation of $\mathcal D$---cf.~\eqref{eq:Gproof}.
\end{remark}
\end{quotation}

\phantom{a}

We approach this programme within the framework of a perturbation theory developed under the assumption that the matrix elements  vary  over long scales, $\nu^{-1}$, along the diagonals. For the sake of simplicity, we focus on representations  of the symbol that are infinitely differentiable in $(\frac{1}{2},n+\frac{1}{2})\times (\mathbb R/(2\pi\mathbb Z))$ and  express the result in the form of asymptotic series
\begin{equation}\label{eq:expansionLB}
\mathcal D\sim \sum_{j=0}^\infty\mathcal D^{(j)}\qquad \mathcal C_\pm\sim \sum_{j=0}^\infty \mathcal C_\pm^{(j)}\, 
\end{equation}
 as $\nu\rightarrow 0$, where $\mathcal D^{(j)}$ and $\mathcal C^{(j)}_\pm$ represent $O(\nu^j)$ contributions. 
Importantly, in this expansion $n$ plays the role of an independent parameter, so the asymptotics of the determinant still depends on the behaviour of the function representing the symbol in $(\frac{1}{2},n+\frac{1}{2})\times (\mathbb R/(2\pi \mathbb Z))$.

\phantom{a}

Here we report some general properties of the expansion as well as the formulas for the first terms.

\begin{theorem}\label{t:main}
Let $a=[f_\nu]$ for $\nu>0$, where $f_\nu$ is defined by $f_\nu(x,p)=e^{g(x\nu,p)}$ and $g$ be a nice enough function on $\mathbb R\times \mathbb R/(2\pi \mathbb Z)$ that allows for the asymptotic expansions~\eqref{eq:expansionLB}---see Remark~\ref{remark assumptions}.
% assumptions in Section~\ref{s:locally} for some sufficient conditions).
Then the following properties can be satisfied
\begin{enumerate}
\item $\mathcal D^{(2j-1)}(g,\partial g,\dots)=0\qquad \forall j\in \mathbb N$
\item $\mathcal D^{(2j)}(g,\partial g,\dots)$ is a multivariate polynomial of total degree $3j$ with respect to $\{\partial_1^{m_1}\partial_2^{m_2} g\}_{m_1,m_2\in\{0,1,\dots,2j\}}$ and consists of monomials with factors including always a single $g$ (i.e., a single factor with $m_1=m_2=0$), in which the derivative with respect to either of the arguments appears exactly $2j$ times.
\item $\mathcal C_\pm^{(j)}(g,\partial g,\dots)$ are multivariate polynomials with respect to $\{\partial_1^{m_1}\partial_2^{m_2} g,\partial_1^{m_1}\partial_2^{m_2} \tilde g\}_{m_1\in\{0,1,\dots,j\}\atop m_2\in\{0,1,\dots,j+1\}}$ (where $\tilde g=g^+-g^-$) consisting of monomials in which the derivative with respect to the first argument appears exactly $j$ times.
\end{enumerate}
In the gauge \eqref{eq:gauge}, the first bulk contributions read
\begin{equation}\label{eq:asympt2L}
\begin{aligned}
\mathcal D^{(0)}(g,\partial g,\dots)=&g\\
\mathcal D^{(2)}(g,\partial g,\dots)=&-\frac{1}{12}g \det {\bf H}_g
\end{aligned}
\end{equation}
where ${\bf H}_g$ is the  Hessian $({\bf H}_g)_{ij}=\partial_{i,j} g$, with $i,j\in\{1,2\}$.
The first corner contributions in the same gauge can be defined as 
\begin{equation}\label{eq:asympt2B}
\begin{aligned}
\mathcal C^{(0)}_{\pm }(g,\partial g,
\dots)=&-\frac{1}{4}  g i \partial_2  \tilde g\\
%\mathcal C^{(1)}_{\pm }(g,\partial g,\dots)=&\pm \frac{1}{8}\left(\partial_1 \tilde g\partial_2 g  \partial_2 \tilde g+\partial_1 g\left(-\frac{1}{3}+\partial_{2,2}g+\frac{2}{3} g \partial_{2,2}g+\frac{(\partial_2 g)^2-(\partial_2 \tilde g)^2}{2}\right)\right)
\mathcal C^{(1)}_{\pm }(g,\partial g,
\dots)=&\pm \frac{1}{24}\left(2\partial_1 \tilde g\partial_2 g  \partial_2 \tilde g+\partial_1 g\left(-1+3\partial_{2,2}g+2 g \partial_{2,2}g+(\partial_2 g)^2-2(\partial_2 \tilde g)^2\right)\right)
\end{aligned}
\end{equation}
\end{theorem}

\begin{quotation}
\begin{remark}\label{remark assumptions}
The assumptions stated in Proposition~\ref{p:asymplogstar} provide sufficient conditions for Theorem~\ref{t:main}. However, since we also assume that the symbol can be regularized (as per Definition~\ref{definition regularisation}), we have chosen a non-rigorous presentation with the informal terminology of a ``nice enough function''. 
%\textcolor{red}{A sufficient condition for the result to hold are the assumptions in Proposition~\ref{p:asymplogstar}, assuming also the existence of the regularisation for the symbols in question. Since some of the conditions involved are technical conditions that could be difficult to check in practice, and we have not tried to establish the most general conditions under which our results holds, in view of applications we have opted for a non-rigorous presentation and the terminology of a nice enough function. }
\end{remark}
\end{quotation}

\begin{quotation}
\begin{remark}
If the first $2k$ derivatives of $g$ are bounded in $\mathbb R\times (\mathbb R/(2\pi \mathbb Z))$, then the error in $\log\det T_n(a)$ after truncating the asymptotic series at order $2k-2$ is $O(n\nu^{2k}+\nu^{2k-1})$ as $\nu\rightarrow 0$.
\end{remark}
\end{quotation}

\phantom{a}

In Ref.~\cite{Tilli1998Locally}, Tilli defined locally Toeplitz sequences with respect to a pair of functions $(\bar a,\bar f)$, where $\bar a$ has the intuitive meaning of a weight function whereas $\bar f$ is a generating function (i.e., a symbol) in the usual Toeplitz sense. In our language, that corresponds to considering a symbol $a^{(n)}=[f_n]$ that explicitly depends on the size $n$ of the matrix  with the Ansatz $f_n(x,p)=\bar a(\tfrac{x}{n})\bar f(p)$.
Tilli also considered finite sums of locally Toeplitz sequences, which required little modifications in his results and have subsequently become full part of the theory under the name of ``generalised locally Toeplitz sequences''\cite{Capizzano2003Generalized}. 
Within the framework of our analysis, it is natural to borrow Tilli's terminology  for sequences of star-Toeplitz matrices with symbols $a^{(n)}=[f_n]$ that can be represented by functions $f_n(x,p)=f(\tfrac{x-\frac{1}{2}}{n},p)$.  

\phantom{a}

\begin{corollary}[Locally Toeplitz sequences]\label{t:ballistic}
Let us represent the symbol $a^{(n)}=[f_n]$ by a function $f_n$ such that $f_n(x,p)=\exp(g(\tfrac{x-\frac{1}{2}}{n},p))$ for some
nice enough function $g$ on $(0,1)\times \mathbb R/(2\pi \mathbb Z)$ that allows for approximating $\mathcal D$ by $\mathcal D^{(0)}$ with an error $O(\mathcal D^{(2)})$ and $\mathcal C_\pm$ by $\mathcal C_\pm^{(0)}$ with an error $O(\mathcal C_\pm^{(1)})$ in \eqref{eq:expansionLB}\footnote{For example, $g$ could be extended so that it is of class $C^2$ with respect to the first argument on some open interval $(a,b)$, with $a<0,b>1$ and its first two derivatives with respect to the first argument are analytic with respect to the second argument.}.
Then, the determinant of the  star-Toeplitz matrix $T_n(a^{(n)})$ satisfies
\begin{equation}\label{eq:locallyT}
\det T_n(a^{(n)})= \exp\left(\frac{E(c_0)+E(c_1)}{2}\right)\exp\left[\frac{n}{2\pi}\int_0^1\int_{-\pi}^\pi g(t,p) dp dt\right](1+O(n^{-1})) \quad \text{as }n\rightarrow\infty \, ,
\end{equation}
where $c_0(e^{ip})=\lim_{t\rightarrow 0^+}e^{g(t,p)}$, $c_1(e^{ip})=\lim_{t\rightarrow 1^-}e^{g(t,p)}$, and $E(a)$ is defined in \eqref{eq:strongS}.
\end{corollary}

\phantom{a}

This scaling limit is only sensitive to $\mathcal D^{(0)}$ and $\mathcal C^{(0)}_\pm$,
but alternative sequences characterised by a different scaling with respect to the matrix size can depend also on $\mathcal D^{(j)}$, with $j>0$. 
We distinguish the following sequences

\phantom{a}

\begin{definition}
We call ``locally-$q$ Toeplitz sequence'' with respect to a function $f$, the sequence $\{T_n(a^{(n)})\}_{n\in\mathbb{N}}$ of square matrices characterised by the symbol $a^{(n)}=[f_n]$, with $f_n(x,p)=f(\frac{x-\frac{1}{2}}{n^q},p)$.
\end{definition}

\phantom{a}

\begin{corollary}[Locally-$\frac{1}{2}$ Toeplitz sequences]\label{t:diffusive}
Let us represent the symbol $a^{(n)}=[f_n]$ by a function $f_n$ such that $f_n(x,p)=\exp(g(\frac{x-\frac{1}{2}}{\sqrt{n}},p))$ for some nice enough function $g$ on $\mathbb R_{>0}\times \mathbb R/(2\pi \mathbb Z)$ that allows for approximating $\mathcal D$ by $\mathcal D^{(0)}+\mathcal D^{(2)}$ with an error $O(\mathcal D^{(4)})$ and $\mathcal C_\pm$ by $\mathcal C_\pm^{(0)}$ with an error $O(\mathcal C_\pm^{(1)})$ in \eqref{eq:expansionLB}\footnote{For example, 
$g$ could be extended so that it is analytic in a strip surrounding $\mathbb R/(2\pi\mathbb Z)$ with respect to the second argument and of class $C^4$, with uniformly bounded derivatives in $(0,+\infty)$ with respect to the first argument and analytic with respect to the second.}.
The sequence of determinants of  $T_n(a^{(n)})$ has the following limit behaviour
\begin{multline}\label{eq:locally1/2}
\det T_n(a^{(n)})=\exp\left( \frac{E(c_0)+E(c_{\sqrt{n}})}{2}\right)\exp\left[-\frac{\int_0^{\sqrt{n}}\int_{-\pi}^\pi g(t,p)\det \mathbf H_g (t,p) dp dt}{24\pi\sqrt{n}}\right]\\
\times \exp\left[\frac{\sqrt{n}}{2\pi}\int_0^{\sqrt{n}}\int_{-\pi}^\pi g(t,p) dp dt\right](1+O(n^{-\frac{1}{2}})) \quad \text{as }n\rightarrow\infty \, ,
\end{multline}
where $c_0(e^{ip})=\lim_{t\rightarrow 0^+}e^{g(t,p)}$, $c_{\sqrt{n}}(e^{ip})=e^{g(\sqrt{n},p)}$ and $E(a)$ is defined in \eqref{eq:strongS}.
\end{corollary}

\subsection{On block matrices}\label{ss:blockToeplitz}
In Section~\ref{ss:asymp} we have reported some results aiming at simplifying the formal expressions obtained for the determinant of star-Toeplitz matrices when the matrix elements vary slowly along the diagonals. 
In some physical applications such a condition is not satisfied, but a weaker one is. Specifically, the matrix could be locally similar to a block Toeplitz matrix rather than to a Toeplitz one. In many respects, this generalisation requires little modification, in that the Moyal star product is already not commutative, and we have already dealt with most of the issues related to the lack of commutation. If the matrix has emergent  $\kappa\times \kappa$ blocks, the symbol $a$ can be defined as a $\kappa\times \kappa$ matrix in such a way that the star-Laurent operator is represented by the doubly-infinite matrix
\begin{equation}
    L_{\kappa \ell-1+j,\kappa n-1+k}(a)=\frac{1}{2\pi}\int_{-\pi}^\pi [a_{\frac{\ell+n}{2}}(e^{ip})]_{j k}e^{-ip(\ell-n)}dp \; , \quad \ell,n\in \mathbb{Z}\,, \quad j,k=1,\dots,\kappa\, .
\end{equation}
Such a redefinition provides an alternative representation of the  star product, in particular, $L(a\star b)=L(a) L(b)$. We can then retrace the proofs reported in Section~\ref{s:limit} bearing in mind that the symbols are matrices. 
For example, Corollary~\ref{c:weakS} is almost unchanged, as we have 
\begin{equation}
\lim_{n\rightarrow\infty} \frac{\det T_{n}(a)}{\det T_{n-1}(a)\det[(a_{+,n}^L)_0(a_{-,n}^L)_0]}=1\, ,
\end{equation}
where $T_n(a)$ is the $n\kappa\times n\kappa$ matrix given by $\left(T_n(a)\right)_{j,k}=L_{j,k}(a)$, for $j,k=1,2,\ldots,n\kappa$. Even the stronger Theorem~\ref{c:strong} can be readily adapted to the block case by adding an overall trace to the right hand side of \eqref{eq:strongSzego}, which is now a $\kappa$-by-$\kappa$ matrix.
Incidentally, this also implies the following result for ordinary block Toeplitz matrices.
\begin{proposition}[Determinant of block Toeplitz matrices]\label{p:blockT}
Let $T_n(a)$ be a block Toeplitz matrix and $a$ be a symbol with zero winding number, analytic in an annulus that includes the unit circle. Then
\begin{equation}\label{eq:predictionblockToeplitz}
\log \det T_n(a)=  n \mathrm{tr}[(\log  a)_0]
-\frac{i}{4}\mathrm{tr}\left[\frac{1}{2\pi}\int_{-\pi}^\pi \left[ \beta_R(p)\gamma_R'(p)+\beta'_L(p)\gamma_L(p)\right] dp\right]+O(\rho^{n})
\end{equation}
for some $0<\rho<1$,
 where
\begin{equation}
\begin{aligned}
\beta_L&:p\mapsto (\log a_+^L+\log a_-^L)(e^{ip})\\
\beta_R&:p\mapsto (\log a_+^R+\log a_-^R)(e^{ip})\\
\gamma_L&:p\mapsto[(\Phi_{BCH}(\log a_+^L,\log a_-^L))^-+(\Phi_{BCH}(-\log a_-^L,-\log a_+^L))^+](e^{ip})\\
\gamma_R&:p\mapsto[(\Phi_{BCH}(\log a_-^R,\log a_+^R))^++(\Phi_{BCH}(-\log a_+^R,-\log a_-^R))^-](e^{ip})\, ,
\end{aligned}
\end{equation}
and $\Phi_{BCH}$ is defined in \eqref{eq:BCH} with $\mathcal M$ the adjoint mapping $\mathcal M(c)d=-i[c,d]$.
\end{proposition}

\begin{quotation}
\begin{remark}
An implicit alternative expression for the boundary term can be derived from Theorem 4.1 of Ref.~\cite{Widom1974}. We have not explored the relationship between the two representations.
\end{remark}
\end{quotation}
%\textcolor{red}{
%It would be interesting to clarify the connection between this result and Theorem 4.1 of Ref.~\cite{Widom1974}, which gives a different formula for the constant term.}

\phantom{a}

The complications underlying the block structure are more evident when trying to generalise the results reported in Section~\ref{ss:asymp}. Although the truncation of the Moyal product and of the Wiener-Hopf star factorizations could be justified in similar ways, the lack of a systematic procedure for the factorisation of a block Toeplitz matrix undermines the usefulness of the expansion. The final result would therefore rely on the knowledge of the Wiener-Hopf factorisations for the block-Toeplitz symbols associated with the local structure of the star block Toeplitz matrix (even in  Proposition~\ref{p:blockT} the constant is expressed in terms of the  factors of both left and right Wiener-Hopf factorisations). 

\subsection{Examples}\label{ss:examples}
We collect here some examples. 
\paragraph{On the Wiener-Hopf star factorization.}

\begin{quotation}
\begin{example}\label{ex:tridiagonal} We consider a symbol $a\in V$ generating a tridiagonal matrix $L(a)$. We represent it as
% Let $L(a)$ be tridiagonal and its symbol be
$a=[f]$, with  $f:\mathbb R\times (\mathbb R/(2\pi\mathbb Z))\rightarrow
\mathbb C$ of the form
\begin{multline}\label{eq:f_tridiag}
f(x,p)=\tfrac{f_0(x-1)f_0(x)-f_{-1}(x-\frac{1}{2})f_{1}(x-\frac{1}{2})f_{-1}(x+\frac{1}{2})f_{1}(x+\frac{1}{2})}{f_0(x-1)-f_{-1}(x-\frac{1}{2})f_{1}(x-\frac{1}{2})}\\
+f_1(x)e^{ip}+f_{-1}(x)f_0(x-\tfrac{1}{2})\tfrac{f_0(x+\frac{1}{2})-f_{-1}(x+1)f_{1}(x+1)}{f_0(x-\frac{1}{2})-f_{-1}(x)f_{1}(x)}e^{-i p}
\end{multline}
where $f_{0,-1,1}:\mathbb R\rightarrow
\mathbb C$ are continuous bounded functions such that, 
%$\forall x\in\overline{\mathbb  R},\ f_0(x)\neq f_{-1}(x+\frac{1}{2})f_{1}(x+\frac{1}{2})$
$\inf_{x\in\mathbb R}|f_0(x)-f_{-1}(x+\frac{1}{2})f_{1}(x+\frac{1}{2})|>0$. For simplicity, we focus on cases in which the limits $f_{0,-1,1}(\pm\infty):=\lim_{x\rightarrow\pm \infty}f_{0,-1,1}(x)$ exist.
% $\exists \lim_{x\rightarrow\pm \infty}f_{0,-1,1}(x)=f_{0,-1,1}(\pm\infty)$.
If  the following conditions hold
\begin{equation}\label{step condition star winding}
\begin{aligned}
 |f_{-1}(-\infty)f_{-1}(\infty)|<&1\\
 |f_{1}(- \infty)f_{1}( \infty)|<& |f_0(- \infty)f_0( \infty)|\, ,
\end{aligned}
\end{equation}
then $a$ has zero star winding number and the left and right factorisations $a_\pm^{L/R}=[f_\pm^{L/R}]$, with $f_\pm^{L/R}:\mathbb R\times (\mathbb R/(2\pi\mathbb Z))\rightarrow\mathbb C$, read
\begin{equation}\label{eq:factors_tridiag}
\begin{aligned}
f_-^R:(x,p)\mapsto &1+f_{-1}(x) e^{-i p}\\
f_+^R:(x,p)\mapsto &f_0(x-1)\tfrac{f_0(x)-f_{-1}(x+\frac{1}{2})f_{1}(x+\frac{1}{2})}{f_0(x-1)-f_{-1}(x-\frac{1}{2})f_{1}(x-\frac{1}{2})}+f_1(x) e^{i p}\\
f_+^L:(x,p)\mapsto &f_0(x)+f_1(x) e^{i p}\\
f_-^L:(x,p)\mapsto &1+f_{-1}(x)
\tfrac{f_0(x+\frac{1}{2})-f_{-1}(x+1)f_{1}(x+1)}{f_0(x-\frac{1}{2})-f_{-1}(x)f_{1}(x)}
e^{-i p}
\end{aligned}
\end{equation}
%The star logarithms of the factors are discussed at the beginning of Section~\ref{s:limit}. %\textcolor{red}{The conditions~\eqref{step condition star winding} ensure that the star logarithms of components exist and are in $V$.}
\end{example}
\begin{proofline}
The factorisations \eqref{eq:factors_tridiag} can be readily proven using that, if $a=[(x,p)\mapsto f(x) e^{i n p}]$ and $b=[(x,p)\mapsto g(x) e^{i m p}]$, then $a\star b=[(x,p)\mapsto f(x+\frac{m}{2})g(x-\frac{n}{2})e^{i(n+m)p} ]$, which  follows from the definition of star product.
%The star inverse of a $+$ or $-$ symbol, $a_\pm$, with zero star-winding number is still a $+$ or $-$ symbol, respectively. Thus, $a_\pm^{-1_\star}$ can be formally computed by representing $a_\pm$ with its Fourier series and inverting the identity $a_\pm \star a_\pm^{-1_\star}=1$.

To identify  conditions under which  the star logarithms of the $\pm$ factors exist, are in $V$, and are $\pm$ symbols, we study generic symbols of the form $a_\pm=[(x,p)\mapsto\lambda_0(x)+\lambda_1(x)e^{\pm i p}]$ for some continuous functions $\lambda_0,\lambda_1:\mathbb{R}\to\mathbb{C}$. 
Their inverse $a_\pm^{-1_\star}$ can be formally computed by representing $a_\pm$ with their Fourier series and inverting the identity $a_\pm \star a_\pm^{-1_\star}=1$.
If $\inf_{x\in \frac{1}{2}\mathbb Z}|\lambda_0(x)|>0$, one readily finds
\begin{equation}
(a_\pm^{-1_\star})_x(e^{ip})=\sum_{n=0}^\infty (-1)^n \tfrac{\prod_{j=0}^{n-1}\lambda_1(x-\frac{n-1}{2}+j)}{\prod_{j=0}^{n}\lambda_0(x-\frac{n}{2}+j)}e^{\pm i n p}\, .
\end{equation}
The convergence of the series depends on the behaviour of $\lambda_1$ and $\lambda_0$ when their argument approaches $\pm \infty$. For the sake of simplicity, we assume that the limits $\lambda_{0,1}(\pm\infty):=\lim_{x\rightarrow\pm \infty}\lambda_{0,1}(x)$ exist. The series is absolutely convergent if
\begin{equation}\label{step ratio limits}
\Bigl|\frac{\lambda_1(-\infty)}{\lambda_0(-\infty)}\frac{\lambda_1(+\infty)}{\lambda_0(+\infty)}\Bigr|<1\, .
\end{equation}
% % if $\exists \lim_{x\rightarrow\pm \infty}h_{0,1}(x)=h_{0,1}(\pm\infty)$,then the
% series is absolutely convergent when
% \begin{equation}
% \Bigl|\frac{\lambda_1(-\infty)}{\lambda_0(-\infty)}\frac{\lambda_1(\infty)}{\lambda_0(\infty)}\Bigr|<1\, ,
% \end{equation}
% provided that $\lambda_0(\frac{1}{2}\mathbb Z)$ does not include $0$\footnote{Since the star-Laurent operator is only determined by the behaviour of the symbol for $x\in\frac{1}{2}\mathbb Z$, one could circumvent the presence of a zero at some $x\notin\frac{1}{2}\mathbb Z$ by choosing an alternative representation of the symbol.}.
% The series of the inverse of the resolvent converges for $z$ different from $\lambda_0(\frac{1}{2}\mathbb Z)$  and
% \begin{equation}
% \left|\frac{\lambda_1(\infty)}{z-\lambda_0(\infty)}\frac{\lambda_1(-\infty)}{z-\lambda_0(-\infty)}\right|<1\, ,
% \end{equation}
%For the sake of simplicity, let us assume $\exists\lim_{x\rightarrow\pm\infty} \lambda_{0,1}(x)=\lambda_{0,1}(\pm \infty)$. Then the region of convergence is the complement of the union of the circles centred at $\lambda_0(\pm \infty)$ with radius $|\lambda_1(\pm \infty)|$ and the image of $\lambda_0$.
The resolvent set of $L(a_\pm)$ consists of all $z\in\mathbb{C}$ for which the symbol $[z-\lambda_0(x)-\lambda_1(x)e^{\pm i p}]$ is star invertible, and the spectrum of $L(a_\pm)$ is its complement. From \eqref{step ratio limits} it follows that the series for the star inverse of $[z-\lambda_0(x)-\lambda_1(x)e^{\pm i p}]$ converges absolutely for all $z\in\mathbb{C}\backslash \lambda_0(\frac{1}{2}\mathbb Z)$ such that
\begin{equation}
\left|\frac{\lambda_1(+\infty)}{z-\lambda_0(+\infty)}\frac{\lambda_1(-\infty)}{z-\lambda_0(-\infty)}\right|<1\, .
\end{equation}
This is a connected region of the complex plane, therefore  
%This condition is equivalent to requiring that $|z-\lambda_0(+\infty)|\geq R$ and $|z-\lambda_0(-\infty)|>\tfrac{|\lambda_1(-\infty)\lambda_1(+\infty)|}{R}$ \textcolor{blue}{or $|z-\lambda_0(+\infty)|> R$ and $|z-\lambda_0(-\infty)|=\tfrac{|\lambda_1(-\infty)\lambda_1(+\infty)|}{R}$}, for some $R>0$. The latter describes a complement of two disks in the complex plane, one open disk of radius $R$ and one closed disk with radius $\frac{|\lambda(-\infty)\lambda(+\infty)|}{R}$. In particular, because of assumption~\eqref{step ratio limits}, there is $R>0$ such that $z=0$ does not belong to any of the two disks. Since the resolvent set includes the complement of these two disks, except for possibly a finite number of points in $\lambda_0(\frac{1}{2}\mathbb Z)$, 
there exists a path outside the spectrum connecting $0$---which belongs to the region by virtue of \eqref{step ratio limits}---to $\infty$. %Since the complementary region is bounded, the The latter implies the existence of the logarithm of $L(a_\pm)$, and hence of the star logarithm of $ a_\pm$.
Thus, the holomorphic  functional calculus provides the following representation of the star logarithm of  $ a_\pm$
% If, for given $x$, the series of the inverse is absolutely convergent and the spectrum does not separate $0$ from $\infty$,  the holomorphic functional calculus provides a representation of the the star-logarithm via a contour integral. The region of convergence can be thought of as the complementary region of the union of a family of regions, parametrized by a real positive variable $R$, consisting of the intersection of two circles centred  at $\lambda_0(\infty)$ and $\lambda_0(-\infty)$, one with radius $R$, and the other with radius $\frac{|\lambda_1(\infty)\lambda_1(-\infty)|}{R}$.  
% The absolute convergence of the series of the inverse is then a sufficient condition for the existence of a path outside the spectrum connecting $0$ to $\infty$.
%Indeed, the star-logarithm can be represented as follows
% \begin{equation}
% \log_\star[\lambda_0(x)+\lambda_1(x)e^{\pm ip}]=\sum_{n=0}^\infty e^{ \pm i n p}\oint_C\frac{dz}{2\pi i}\log z \tfrac{\prod_{j=\frac{1-n}{2}}^{\frac{n-1}{2}}\lambda_1(x+j)}{\prod_{j=-\frac{n}{2}}^{\frac{n}{2}}[z-\lambda_0(x+j)]}\, ,
% \end{equation}
\begin{equation}\label{eq:starlogexample}
(\log_\star a_\pm)_x(e^{ip})=\sum_{n=0}^\infty e^{ \pm i n p}\frac{(-1)^n}{2\pi i}\oint_C\log z \frac{\prod_{j=0}^{n-1}\lambda_1(x-\frac{n-1}{2}+j)}{\prod_{j=0}^{n}[z-\lambda_0(x-\frac{n}{2}+j)]} dz\, ,
\end{equation}
for whatever contour $C$ %\textcolor{red}{in the complement of the two disks, that does not} 
in the region of absolute convergence that surrounds the complementary---bounded---region and 
does not intersect the branch cut of the logarithm, which is identified with the path discussed above. 
The convergence of the series is not affected 
if we extend $p$ to the strip in $\mathbb C$ with $|\mathrm{Im}(p)|<-\log\rho$, for $\rho\in(0,1)$ such that the series of the inverse of $\lambda_0(x)+\frac{1}{\rho}\lambda_1(x)e^{\pm i p}$ is absolutely convergent.
% In conclusion, the star-logarithm of the symbol is in $V_\rho$ if  
% \begin{equation}
% \left|\frac{\lambda_1(\textcolor{red}{+}\infty)}{\lambda_0(\textcolor{red}{+}\infty)}\frac{\lambda_1(-\infty)}{\lambda_0(-\infty)}\right|<\rho\, 
% \end{equation}
Thus, a sufficient condition for the existence in $V_\rho$ of the star logarithm of the symbol $[(x,p)\mapsto \lambda_0(x)+\lambda_1(x)e^{\pm i p}]^{-1_\star}$ is 
\begin{equation}
\left|\frac{\lambda_1(+\infty)}{\lambda_0(+\infty)}\frac{\lambda_1(-\infty)}{\lambda_0(-\infty)}\right|<\rho\, 
\end{equation}
provided that the limits exist and $\inf_{x\in \frac{1}{2}\mathbb Z}|\lambda_0(x)|>0$. 
\end{proofline}
\end{quotation}

\begin{quotation}
\begin{example}\label{example nonzero winding number}
Let the symbol of the star Laurent matrix be $a=[f]$, with $f$ of the form \eqref{eq:f_tridiag} and
\begin{equation}
f_0:x\mapsto 1-\frac{2}{3\cosh(\omega x)}, \qquad f_{\pm 1}:x\mapsto \frac{1}{2},
\end{equation}
for some $\omega>0$. 
For $|x|\rightarrow\infty$, $a_x$ approaches a Toeplitz symbol with zero winding number, but for small enough $\omega$ the winding number of $a_x$ with $x\ll\frac{1}{\omega}$ is nonzero. 
As discussed in Example~\ref{ex:tridiagonal}, $a$ has zero star winding number. 
Remarkably, although $\log_\star (a^{R/L}_\pm)$ can be easily upper bound by an exponentially fast decay with a rate independent of $\omega$,  such a bound is accompanied by a prefactor that grows exponentially with $\frac{1}{\omega}$. Specifically, we have
\begin{equation}\label{eq:exampleinequality}
\left| \left((\log_\star a_+^L)_x   \right)_n \right| \leq  \tfrac{4}{\pi}\log\epsilon\left(\tfrac{3}{(1-3\epsilon)2\rho}\right)^{\frac{2}{\omega}\mathrm{arcosh}\left(\tfrac{2}{3(1-\epsilon-\frac{1}{2\rho})}\right)+2}\rho^{n+1} \; ,
\end{equation}
uniformly in $x\in\mathbb{Z}/2$ for any $\epsilon\in(0,\frac{1}{3})$ and $\rho\in(\frac{1}{2(1-\epsilon)},1)$. This is consistent with the expected pathological behaviour in the limit $\omega\rightarrow 0$, in which $a$ approaches the symbol of a Laurent matrix with nonzero winding number. 
\end{example}
%\begin{comment}
\begin{proofline}
Inspecting \eqref{eq:factors_tridiag}, one readily realizes that the star logarithms of the factors $a^{R/L}_-$ are not significantly affected by $\omega$. On the other hand, the asymptotic behaviours of the Fourier coefficients of $\log_\star a^{R/L}_+$ are more involved. We focus on $a^L_+$ and exhibit an explicit upper bound. 
Factor $a^L_+$ is given by $[f_+^L]$ with
\begin{equation}
f_+^L:(x,p)\mapsto 1-\frac{2}{3\cosh(\omega x)}+\frac{1}{2}e^{i p}\, .
\end{equation}
Using \eqref{eq:starlogexample}, we represent the star logarithm as
% \begin{equation}
% \log_\star f_+^L:(x,p)\mapsto \sum_{n=0}^\infty e^{ \pm i n p}\oint_C\frac{dz}{2\pi i}  \tfrac{2^{-n}\log (1+z)}{\prod_{j=-\frac{n}{2}}^{\frac{n}{2}}[z+\frac{2}{3\cosh(\omega(x+j))}]}\, ,
% \end{equation}
\begin{equation}
(\log_\star a_+^L)_x(e^{ip})=\sum_{n=0}^\infty e^{i n p}\frac{(-2)^n}{2\pi i}\oint_C\tfrac{\log (1+z)}{\prod_{j=0}^{n}[z+\frac{2}{3\cosh(\omega(x-n/2+j))}]}dz\, ,
\end{equation}
where $C$ is the rectangle in the complex plane with corners $-1+\epsilon\pm i ,1+\epsilon\pm i$, and  $\epsilon\in (0,\frac{1}{3})$. We can easily bound the absolute value of the Fourier coefficients of $(\log_\star a_+^L)_x$ from above. In particular  we have
\begin{equation}\label{eq:ex_bounds}
\begin{aligned}
|C|=&4\qquad (\text{the length of $C$})\\
|\log(1+z)|\leq & |\log \epsilon |\qquad z\in C\\
% \prod_{j=-\frac{n}{2}}^{\frac{n}{2}}|z+\tfrac{2}{3\cosh(\omega(x+j))}|\geq &(\tfrac{1}{3}-\epsilon)^{1+\frac{2}{\omega}\mathrm{arccosh}(\tfrac{2}{3(1-\epsilon-\frac{1}{2\rho})})}(2\rho)^{-n+\frac{2}{\omega}\mathrm{arccosh}(\tfrac{2}{3(1-\epsilon-\frac{1}{2\rho})})}\\
\prod_{j=0}^{n}|z+\tfrac{2}{3\cosh(\omega(x-n/2+j))}|\geq & (2\rho)^{-n-1}\left[2\rho(\tfrac{1}{3}-\epsilon)\right]^\mu\qquad z\in C\\
\mu= &\frac{2}{\omega}\mathrm{arcosh}(\tfrac{2}{3(1-\epsilon-\frac{1}{2\rho})})+2
\end{aligned}
\end{equation}
for any $\rho\in(\frac{1}{2(1-\epsilon)},1)$, where $\mu$ is an upper bound to the number of terms in the product on the left hand side of the third inequality for which $\kappa\equiv 1-\epsilon-\tfrac{2}{3\cosh(\omega (x-n/2+j))}\leq \frac{1}{2\rho}$.
Specifically, the third inequality of \eqref{eq:ex_bounds} is obtained using $|z|-\tfrac{2}{3\cosh(\omega (x-n/2+j))}\geq \tfrac{1}{3}-\varepsilon$ for the terms of the product with $\kappa\leq \frac{1}{2\rho}$ and  $|z|-\tfrac{2}{3\cosh(\omega (x-n/2+j))}\geq \tfrac{1}{2\rho}$ for the remaining ones. Since $2\rho(\tfrac{1}{3}-\epsilon)<1$, we finally obtain \eqref{eq:exampleinequality}.
Thus $\log_\star a_+^L \in V_\rho$ for any $\rho>1/2$.
% for any $\rho>\frac{1}{2(1-\epsilon)}$ and $n\geq \frac{2}{\omega}\mathrm{arccosh}(\tfrac{2}{3(1-\epsilon-\frac{1}{2\rho})})$. To obtain the last inequality we used that the number of terms in the product for which $\tfrac{2}{3\cosh(\omega (x+j))}\leq 1-\epsilon-\frac{1}{2\rho}$ is smaller than or equal to $1+\frac{2}{\omega}\mathrm{arccosh}(\tfrac{2}{3(1-\epsilon-\frac{1}{2\rho})})$. Finally, we obtain
% \begin{equation}
% \left|\oint_C\frac{dz}{2\pi i}  \tfrac{2^{-n}\log (1+z)}{\prod_{j=-\frac{n}{2}}^{\frac{n}{2}}[z+\frac{2}{3\cosh(\omega(x+j))}]}\right|\leq \tfrac{6\log\epsilon}{\pi(1-3\epsilon)}(\tfrac{3}{(2-6\epsilon)\rho})^{\frac{2}{\omega}\mathrm{arccosh}(\tfrac{2}{3(1-\epsilon-\frac{1}{2\rho})})}\rho^{n}
% \end{equation}
\end{proofline}
%\end{comment}
\end{quotation}

\paragraph{On the asymptotic expansion.}

\begin{quotation}
\begin{example}\label{e:1}
Let the symbol be $a=[f_\nu]$, with $f_\nu:(x,p)\mapsto e^{h(x\nu )-J(x\nu)\cos p}$, for some bounded functions $h,J:\mathbb{R}\to\mathbb{C}$ of class $C^4$. We then have
\begin{equation*}
g(x,p)=h(x)-J(x)\cos p\, ,\qquad\qquad
\tilde g(x,p)=h(x)-i J(x)\sin p\, ,
\end{equation*}
and hence
\begin{equation*}
\begin{aligned}
\frac{1}{2\pi}\int_{-\pi}^\pi \mathcal D(g(x,p),\partial g(x,p),\dots) dp=&h+\frac{\nu^2}{24}\left(h[J J''+(J')^2]+J^2 h''\right)\ \Big|_{x\nu} +O(\nu^4)\\
\frac{1}{2\pi}\int_{-\pi}^\pi \mathcal C_\pm(g(x,p),\partial g(x,p),\dots) dp=&\frac{J^2}{8}\pm \frac{\nu}{48}\left(h'(J^2-2)-(3+2h) J J'\right)\ \Big|_{x\nu}+O(\nu^2)\, .
\end{aligned}
\end{equation*}
\end{example}
\end{quotation}

\begin{quotation}
\begin{example}
Let us consider a locally-$\frac{1}{2}$ Toeplitz sequence characterised by a class of symbols of the form discussed in example~\ref{e:1}. Specifically, $a^{(n)}=[f_n]$, with $f_n(x,p)=\exp(g(\frac{x-\frac{1}{2}}{\sqrt{n}},p))$ and we choose 
\begin{equation*}
g(t,p)=1+\cos(2\omega t)-\cos(\omega t)\cos p\, .
\end{equation*}
The asymptotic behaviour can be readily computed from \eqref{eq:locally1/2}
\begin{equation*}
\log \det T_n(a^{(n)})=n+\frac{1}{2\omega}\sqrt{n}\sin(2\omega\sqrt{n})+\frac{3+\cos(2\omega\sqrt{n})-\omega^2}{16}+O(n^{-\frac{1}{2}})\, .
\end{equation*}
\end{example}
\end{quotation}

\paragraph{On block-Toeplitz matrices.}

\begin{quotation}
\begin{example}
Let us consider a block Toeplitz matrix $T_n(a)$ with a smooth $2$-by-$2$ symbol, $a$, close to the identity. Let us parametrize the left Wiener-Hopf factorisation  as $\log (a^L_\pm):e^{i p}\mapsto \lambda^\pm(p)+\mathbf n^\pm(p)\cdot\boldsymbol\sigma$, where $\boldsymbol \sigma$ denotes the vector of Pauli matrices. Prediction~\eqref{eq:predictionblockToeplitz} then gives
\begin{equation*}
\log \det T_n(a)=  \frac{n}{\pi}\int_{-\pi}^\pi \lambda(p) dp -\frac{i}{2\pi}\int_{-\pi}^\pi \left(\lambda\tilde\lambda'+\mathbf n\cdot \tilde{\mathbf n}'+i\mathbf n\cdot (\tilde{\mathbf n}'\times \tilde{\mathbf n})\right)(p)\, dp+O(\|a-\mathrm I\|^4_{\infty})
\end{equation*}
as $\|a-\mathrm I\|_\infty\rightarrow 0$, uniformly in $n$.
As an explicit example we set $\lambda=0$ and $\mathbf n=\epsilon \begin{pmatrix}\cos p&\gamma\sin p&h\end{pmatrix}$; we then find
\begin{equation*}
\log \det T_n(a)\sim  \frac{\epsilon^2}{2}(1+\gamma^2+2h \gamma \epsilon)+O(\epsilon^4)
\end{equation*}
as $\epsilon\rightarrow 0$, uniformly in $n$.
\end{example}
\end{quotation}

\section{Limit behavior of determinants}\label{s:limit}
In this section we prove the results reported in Section~\ref{ss:readapt}. We start by deriving a Borodin-Okounkov type formula for star-Toeplitz matrices.  
We highlight the analogy with the   theory of Toeplitz operators by following the same steps of standard derivations~\cite{Bottcher1999,Bottcher2006book,Bottcher2006}.

\subsection{A Borodin-Okounkov-Case-Geronimo type formula---proof of Proposition~\ref{t:BOCG}}

\begin{definition}[Projections]
We denote by $P$ the  orthogonal projection of $\ell^2(\mathbb Z)$ onto $\ell^2(\mathbb N)$, i.e.,  
if $f$ has components $f_j,\ j \in\mathbb{Z}$, $P$ acts as follows
\begin{equation}
    (Pf)_j=\left\{ \begin{array}{cc}
    f_j     &  j\geq 1\\
    0     &  j<1    \, .
    \end{array} \right.
\end{equation}
Analogously, we define $Q$ and $J$ in $\ell^2(\mathbb Z)$ as follows:
\begin{equation}
    (Qf)_j=\left\{ \begin{array}{cc}
    0     &  j\geq 1\\
    f_j     &  j<1
    \end{array} \right. 
\qquad \quad (Jf)_j=f_{-j+1}\, .
\end{equation}
We also denote by $P_n$ and $Q_n$ the projection on $\ell^2(\mathbb N)$ acting by the rules
\begin{equation}
    (P_nf)_j=\left\{ \begin{array}{cc}
    f_j     &  1\leq j\leq n\\
    0     &  \textrm{otherwise}
    \end{array} \right.\qquad (Q_nf)_j=\left\{ \begin{array}{cc}
    0     &  1\leq j\leq n\\
    f_j     &  \textrm{otherwise.}
    \end{array} \right.
\end{equation}
\end{definition}

\begin{quotation}\label{remark P Q J}
\begin{remark}
By definition, $P+Q=I$ and $J^2=I, P^2=P, Q^2=Q$. In addition, the  star-Toeplitz and star-Hankel operators are expressed in terms of the corresponding star-Laurent operators as follows
\begin{align}
&T(a)=PL(a)P\nonumber\\
    &H(a)=PL(a)QJ \qquad H(a^{\sim})=JQL(a)P\, .
\end{align}
The star-Toeplitz matrix $T_n$ is instead given by $P_nT(a)P_n$.
\end{remark}
\end{quotation}

\begin{lemma}\label{l:TTHH}
For $a,b\in\anb$ the following decomposition holds:
\begin{equation}
    T(a\star b)=T(a)T(b)+H(a)H(b^\sim) \label{toeplitz hankel}
\end{equation}
\end{lemma}
 \begin{proofline}
From Remark~\ref{remark P Q J} it follows
\begin{multline}
T(a\star b)=PL(a\star b)P=PL(a)L(b)P=PL(a)(P+Q)L(b)P=PL(a)PL(b)P+PL(a)QL(b)P\\
=PL(a)PPL(b)P+PL(a)QJJQL(b)P=T(a)T(b)+H(a)H(b^\sim)
\end{multline}
\end{proofline}

\begin{lemma}
    If $a,b\in\anb$ are $\pm$ symbols then $a\star b$ is also a $\pm$ symbol.
\end{lemma}
\begin{proof}
A symbol $a$ is of $+$ ($-$) type if and only if the matrix $L(a)$ is lower (upper) triangular. The claim follows from the fact that the product of lower (upper) triangular matrices is lower (upper) triangular.
\end{proof}

\begin{lemma} \label{facttoeplitzmp} If $a_-,c_+\in\anb$ are $-$ and $+$ symbols, respectively, then for any $b\in\anb$ we have
\begin{equation}\label{factorization toeplitz minus plus}
    T(a_-\star b \star c_+)=T(a_-)T(b)T(c_+)\, ,
\end{equation}
and hence
\begin{equation}
    T(a_-\star b)=T(a_-)T(b), \qquad T(b\star c_+)=T(b)T(c_+) .
\end{equation}
\end{lemma}
\begin{proofline}
The assumptions imply $H(a_-)=0$ and $H((c_+)^\sim)=0$. The result follows from applying formula (\ref{toeplitz hankel}) twice in succession.
\end{proofline}

\begin{lemma}\label{lemma invertibility factorization} If $a\in\anb$ has zero star winding number then
\begin{enumerate}
    \item  $T(a)=T(a_-^R)T(a_+^R)$ \vspace{0.1 cm}
    \item $T(a^{-1_\star})=T((a_-^L)^{-1_\star})T((a_+^L)^{-1_\star})$  \vspace{0.1 cm}
     \item $T(a_\pm^{L/R})$ are invertible with $T^{-1}(a_\pm^{L/R})=T((a_\pm^{L/R})^{-1_\star})$ \vspace{0.1 cm}
    \item $T(a)$ is invertible and $T^{-1}(a)=T((a_+^R)^{-1_\star})T((a_-^R)^{-1_\star})$
\end{enumerate}
\end{lemma}
\begin{proofline}
The first claim is a trivial consequence of Lemma~\ref{facttoeplitzmp}. Since $a^{-1_\star}=(a_-^L)^{-1_\star}\star (a_+^L)^{-1_\star}$, the second claim is a consequence of the first one. Since $a=a_-^R\star a_+^R$ we have $a^{-1_\star}=(a_+^R)^{-1_\star}\star (a_-^R)^{-1_\star}$. Lemma~\ref{facttoeplitzmp} then gives
\begin{align}
   & T((a_-^R)^{-1_\star})T(a_-^R)=T((a_-^R)^{-1_\star}\star a_-^R)=I=T(a_-^R\star (a_-^R)^{-1_\star})=T(a_-^R)T((a_-^R)^{-1_\star})\\
  &  T((a_+^R)^{-1_\star})T(a_+^R)=T((a_+^R)^{-1_\star}\star a_+^R)=I=T(a_+^R\star (a_+^R)^{-1_\star})=T(a_+^R)T((a_+^R)^{-1_\star}) \; .
\end{align}
Analogous relations hold for the factorization $a=a_+^L\star a_-^L$. The third claim readily follows. The fourth claim follows from the first and the third one.
\end{proofline}

\begin{lemma}\label{l:boundHH}
Let $b$ be in $V_\rho$. The following asymptotic behaviour holds
\begin{equation}
  \det\left(\mathrm I- H(z^{-n}\star b)H((b^{-1_\star})^\sim\star z^{-n}) \right)=1+O(\rho_1^{2n}) \; ,
\end{equation}
for any $\rho_1\in(\rho,1)$.
\end{lemma}
\begin{proofline}
The analyticity and uniform boundedness of $b$ in the annulus implies that, for any  $\rho_1\in(\rho,1)$, there is $M>0$  such that $|H_{j,k}(z^{-n}\star b)|=|(b_{\frac{j+n-k+1}{2}})_{j+n+k-1}|\leq M\rho_1^{j+k-1}\rho_1^n $ for every $j,k\in\mathbb{N}$ (see Lemma~\ref{lemma bound laurent coefficients}). Consequently, the Hilbert-Schmidt norm of $H (z^{-n}\star b)$ is bounded as $\|H (z^{-n}\star b)\|_2=O(\rho_1^n)$, and similarly $\|H((b^{-1_\star})^\sim\star z^{-n})\|_2=O(\rho_1^n)$. This implies the following bound for the trace norm of the product of star Hankel operators
\begin{align}
  \| H(z^{-n}\star b)H((b^{-1_\star})^\sim\star z^{-n}) \|_1\leq \|H(z^{-n}\star b)\|_2 \| H((b^{-1_\star})^\sim\star z^{-n})\|_2=O(\rho_1^{2n})
\end{align}
and hence the claim.
\end{proofline}

\begin{proof}[Proof of Proposition~\ref{t:BOCG}]
A proof of the Boroding-Okounkov type formula for star-Toeplitz matrices as stated in Proposition~\ref{t:BOCG}  follows the same steps as in the (block-)Toeplitz case. Consider a function $b\in\anb[\rho]$ with zero star winding number. Applying Lemma~\ref{l:TTHH} to the trivial symbol $b\star b^{-1_\star}$ gives ($b=b_-^R\star b_+^R=b_+^L\star b_-^L$)
\begin{eqnarray}
\mathrm I-H(b)H((b^{-1_\star})^\sim)&=&T(b)T(b^{-1_
\star})\nonumber\\
&=&T(b_-^R)T(b_+^R)T((b_-^L)^{-1_\star})T((b_+^L)^{-1_\star})\nonumber\\
&=&T(b_-^R)T^{-1}(b_-^L\star(b_+^R)^{-1_\star})T((b_+^L)^{-1_\star}) \ ,
\end{eqnarray}
where the last two equalities follow from Lemma~\ref{lemma invertibility factorization}. Hence
\begin{equation}
(\mathrm I-H(b)H((b^{-1_\star})^\sim))^{-1}=T(b_+^L)T(b_-^L\star(b_+^R)^{-1_\star})T((b_-^R)^{-1_\star})
\end{equation}
Let us then enforce the identity~\cite{Bottcher2006}
\begin{equation}\label{identity szego via jacobi}
\det\left( P_n(\mathrm I-K)^{-1}P_n\right)=\frac{\det\left(\mathrm I-Q_n K Q_n\right)}{\det\left(\mathrm I-K\right)}\, ,
\end{equation}
valid whenever $K$ is trace-class and $I-K$ is invertible. Choosing $K=H(b)H((b^{-1_\star})^\sim)$ gives
\begin{equation}\label{step determinant before truncated}
\det\left(P_nT(b_+^L)T(b_-^L\star(b_+^R)^{-1_\star})T((b_-^R)^{-1_\star})P_n\right)=\frac{\det\left(\mathrm I- H(z^{-n}\star b)H((b^{-1_\star})^\sim\star z^{-n}) \right)}{\det\left(T(b)T(b^{-1_\star})\right)}\, .
\end{equation}
Since $T(b_+^L)$ and $T((b_-^R)^{-1_\star})$ are lower and upper triangular operators, respectively, we have
\begin{equation}
    P_nT(b_+^L)=P_nT(b_+^L)P_n , \quad T((b_-^R)^{-1_\star})P_n =P_nT((b_-^R)^{-1_\star})P_n
\end{equation}
so the determinant on the left side of eq.~\eqref{step determinant before truncated} can be written in terms of finite sections of star Toeplitz operators. Moreover, for the same reason it holds
\begin{equation}
    T_n(b_+^L)=T_n^{-1}((b_+^L)^{-1_\star}), \quad T_n((b_-^R)^{-1_\star})=T_n^{-1}(b_-^R) \; .
\end{equation}
This allows us to rewrite \eqref{step determinant before truncated} as follows
\begin{equation}
\det T_n(b_-^L\star(b_+^R)^{-1_\star})=
\frac{\det\left(\mathrm I- H(z^{-n}\star b)H((b^{-1_\star})^\sim\star z^{-n}) \right)}{\det\left(T(b)T(b^{-1_\star})\right)} \det T_n(b_-^R)\det T_n\left((b_{+}^L)^{-1_\star}\right)\, .
\end{equation}
The determinant of the star Toeplitz matrices on the right hand side is the product of their diagonal elements, thus we have
\begin{multline}\label{eq:BOCGgen}
\det T_n(b_-^L\star(b_+^R)^{-1_\star})\\
=
\frac{\det\left(\mathrm I- H(z^{-n}\star b)H((b^{-1_\star})^\sim\star z^{-n}) \right)}{\det\left(T(b)T(b^{-1_\star})\right)}\exp\left(\sum_{j=1}^n\log\left[ \left(\left(b_{-}^R\right)_j\right)_0
\left(\left(b_{+}^L\right)^{-1_\star}_j\right)_0\right] \right)\, .
\end{multline}
The final statement follows from the identification $a_-^R
\equiv b_-^L$, $a_+^R\equiv (b_+^R)^{-1_\star}$, and hence $a_+^L=(b_+^L)^{-1_\star}$, $a_-^L=b_-^R$.
Note that the numerator of the fraction on the right hand side of
\eqref{eq:BOCGgen} can be bounded using Lemma~\ref{l:boundHH}.
\end{proof}

\subsection{Szeg\"o type limit theorems---proof of theorem~\ref{c:strong}}

\begin{lemma}\label{l:trT}
Let $a_\pm$ and $b_\pm$ be $\pm $ symbols in $V$ such that $T(a_\pm),T(b_\pm)$ are trace class.
Then the following identities hold:
\begin{align}
&\mathrm{tr}\, T(\{a_-,b_-\}_M)=\mathrm{tr}\, T(\{a_+,b_+\}_M)=0\label{eq:trT--}\\
&\mathrm{tr} \ T(\{a_+,b_-\}_M)=-i\sum_{m=1}^\infty \sum_{j=1}^{m}(a_{+,j-\frac{m}{2}})_{ m}(b_{-,j-\frac{m}{2}})_{-m}
\label{eq:trT+-}\, ,
\end{align}
where $\{c_1,c_2\}_M$ is a shorthand for $-i (c_1\star c_2-c_2\star c_1)$ and is known as Moyal bracket. 
\end{lemma}
\begin{proofline}
The first identity follows from Lemma~\ref{facttoeplitzmp} and from the fact  that the trace of the commutator of two trace class operators vanishes. Specifically, we have
\begin{equation}
T(\{a_-,b_-\}_M)=-i [T(a_-),T(b_-)]
\end{equation}
which has zero trace. Similarly, Lemma~\ref{l:TTHH} implies
\begin{equation}
T(\{a_+,b_-\}_M)=-i [T(a_+), T(b_-)]-i H(a_+)H((b_-)^\sim),
\end{equation}
where the first term has zero trace and the trace of the second one can be evaluated in the canonical basis, giving Eq.~\eqref{eq:trT+-}. 
\end{proofline}

To simplify the expressions appearing in Proposition~\ref{t:BOCG} it would be convenient to work with trace-class operators, for which results such as Lemma~\ref{l:trT} could be applied. By contrast, for a given symbol $a\in V$ with zero star winding number and factorized as $a=\exp_\star(b^R_-)\star \exp_\star(b^R_+)=\exp_\star(b^L_-)\star \exp_\star(b^R_+) $, the operators $T(b^{R/L}_\pm)$ are not necessarily trace class.
Enforcing them to be trace class might seem a strong assumption, but, in fact, it can be done without loss of generality.  
This is possible because the determinants of finite sections do not depend on the behavior of the operator at infinity, which instead determines the trace-class property. We are then free to deform the operators so as to make them trace class without affecting the determinant of the finite sections. We formalize this intuition by introducing a regularisation of the symbol, i.e., a sequence of symbols converging to the one of interest with respect to a suitable metric,
introduced in the following definition, that is not sensitive to irrelevant details at infinity. 

\phantom{a}

\begin{definition}
For $a\in V$, we define the norm $\| \cdot \|_{loc}$ as
\begin{equation}\label{loc norm def}
\parallel a \parallel_{loc}=\sum_{n=0}^\infty \frac{1}{2^{n+1}}\parallel T_{2n+1}(\texttt T^{-n-1}(a))\parallel\, ,
\end{equation}
where $\texttt T$ is the shift operator ($\texttt T^{-n-1}: a_x(z)\mapsto a_{x-n-1}(z)$).
\end{definition}

\phantom{a}

\begin{definition}
For $n\in \mathbb{N}$ we introduce the projections $\tilde P_n$ and $ \tilde Q_n$ on $\ell^2(\mathbb Z)$ acting by the rules
\begin{equation}
    (\tilde P_nf)_j=\left\{ \begin{array}{cc}
    f_j     &  -n\leq j\leq n\\
    0     &  \textrm{otherwise}
    \end{array} \right.\qquad (\tilde Q_nf)_j=\left\{ \begin{array}{cc}
    0     &  -n\leq j\leq n\\
    f_j     &  \textrm{otherwise.}
    \end{array} \right.
\end{equation}
\end{definition}

\begin{quotation}
\begin{remark}
Norm~\eqref{loc norm def} can be written as $ \|a\|_{loc}=\sum_{n=0}^\infty \frac{1}{2^{n+1}}\|\tilde P_n L(a) \tilde P_n\|$.
\end{remark}
\end{quotation}

% \begin{definition}
% For $a\in V$, we define the norm $\parallel \cdot \parallel_{loc}$ as
% \begin{equation}\label{loc norm def}
% \parallel a \parallel_{loc}=\sum_{n=0}^\infty \frac{1}{2^{n+1}}\parallel T_{2n+1}(\texttt T^{-n}(a))\parallel\, .
% \end{equation}
% \end{definition}

\begin{quotation}
\begin{remark}\label{bound loc norm}
The norm satisfies $\parallel a \parallel_{loc} \leq\|L(a)\|$.
\end{remark}
\end{quotation}

% {\color{red}
% \begin{quotation}
% \begin{remark}\label{remark order reflection projection}
% For any $a\in V$ and $k\in \mathbb{N}$ we have $\tilde P_k L( a^\sim) \tilde P_k =L((a_k)^\sim)$, where $a_k\in V$ is defined by $L(a_k)=\tilde P_k L( a) \tilde P_k$. That is to say that the mutual order of the reflection and the projection is irrelevant.
% \end{remark}
% \end{quotation}

% \begin{quotation}
% \begin{remark}\label{remark order reflection projection}
% For a $+$ symbol the matrix $L(a_+)$ is upper triangular symbol $a_\pm\in V$ we have $\tilde P_n L(a_+)$ 
% \end{remark}
% \end{quotation}
% }

\phantom{a}

\begin{lemma}\label{lemma bound loc norm}
    For any $a,b\in V$ and $m\in\mathbb{N}_0$ we have
    \begin{equation}\label{bound loc norm 1}
             \| a\star b\|_{loc} \leq \| L(b) \| \left( 2^{m+1} \|a\|_{loc}+\sum_{n=0}^\infty\frac{1}{2^{n+1}} \| \tilde  P_n L(a) \tilde Q_{n+m}\|\right)  ,
        \end{equation}
       \begin{equation}
             \| a\star b\|_{loc} \leq \| L(a) \| \left( 2^{m+1} \|b\|_{loc}+\sum_{n=0}^\infty\frac{1}{2^{n+1}} \| \tilde  Q_{n+m} L(b) \tilde P_{n}\|\right)  .
        \end{equation}
\end{lemma}
\begin{proof}
    For any $n,m\in \mathbb{N}_0$, we have
\begin{equation}
\begin{split}
    &\tilde P_n L(a\star b) \tilde P_n= \tilde P_n L(a) L(b) \tilde P_n =  \tilde P_n L(a)(\tilde P_{n+m}+\tilde Q_{n+m})L(b) \tilde P_n\\
    &=\tilde P_n \tilde P_{n+m}L(a) \tilde P_{n+m}  L(b) \tilde P_{n+m}\tilde P_{n}+ \tilde P_n L(a) \tilde Q_{n+m}   L(b) \tilde P_{n},
\end{split}
\end{equation}
where in the last equality we  used $\tilde P_n \tilde P_{n+m} f=\tilde P_n f$ and $f\tilde P_{n+m} \tilde P_{n}=f \tilde P_n$. Both inequalities readily follow.
\end{proof}
\phantom{a}

\begin{lemma}\label{lemma continuity loc norm}
The following functions are continuous in the metric space $(V,\|\cdot\|_{loc})$:
\begin{enumerate}
\item $\det T_n:a\mapsto \det T_n(a) $ for every $n\in\mathbb{N}$ 
\item $F^l_b:V\rightarrow V, a\mapsto b\star a$ and $F^r_b: V\rightarrow V, a\mapsto a\star b$
\item $\Phi:V^2\rightarrow V :(a,b)\mapsto\Phi_{BCH}(a,b) $ with $\Phi_{BCH}$ defined in \eqref{eq:BCH}.
\end{enumerate}
\end{lemma}
\begin{proofline}
% Let us introduce the projections $\tilde P_n$ and $ \tilde Q_n$ on $\ell^2(\mathbb Z)$ acting by the rules
% \begin{equation}
%     (\tilde P_nf)_j=\left\{ \begin{array}{cc}
%     f_j     &  -n\leq j\leq n\\
%     0     &  \textrm{otherwise}
%     \end{array} \right.\qquad (\tilde Q_nf)_j=\left\{ \begin{array}{cc}
%     0     &  -n\leq j\leq n\\
%     f_j     &  \textrm{otherwise.}
%     \end{array} \right.
% \end{equation}
% Norm~\eqref{loc norm def} can be written as $ \|a\|_{loc}=\sum_{n=0}^\infty \frac{1}{2^{n+1}}\|\tilde P_n L(a) \tilde P_n\|$.
For any $a,b\in V$ we have
\begin{equation}
  \|T_n(a-b) \|\leq  \|\tilde P_n L(a-b)\tilde P_n \|\leq 2^{n+1}\|a-b\|_{loc} \ ,
\end{equation}
which implies $\| T_n(a)-T_n(b)\| \to 0$ as $a\to b$. Since $\det T_n$ is continuous with respect to the metric induced by $\|\cdot \|$ we also have $\det T_n(a)\to\det T_n(b)$ as $a\to b$.

The functions $F_b^{l},F_b^r$ are linear, so proving continuity is equivalent to proving that they tend to zero as their argument approaches zero.
% To that aim, note that, for any $n,m\in \mathbb{N}$, we have
% \begin{equation}
% \begin{split}
%     &\tilde P_n L(a\star b) \tilde P_n= \tilde P_n L(a) L(b) \tilde P_n =  \tilde P_n L(a)(\tilde P_{n+m}+\tilde Q_{n+m})L(b) \tilde P_n\\
%     &=\tilde P_n \tilde P_{n+m}L(a) \tilde P_{n+m}  L(b) \tilde P_{n+m}\tilde P_{n}+ \tilde P_n L(a) \tilde Q_{n+m}   L(b) \tilde P_{n},
% \end{split}
% \end{equation}
% where in the last equality we  used $\tilde P_n \tilde P_{n+m} f=\tilde P_n f$ and $f\tilde P_{n+m} \tilde P_{n}=f \tilde P_n$. From this it follows
% {\color{red}
% \begin{equation}
%     \| a\star b\|_{loc} \leq \| L(b) \| \left( 2^{m+1} \|a\|_{loc}+\sum_{n=0}^\infty\frac{1}{2^{n+1}} \| \tilde  P_n L(a) \tilde Q_{n+m}\|\right)
% \end{equation}
This follows from inequality~\ref{bound loc norm 1}, indeed, by Lemma~\ref{lemma bound norm infinite matrix}, $\| \tilde  P_n L(a) \tilde Q_{n+m} \|=O(\rho_1^m)$ uniformly in $n\in\mathbb{N}$, for some $\rho_1\in(0,1)$.
% Now, by lemma~\ref{lemma bounded linear operators}, we have $\| \tilde  P_n L(a) \tilde Q_{n+m} \|=O(\rho_1^m)$ for some $\rho_1\in(0,1)$.
Thus for every $\varepsilon>0$ there is $m\in\mathbb N$ sufficiently large that $\| \tilde  P_n L(a) \tilde Q_{n+m} \| \| L(b) \|<\varepsilon/2$. Taking $\delta>0$ such that $2^{m+1} \| L(b)\| \delta<\varepsilon/2$ we obtain the implication $\| a\|_{loc}<\delta\Rightarrow \|a\star b \|_{loc}<\varepsilon$. Consequently, we have $\lim_{a\to 0} F_b^r(a)=0$ and, by linearity, $F_b^r$ is continuous on its domain. The proof for $F_b^l$ is analogous. 

Finally, the function $\Phi_{BCH}$ defined in \eqref{eq:BCH} is a uniformly convergent series of Moyal products, and we have  just shown that the Moyal product is a continuous function of its arguments; thus, $\Phi_{BCH}$ is continuous.

\end{proofline}

\begin{definition}[Regularisation]\label{definition regularisation}
Let $\{\mathfrak a_{k}\}_{k=1}^\infty$ be a sequence of symbols $\mathfrak a_{k}$ with zero star winding number in $(V,\|\cdot\|_{loc})$ , and let $\mathfrak b_{k,\pm}^{L/R}$ denote $\log_\star \mathfrak a_{k,\pm}^{L/R}$. 
We say that $\{\mathfrak a_{k}\}_{k=1}^\infty$ regularises the symbol $a$ if 
\begin{itemize}
\item[-] $a=\lim_{k\rightarrow \infty}\mathfrak a_k$;
% \item[-] $L(\mathfrak b_{k,\pm}^{L/R})$ is uniformly bounded  in $k$;
% \item[-] There is $\rho\in(0,1)$ such that, for any $k\in \mathbb N$, $\mathfrak b_{k,\pm}^{L/R}\in V_\rho$;
\item[-] $\mathfrak b_{k,\pm}^{L/R}$ are bounded uniformly in $k\in \mathbb N$ on the same annulus of analyticity, and $(\mathfrak b_{k,\pm}^{L},\mathfrak b_{k,\mp}^{L}),(\mathfrak b_{k,\pm}^{R},\mathfrak b_{k,\mp}^{R})\in S_{\Phi}$ for all $k\in\mathbb{N}$.
\item[-] For any $k\in \mathbb N$, $T(\mathfrak b_{k,\pm}^{L/R})$ are trace class.
\end{itemize}

\end{definition}

\begin{quotation}
\begin{remark}\label{remark regularisation tridiagonal}
For instance, the following step regularization is effective in the tridiagonal case  discussed in Example~\ref{ex:tridiagonal}: 
\begin{equation}
\begin{aligned}
\mathfrak a_{k,-}^{R}=&1+K_k\star (a_-^{R}-1)\\
\mathfrak a_{k,+}^{L}=& 1+(a_+^{L}-1)\star K_k
\end{aligned}
\end{equation}
where $K_k=[(x,p)\mapsto\chi_{[-k,k]}(x)]$ and $\chi$ denotes the characteristic function. 
\end{remark}
\end{quotation}

\phantom{a}

\begin{lemma}\label{remark regularisation limit a b}
Let $a\in V$ be a symbol with zero star winding number and define its Wiener-Hopf star factorisation $a_\pm^{R/L}\equiv \exp_\star(b^{R/L}_\pm)$ in such a way that all the diagonal matrix elements of $L(a_-^{R/L})$  are equal to $1$.  If $\{\mathfrak a_k\}_{k=1}^\infty$ is a regularisation for $a$ and $\mathfrak a_{k,\pm}^{R/L}\equiv\exp_\star(\mathfrak b^{R/L}_\pm)$ are the corresponding factors with the same convention, then $\lim_{k\to\infty}\mathfrak b^{R/L}_{k,\pm}=b^{R/L}_\pm$ in $(V,\|\cdot\|_{loc})$.
\end{lemma}

\begin{proofline}
Let us focus on the right factorization and ease the notation by writing $b^\pm\equiv b^R_\pm$, $\mathfrak{b}_k^\pm \equiv \mathfrak b^{R}_{k,\pm}$. The proof for the left factorization is analogous. From $\mathfrak a_k=e^{\mathfrak b^-_{k}}_\star \star e^{\mathfrak b^+_{k}}_\star$, $a=e^{ b^-}_\star \star e^{ b^+}_\star$ we get
\begin{equation}
    e^{\mathfrak b^+_k}_\star \star e_\star^{-b^+}-  e_\star^{-\mathfrak b^-_k}\star e_\star^{b^-}=e^{-\mathfrak b^-_k}_\star\star (\mathfrak a_k-a)\star e_{\star}^{-b^+}
\end{equation}
% \begin{multline}
%      \| a\star b\star c\|_{loc}\leq 2^{m_1+m_2+2}\|L(a)\|\|L(c)\|\|b\|_{loc} +2^{m_1+1}\|L(a)\| \|L(c)\|\sum_{n_2=0}^{\infty}\frac{1}{2^{n_2+1}}\| \tilde P_{n_2} L(b) \tilde Q_{n_2+m_2}\|  \\+\|L(a)\| \sum_{n_1=0}^\infty\frac{1}{2^{n_1+1}}\|\tilde Q_{n_1+m_1}L(b\star c) \tilde P_{n_1}\|
% \end{multline}
By applying the two inequalities in Lemma~\ref{lemma bound loc norm} to the right hand side we have
\begin{multline}\label{step loc norm  three factors inequality}
     \| e^{\mathfrak b^+_k}_\star \star e_\star^{-b^+}-  e_\star^{-\mathfrak b^-_k}\star e_\star^{b^-}\|_{loc}\leq 2^{m_1+m_2+2}\|L(e^{-\mathfrak b^-_k}_\star)\|\|L(e_{\star}^{-b^+})\|\|\mathfrak a_k-a\|_{loc} \\+2^{m_1+1}\|L(e^{-\mathfrak b^-_k}_\star\| \|L(e_{\star}^{-b^+})\|\sum_{n_2=0}^{\infty}\frac{1}{2^{n_2+1}}\| \tilde P_{n_2} L(\mathfrak a_k-a) \tilde Q_{n_2+m_2}\|  \\+\|L(e^{-\mathfrak b^-_k}_\star)\| \sum_{n_1=0}^\infty\frac{1}{2^{n_1+1}}\|\tilde Q_{n_1+m_1}L((\mathfrak a_k-a)\star e_{\star}^{-b^+}) \tilde P_{n_1}\|
\end{multline}
We proceed similarly as in the proof of Lemma~\ref{lemma continuity loc norm}. Let $\varepsilon>0$ be arbitrary. By uniform boundedness in $k$ and Lemma~\ref{lemma bound norm infinite matrix}, we have $\|\tilde Q_{n_1+m_1}L((\mathfrak a_k-a)\star e_{\star}^{-b^+}) \tilde P_{n_1}\|=O(\rho_1^{m_1})$ uniformly in ${n_1,k}\in\mathbb{N}$, for some $\rho_1\in(0,1)$. We choose sufficiently large $m_1$ so that the third term in~\eqref{step loc norm  three factors inequality} is smaller than $\varepsilon/3$. Similarly we have $\| \tilde P_{n_2} L(\mathfrak a_k-a) \tilde Q_{n_2+m_2}\| =O(\rho_1^{m_2})$ uniformly in ${n_2,k}\in\mathbb{N}$. With $m_1$ already chosen, we take sufficiently large $m_2$ so that the second term in~\eqref{step loc norm  three factors inequality} is smaller than $\varepsilon/3$. Finally, since $\lim_{k\to\infty}\|\mathfrak{a}_k-a\|_{loc}=0$, for sufficiently large $k$ the first term in~\eqref{step loc norm  three factors inequality} is smaller than $\varepsilon/3$, and the whole expression is smaller than $\varepsilon$. Since $\varepsilon$ is arbitrary, it follows
\begin{equation}
    \lim_{k\to\infty} \| e^{\mathfrak b^+_k}_\star \star e_\star^{-b^+}-e_\star^{-\mathfrak b^-_k}\star e_\star^{b^-}\|_{loc}=0 
\end{equation}
Since the $\pm$ components are split here we have $\lim_{k\to\infty}\mathfrak b_k^\pm=b^\pm$.
\end{proofline}

\begin{lemma}\label{lemma BCH}
Given two bounded linear operators $A$ and $B$ in a Banach algebra, for small enough $\nu\in \mathbb R_{>0}$ we have
\begin{equation}
\log(e^{\nu A}e^{\nu B})-\nu A-\nu B=\frac{1}{4}[\nu A,\Phi(\nu A,\nu B)]-\frac{1}{4}[\Phi(-\nu B,\textcolor{red}{-}\nu A),\nu B]
\end{equation}
where
\begin{multline}%\label{eq:BCH}
\Phi(A,B)=\left(\int_0^1\int_0^1 e^{- s\,  \mathrm{ad}_A}\psi_1(e^{ \mathrm{ad}_A}e^{t\,  \mathrm{ad}_B})\, ds \, dt\right) B\\
-\left(\int_0^1\int_0^t e^{ s\,  \mathrm{ad}_A}e^{\mathrm{ad}_B}\psi_1(e^{- \mathrm{ad}_B}e^{- t\,  \mathrm{ad}_A})\, ds \, dt\right) A\, ,
\end{multline}
$
\psi_1(x)=2x\frac{x\log x+1-x}{(x-1)^2}
$ and $\mathrm {ad}_A(B)=[A,B]$.
\end{lemma}

\begin{proofline}
This can be seen as a corollary  of the Baker–Campbell–Hausdorff theorem, which can be expressed as
\begin{equation}
\log(e^{A}e^{B})- A- B=\int_0^1\psi(e^{ \mathrm{ad}_A}e^{t\,  \mathrm{ad}_B}) dt B-B=\int_0^1\bigl(\psi(e^{ \mathrm{ad}_A}e^{t\,  \mathrm{ad}_B})-1\bigr) dt B
\end{equation}
where $\psi(x)=\frac{x\log x}{x-1}$. The first step is to make it explicit the symmetry $\log(e^A e^B)=-\log(e^{-B} e^{-A})$. Specifically, we have
\begin{equation}
\log(e^{A}e^{B})- A- B=\frac{1}{2}\int_0^1\bigl(\psi(e^{ \mathrm{ad}_A}e^{t\,  \mathrm{ad}_B})-1\bigr) dt B-\frac{1}{2}\int_0^1\bigl(\psi(e^{- \mathrm{ad}_B}e^{-t\,  \mathrm{ad}_A})-1\bigr) dt A\, .
\end{equation}
Since $\psi(x)-1$ has a simple zero at $x=1$, it is convenient to define $\psi_1(x)=2x \frac{\psi(x)-1}{x-1}$ and express the right hand side in terms of $\psi_1$ as follows
\begin{multline}\label{eq:BCHsym}
\log(e^{A}e^{B})- A- B=\\
\frac{1}{4}\int_0^1(1-e^{-t\,  \mathrm{ad}_B}e^{- \mathrm{ad}_A})\psi_1(e^{ \mathrm{ad}_A}e^{t\,  \mathrm{ad}_B}) dt B+\frac{1}{4}\int_0^1(1-e^{t\,  \mathrm{ad}_A}e^{ \mathrm{ad}_B})\psi_1(e^{- \mathrm{ad}_B}e^{-t\,  \mathrm{ad}_A}) dt A\, .
\end{multline}
We can identify the terms of the expansion written as a commutator with respect to $A$ or $B$ by using the identity
\begin{equation}\label{eq:ident_comm}
1-e^A e^B=
-B \int_0^1 e^{s B} ds-A\int_0^1 e^{s A} e^B ds\, .
\end{equation}
Applying \eqref{eq:ident_comm} to the factors multiplying $\psi_1$ from the left in \eqref{eq:BCHsym} gives
\begin{multline}
\log(e^{A}e^{B})- A- B=\\
\frac{1}{4}\left[A,\int_0^1\int_0^1 e^{-s\,  \mathrm{ad}_A} \psi_1(e^{ \mathrm{ad}_A}e^{t\,  \mathrm{ad}_B}) dsdt B-
\int_0^1\int_0^t e^{s \mathrm{ad}_A} e^{\mathrm{ad}_B } \psi_1(e^{- \mathrm{ad}_B}e^{-t\,  \mathrm{ad}_A}) ds dt A
\right]\\
+\frac{1}{4}\left[B,\int_0^1\int_0^t e^{-s \mathrm{ad}_B} e^{-\mathrm{ad}_A} \psi_1(e^{ \mathrm{ad}_A}e^{t\,  \mathrm{ad}_B}) ds dt B-
\int_0^1\int_0^1 e^{s \mathrm{ad}_B } \psi_1(e^{- \mathrm{ad}_B}e^{-t\,  \mathrm{ad}_A}) ds dt A
\right]=\\
\frac{1}{4}[A,\Phi(A,B)]-\frac{1}{4}[\Phi(-B,-A),B]\, .
\end{multline}
Finally, the statement of the lemma is a consequence of the fact that, for bounded operators $A$ and $B$,  the Baker-Campbell-Hausdorff formula for $\log(e^{\nu A}e^{\nu B})$ has a nonzero radius of absolute convergence for small $\nu$~\cite{Suzuki1977}.  
\end{proofline}
\begin{quotation}
\begin{remark}\label{remark convergence BCH}
A sufficient condition for the absolute convergence of the Baker-Campbell-Hausdorff formula is $\nu(\|A\|+\|B\|)<\log 2$~\cite{Suzuki1977}.  
\end{remark}
\end{quotation}

\phantom{a}

\begin{proof}[Proof of Theorem~\ref{c:strong}]
Let us focus on a particular symbol $\mathfrak a$ belonging to a sequence regularising the symbol $a$. The denominator of the BOCG formula can be manipulated as follows
\begin{equation}
    \begin{split}
    \det(T(\mathfrak a_-^L\star(\mathfrak a_+^R)^{-1_\star})T((\mathfrak a_-^R)^{-1_\star}\star \mathfrak a_+^L))&=\det\left(T(\mathfrak a_-^L)T((\mathfrak a_+^R)^{-1_\star})T((\mathfrak a_-^R)^{-1_\star})T( \mathfrak a_+^L)\right)\\
 &=\det\left(e^{T(\mathfrak{b}^L_-)}\right)\det\left(e^{-T(\mathfrak{b}^R_+)}\right)\det\left(e^{-T(\mathfrak{b}^R_-)}\right)\det\left(e^{T(\mathfrak{b}^L_+)}\right) \\
 &=\exp\left(\mathrm{tr}\ T(\mathfrak b_+^L+\mathfrak b_-^L-\mathfrak b_-^R-\mathfrak b_+^R)\right) , \label{step proof szego exponent trace}
    \end{split}
\end{equation}
where $\mathfrak b_{\pm}^{L/R}=\log_\star \mathfrak a_{\pm}^{L/R}$.
Combining this with the numerator of the BOCG formula we have
\begin{equation}\label{step szego fraction}
\begin{split}
&\frac{\exp\left(\sum_{j=1}^n\log\bigl(( \mathfrak  a_{+,j}^L)_0 (\mathfrak a_{-,j}^L)_0\bigr) \right)}{\det(T(\mathfrak a_-^L\star(\mathfrak a_+^R)^{-1_\star})T((\mathfrak a_-^R)^{-1_\star}\star \mathfrak a_+^L))}\\&=\exp\left[\left(\sum_{j=1}^n (\log_\star  \mathfrak a)_j+\sum_{j=n+1}^\infty \left( \log_\star\mathfrak a -\mathfrak b_-^L -\mathfrak b_+^L \right)_j -\sum_{j=1}^\infty \left( \log_\star\mathfrak a -\mathfrak b_-^R -\mathfrak b_+^R \right)_j \right)_0\right],
\end{split}
\end{equation}
trivially obtained using $((\log \mathfrak{a}_\pm^L)_j)_0=((\log_\star \mathfrak{a}_\pm^L)_j)_0$. Now we use Lemma~\ref{lemma BCH} to write
\begin{equation}\label{step proof szego 2}
\begin{split}
&\log_\star \mathfrak a -  \mathfrak b_+^L-  \mathfrak b_-^L=\frac{1}{4}\left(i \{  \mathfrak b_+^L,  \mathfrak d_{(-)}^L\}_M+i \{ \mathfrak d_{(+)}^L,  \mathfrak b_-^L\}_M\right), \quad \\
&\log_\star \mathfrak a -  \mathfrak b_+^R-  \mathfrak b_-^R=\frac{1}{4}\left(i \{  \mathfrak b_-^R, \mathfrak d^R_{(+)}\}_M+i \{ \mathfrak d^R_{(-)},  \mathfrak b_+^R\}_M\right)
\end{split}
\end{equation}
with
\begin{equation}
\begin{aligned}
 \mathfrak d^L_{(-)}=&\Phi_{BCH}( \mathfrak b_+^L, \mathfrak b_-^L)&  \mathfrak d^R_{(+)}=&\Phi_{BCH}( \mathfrak b_-^R, \mathfrak b_+^R)\\
 \mathfrak d^L_{(+)}=&-\Phi_{BCH}(- \mathfrak b_-^L,- \mathfrak b_+^L)&
 \mathfrak d^R_{(-)}=&-\Phi_{BCH}(- \mathfrak b_+^R,- \mathfrak b_-^R)\, .
\end{aligned}
\end{equation}
where $\Phi_{BCH}$ is defined in \eqref{d:PhiBCH}. 
The final expression follows from applying, to the BOCG formula, the limit of the sequence regularising $a$. By the first point of Lemma \ref{lemma continuity loc norm}, the limit of the left hand side of the BOCG formula reads $\lim_{k\to\infty} {\det T_n(\mathfrak a_k)}=\det T_n(a)$. The limit $k\to\infty$ of the contribution \eqref{step szego fraction} to the BOCG formula is evaluated  using Lemma~\ref{remark regularisation limit a b} and the continuity properties of the Moyal product and of function $\Phi_{BCH}$, given in Lemma~\ref{lemma continuity loc norm}. Specifically, continuity implies that the limit $k\to\infty$ is obtained simply by replacing the symbols $\mathfrak a_k, \mathfrak b_{k,\pm}^{R/L}$ with their respective limits $a, b_{\pm}^{R/L}$. Finally, due to uniform boundedness in $k$ in the definition of the regularisation, the bound for the remaining part of the BOCG formula, given in Remark~\ref{rem:boundHH}, holds uniformly in $k$, and hence also in the limit $k\to\infty$.
% The final expression can be obtained using lemma~\ref{l:trT}. and taking the limit of the sequence regularising $a$, in light of the continuity properties given in lemmas \ref{lemma continuity loc norm} and \ref{remark regularisation limit a b}
\end{proof}

\section{Matrices with slowly varying elements along the diagonals}\label{s:locally}

Here we refine the results obtained in the previous section when there is a  flow that connects the terms of the asymptotic expansion of $\log\det T_n(a)$ for large $n$  with those corresponding to the symbol of a Toeplitz operator. The aim of the connection is to work out the Wiener-Hopf star factorisations.  
We restrict ourselves to the case in which the flow can be parametrised by a scaling variable, $s \in [0,1]$, in the sense that each symbol of the family is represented by a function $f_{s\nu}:(x,p)\rightarrow f(x s\nu,p)$, where $\nu$ is an auxiliary parameter quantifying the contraction of the image of $f_{\nu}$ associated with taking a derivative with respect to its first argument.  The ultimate goal would be to express the strong Szeg\"o limit theorem~\ref{c:strong} in the form of a series in $\nu$ whose terms can be computed systematically. 
A close inspection  to Theorem~\ref{c:strong} reveals that one needs the asymptotic expansion of
\begin{enumerate}
\item sums
\item the Moyal star product
\item left and right Wiener-Hopf star factorisations
\item the star logarithm
\item the function $\Phi_{BCH}$ in \eqref{eq:BCH}. 
\end{enumerate}
The former ingredient is required because the distinction between corner and bulk terms envisaged in \eqref{eq:Ansatz_st} and \eqref{eq:Ansatz_st1} is not yet manifest in \eqref{eq:strongSzego}, in which the (pretended) bulk term is expressed as a sum rather than as an integral. To that aim, we adapt the Euler-Maclaurin formula to our case. 
The formal expansion of the Moyal star product in the order of the derivatives is well known~\cite{Zachos2005Quantum} and Section~\ref{ss:truncatedMoyal} is mainly focussed on providing  conditions that are sufficient for truncating the series. 
We are not aware of analogous results for the left and right Wiener-Hopf star factorisations, so Section~\ref{ss:truncatedWH} is also devoted to develop a systematic asymptotic expansion.
The expansion of the star logarithm is worked out in Section~\ref{ss:starfunctions} with the standard method of the resolvent. 
Finally, we will be in a position to use the Baker-Campbell-Hausdorff formula for the asymptotic expansion of $\Phi_{BCH}$ as the series is uniformly convergent for small enough $\nu$. 

We collect the main asymptotic formulas in Table~\ref{tab:1} for quick reference. 

\begin{table}
\begin{tabular}{c|c|c}
&asymptotic series&\\
\hline
$\sum_{j=1}^n g(j\nu)$&$\frac{1}{\nu}\int_{\frac{\nu}{2}}^{(n+\frac{1}{2})\nu}\left[ g(y)-\sum_{j=1}^{\infty}\frac{\nu^{2j}}{(2j)!}(1-2^{1-2j})B_{2j}(1)g^{(2j)}(y)\right] dy$&Euler-Maclaurin\\
$(f_\nu \star g_\nu )(x,p)$&$\bigl(f e^{i \nu \frac{\overleftarrow\partial_p\overrightarrow\partial_x-\overleftarrow\partial_x\overrightarrow\partial_p}{2}}g\bigr)(x\nu,p)$&Groenewold\\
$f_{\pm}^{R/L}(\nu)$&$e^{(\log f)^{\pm}}\sum_{j=0}^\infty(\psi_j^{R/L}[\nabla\log f])^{\pm}\frac{\nu^j}{j!}$& Lemma~\ref{l:leftrightexp} with def.~\ref{d:psi}\\
$(\zeta-f)^{-1_\star}$&$\sum_{j=0}^{\infty}\frac{R_j^{(\zeta)}[f]}{(2j)!}\nu^{2j}$&Lemma~\ref{lemma star inverse} with def.~\ref{d:R}\\
$\nu^{-1}\Phi_{BCH}(\nu c,\nu d)$&$ d+\frac{1}{3}i\nu\{c,d\}_M+\frac{1}{12}\nu^2\{\{c,d\}_M,d\}_M+\dots$&Baker-Campbell-Hausdorff
\end{tabular}\caption{Outline of the asymptotic series used or worked out in this section. Here $f_\nu(x,p):(x,p)\mapsto f(x \nu,p)$ is a nice enough function representing the symbol and $f_\pm^{R/L}(\nu)$ represent the factors of the right and left factorisations.}\label{tab:1}
\end{table}

\subsection{Truncated Moyal product}\label{ss:truncatedMoyal}
When two symbols characterised by the same value of $\nu$ are (star) multiplied, $\nu$ effectively enters the star product as the reduced Planck's constant $\hbar$ does in phase-space quantum mechanics. This opens the door to  asymptotic expansions in the limit of small $\nu$ analogous to semiclassical approximations based on the expectation that, close to classical settings, $\hbar$ can be treated as a small parameter. The following lemma provide a justification for the truncation of the star product.

\begin{definition}\label{definition strip}
We denote by $\str[\rho]$ the  horizontal strip in the complex plane $\{z\in\mathbb{C}: \ \log \rho <\mathrm{Im}z<-\log \rho \}$ for some $\rho\in(0,1)$. We also use $\strp[\rho]$ as a shorthand for $\str[\rho]/(2\pi\mathbb Z)$.
\end{definition}

\begin{lemma}\label{lemma truncated moyal analytic}
Let $f,g$ be two functions on $\mathbb R\times \strp[\rho]$ uniformly bounded on the domain. Suppose they are of class $C^k$ in the first argument, for some $k\in\mathbb{N}$, such that the partial derivatives $\partial_1^jf, \partial_1^j g$, for $j=0,1,\ldots,k$ are analytic with respect to the second argument.
Then for two symbols $a=[f_\nu], \ b=[g_\nu]$ with $f_\nu(x,p)=f(x\nu,p),\ g_\nu(x,p)=g(x\nu,p)$, we have $a\star b=[h_\nu]$, where $h_\nu$ can be approximated as 
\begin{equation}\label{truncation moyal h h}
h_\nu:(x,p)\mapsto (h(\nu))(x\nu,p)+ O(\nu^k)
\end{equation}
as $\nu\to 0$, uniformly in  $(x,p)\in\tfrac{\Omega}{\nu}\times (\mathbb R/(2\pi\mathbb Z))$, for any compact set $\Omega\subset\mathbb R$, where
\begin{equation}\label{eq:truncationstar}
h(\nu):(y,p)\mapsto \sum_{j=0}^{k-1}\frac{1}{j!} \left.\Bigl(i\nu\frac{\partial_{p_1}\partial_{y_2}-\partial_{y_1}\partial_{p_2}}{2 }\Bigr)^j f(y_1,p_1)g(y_2,p_2)\right|_{p_1=p_2=p\atop y_1=y_2=y}\, .
\end{equation}In particular, for $k=\infty$ we can express the result as  an asymptotic series
\begin{equation}
    h_{\textcolor{red}{\nu}}(x,p)\sim \sum_{j=0}^{\infty}\frac{1}{j!} \left.\Bigl(i\nu\frac{\partial_{p_1}\partial_{y_2}-\partial_{y_1}\partial_{p_2}}{2 }\Bigr)^j f(y_1,p_1)g(y_2,p_2)\right|_{p_1=p_2=p\atop y_1=y_2=x\nu} \; .
\end{equation}
as $\nu\rightarrow 0$, uniformly in  $(x,p)\in\tfrac{\Omega}{\nu}\times (\mathbb R/(2\pi\mathbb Z))$.
\end{lemma}
\begin{proofline}
A proof is reported in  Appendix~\ref{a:extra}.
\end{proofline}

\begin{quotation}
\begin{remark}\label{remark for locally toeplitz}
If $f$ and $g$ and their $k$ first derivatives are uniformly bounded on the domain then Eq.~\eqref{truncation moyal h h} holds as $\nu\to 0$ uniformly in $(x,p)\in\mathbb{R}\times (\mathbb{R}/2\pi\mathbb{Z})$. An analogous result holds if the domain of the first argument is some open subset of $\mathbb{R}$.
\end{remark}
\phantom{a}
\begin{remark}
If $f$ and $g$ are entire with respect to their first argument and their derivatives are uniformly bounded, then the series in Eq.~\eqref{eq:truncationstar} is not only asymptotic, but converges absolutely for sufficiently small $\nu$---see Lemma~\ref{t:entire}.
\end{remark}
\end{quotation}

\subsection{Truncated Wiener-Hopf star factorisation}\label{ss:truncatedWH}
In this section we take advantage of the possibility to truncate the star product to work out an explicit approximation for the Wiener-Hopf star factorisations. An approximation for the product provides indeed an approximation for the factors:
\begin{lemma}\label{lemma approximation left right}
Consider a symbol $a\in \anb$ with zero star winding number. Suppose $a=a_-^\epsilon \star a_+^\epsilon+O(\epsilon)$ as $\epsilon\rightarrow 0$ for given $\pm$ symbols $a^\epsilon_\pm$ with $\log_\star a_\pm^\epsilon \in \anb$ uniformly bounded in $\epsilon$, where $\epsilon$ is some auxiliary parameter. Then, there is a right Wiener-Hopf star factorisation of $a$ such that
\begin{equation}
a_\pm ^R=a_\pm^\epsilon+O(\epsilon)\qquad \text{as $\epsilon\rightarrow 0$.}
\end{equation}
An analogous result applies to the factors of the left Wiener-Hopf star factorisation if $a=a_+^\epsilon\star a_-^\epsilon+O(\epsilon)$ as $\epsilon\rightarrow 0$.
\end{lemma}
\begin{proofline}
By assumptions the symbol satisfies
$
a=a_-^R\star a_+^R=a_-^\epsilon \star a_+^\epsilon+O(\epsilon)
$, 
and hence
\begin{equation}
a_+^R\star (a_+^\epsilon)^{-1_\star}=(a_-^R)^{-1_\star}\star a_-^R\star a_+^R\star (a_+^\epsilon)^{-1_\star}=(a_-^R)^{-1_\star}\star a_-^\epsilon +O(\epsilon)
\end{equation}
where we used that  $(a_-^R)^{-1_\star}$ and $(a_+^\epsilon)^{-1_\star}$ are uniformly bounded in $\epsilon$. Extracting the $+$ part of the members of the equation gives
\begin{equation}
a_+^R\star (a_+^\epsilon)^{-1_\star}=((a_-^R)^{-1_\star}\star a_-^\epsilon)_0 +O(\epsilon)=\frac{(a_-^\epsilon)_0}{(a_-^R)_0} +O(\epsilon)\, .
\end{equation}
The first identity follows from the fact that $(a_-^R)^{-1_\star}\star a_-^\epsilon$ is a $-$ symbol and hence its $+$ part consists of solely the zeroth Fourier coefficient; the second identity is instead a special case of  $(a_-\star b_-)_0=(a_-)_0( b_-)_0$. The star product with $a_+^\epsilon$ then gives
\begin{equation}\label{eq:Rpeps}
a_+^R =\frac{(a_-^\epsilon)_0}{(a_-^R)_0}\star a_+^\epsilon+O(\epsilon)\, ,
\end{equation}
where we used again the boundedness properties. 
Analogously we have
\begin{equation}\label{eq:Rmeps}
a_-^R =a_-^\epsilon\star \frac{(a_-^R)_0}{(a_-^\epsilon)_0}  +O(\epsilon)\, .
\end{equation}
The result of the lemma applies to the alternative factorisation $a=\tilde a_-^R\star \tilde a_+^R$, with 
$
\tilde a_-^R=a_-^R\star  \frac{(a_-^\epsilon)_0} {(a_-^R)_0}$ and $ \tilde a_+^R=  \frac {(a_-^R)_0}{(a_-^\epsilon)_0}\star a_+^R$, 
indeed \eqref{eq:Rpeps} and \eqref{eq:Rmeps} imply
$
\tilde a_-^R=a_-^\epsilon+O(\epsilon)$ and $ \tilde a_+^R=  a_+^\epsilon+O(\epsilon)$. 
The proof for the left factorisation is analogous. 
\end{proofline}

The idea is then to construct an asymptotic series for the factors based on the knowledge of the case $\nu=0$, which corresponds to a Toeplitz matrix. 
\begin{definition}
Let  $\{a\}=\{a_s\in \anb : s\in[0,1]\}$ be a one-parameter family of symbols. If $a_0$ induces a Toeplitz operator then we say that $\{a\}$ is a star-Toeplitz flow of symbols in $\anb$. 
If there exists also the one-parameter family $\{a^{-1_\star}\}=\{a^{-1_\star}_{s} \in  \anb: s\in[0,1]\}$ of star inverses ($a_s\star a_s^{-1_\star}=1$),  then we say that $\{a\}$ is a star-Toeplitz flow of star-invertible symbols in $\anb$.
If there exist also the one-parameter families $\{a^{R/L}_\pm\}=\{a^{R/L}_{s,\pm} \in  \anb: s\in[0,1]\}$ of $\pm$ symbols such that $a_s=a_{s,-}^R\star a_{s,+}^R=a_{s,+}^L\star a_{s,-}^L$ and $\log_\star a^{R/L}_{s,\pm} \in  \anb$ then we say that $\{a\}$ is a star-Toeplitz flow of symbols in $\anb$ with zero star winding number. 
\end{definition}

By imposing that the star product of the factors is the symbol, we obtain a recurrence relation between the terms of the asymptotic series of the factors. %, which is related to the \textcolor{red}{following.}

\begin{definition}\label{d:psi}
Let $v_1 \partial_1+v_2\partial_2$ be a vector field with $v_1(\vec \varphi),v_2(\vec \varphi)$ its coordinates. We define the functions $\psi^H_m[\vec v]$ by the recursive equation
\begin{equation}\label{eq:inversion}
\begin{aligned}
\psi^H_0[\vec v]:&\vec\varphi\mapsto 1\\
\psi^H_m[\vec v]:&\vec\varphi\mapsto-\sum_{j',j''=0\atop{ j'+j''\leq m}}^{m-1}\binom{m}{j',j''}\left( i\tfrac{\sigma_H}{2}\det
\vec D^{v}_{\vec\varphi,\vec\varphi'}\right)^{m-j'-j''}(\psi_{j'}[\vec v(\vec\varphi)])^-(\psi_{j'}[\vec v(\vec\varphi')])^+\Bigr|_{\vec\varphi'=\vec\varphi}
\end{aligned}
\end{equation}
where $\binom{m}{j_1,j_2}=\frac{m!}{j_1! j_2! (m-j_1-j_2)!}$, $H\in\{L,R\}$, 
\begin{equation}
\sigma_R=-\sigma_L=1\, 
\end{equation}
and
\begin{equation}
D^{v}_{\vec\varphi,\vec\varphi'}=\begin{pmatrix}
\partial_{\varphi_1'}+ v_1^+(\vec\varphi')&\partial_{\varphi_1}+ v_1^-(\vec\varphi)\\
\partial_{\varphi_2'}+ v_2^+(\vec\varphi')&\partial_{\varphi_2}+ v_2^-(\vec\varphi)
\end{pmatrix}\, .
\end{equation}
\end{definition}

\begin{lemma}\label{l:leftrightexp} 
Let $f:\mathbb R\times \strp[\rho]\to\mathbb{C}$ be of class $C^k$ in the first argument, for some $k\in\mathbb{N}$, such that the partial derivatives $\partial_1^jf$, for $j=0,1,\ldots,k$ are analytic with respect to the second argument. Suppose $f_{s\nu}: (x,p)\mapsto f(xs\nu,p)$, with $s\in[0,1]$, parametrises a star-Toeplitz flow of symbols in $\anb$ with zero star winding number connecting $a=[f_\nu]$ with the Toeplitz symbol $e^{i p}\mapsto f(0,p)$.  
Then, for any compact $\Omega\subset\mathbb R$, the factors $a^{R/L}_\pm$ of the  right and left Wiener-Hopf star factorisations of $a$ are represented by functions $f^{R/L}_{\nu\pm}:(x,p)\mapsto (f_{\pm}^{R/L}(\nu))(x\nu,p)$ where $f^{R/L}_{\pm}(\nu):\mathbb{R}\times (\mathbb R/(2\pi\mathbb Z)) \to \mathbb{C}$ exhibit the following asymptotic behaviour 
\begin{equation}\label{eq:leftrightsmooth}
    f^{R/L}_\pm(\nu)=e^{(\log f)^{\pm}}\sum_{j=0}^k(\psi_j^{R/L}[\nabla\log f])^{\pm}\frac{\nu^j}{j!}  + O(\nu^{k+1})
\end{equation}
as $\nu\to 0$, uniformly in $\tfrac{\Omega}{\nu}\times (\mathbb R/(2\pi\mathbb Z))$.
\end{lemma}
\begin{proofline}
A proof is reported in  Appendix~\ref{a:extra}.
\end{proofline}
\begin{quotation}
\begin{remark}
If the conditions of Lemma~\ref{l:leftrightexp} hold for $k=\infty$ then $f^{R/L}_\pm(\nu)$ exhibits the asymptotic series
\begin{equation}
f^{R/L}_\pm(\nu)\sim e^{(\log f)^{\pm}}\sum_{j=0}^\infty(\psi_j^{R/L}[\nabla\log f])^{\pm}\frac{\nu^j}{j!} \ .
\end{equation}
\end{remark}
\end{quotation}

\phantom{a}

\begin{definition}
Given two functions $f,g$ of two arguments, their Poisson bracket is defined as follows 
\begin{equation}
    \{ f, g\}:(y,p)\mapsto\frac{\partial f(y,p)}{\partial y}\frac{\partial g(y,p)}{\partial p}-\frac{\partial g(y,p)}{\partial y}\frac{\partial f(y,p)}{\partial p} \  .
\end{equation}
\end{definition}

\begin{corollary} \label{lemma truncated moyal bound}
Let $f$ be a function as in Lemma~\ref{l:leftrightexp} for $k=3$, such that $g=\log f$ is analytic on the domain of $f$. 
Then, for any compact $\Omega\subset \mathbb R$, we have
\begin{multline}\label{eq:logpHplogmH}
\log f^{H}_+(\nu)+\log f^{H}_-(\nu)=\textcolor{blue}{g}+ \frac{i}{2}\sigma_H \nu \{g^-,g^+\}\\
+\frac{\nu^2}{8}\Bigl(2\{g^+ ,\{g^-,g^+\}^-\}-
2\{g^- ,\{g^-,g^+\}^+\}-\partial_p g^-\{g^-,\partial_y g^+\}+\partial_y g^-\{g^-,\partial_p g^+\}\\
-\partial_y g^+\{\partial_p g^-, g^+\}+\partial_p g^+\{\partial_y g^-,g^+\}-\{\partial_p g^-,\partial_y g^+\}+\{\partial_y g^-,\partial_p g^+\}\Bigr)+ O(\nu^{3})
\end{multline}
asymptotically as $\nu\to 0$,  uniformly in $\tfrac{\Omega}{\nu}\times (\mathbb R/(2\pi\mathbb Z))$.
\end{corollary}

\begin{proofline}
This can be obtained by expanding the logarithm of $f_\pm^H(\nu)$, as defined in \eqref{eq:leftrightsmooth}, at the second order in $\nu$, and using \eqref{eq:inversion} with $m=0,1,2$, which can be explicitly written as follows: 
 \begin{align}
\psi^{H}_0[\nabla g]&=1\\
\psi^{H}_1[\nabla g]&=\frac{i}{2}\sigma^H \{g^-,g^+\} \\
\psi^{H}_2[\nabla g]
&=
\frac{1}{2}\{g^-,g^+\}^-\{g^-,g^+\}^+-\frac{1}{4}\{g^-,g^+\}^2
+\frac{1}{2}\{g^+ ,\{g^-,g^+\}^-\}-
\frac{1}{2}\{g^- ,\{g^-,g^+\}_+\}\nonumber \\&
-\frac{1}{4}(\partial_p g^-\{g^-,\partial_y g^+\}-\partial_y g^-\{g^-,\partial_p g^+\}+\partial_y g^+\{\partial_p g^-, g^+\}-\partial_p g^+\{\partial_y g^-,g^+\})\nonumber \\
&-\frac{1}{4}(\{\partial_p g^-,\partial_y g^+\}-\{\partial_y g^-,\partial_p g^+\})\, .
\end{align}

\end{proofline}

\subsection{Truncated star functions}\label{ss:starfunctions}

The main ingredient for the calculation of star functions is the star inverse of $\zeta-a$ with $\zeta\in \mathbb C$. This problem can be solved asymptotically just as the Wiener-Hopf star factorisations, indeed the starting point is again an identity for a star product 
\begin{equation}\label{eq:identity}
(\zeta-a)\star (\zeta-a)^{-1_\star}=1\, .
\end{equation}
Under the assumptions of  Lemma~\ref{lemma truncated moyal analytic}, we can  truncate the Moyal product and follow the lines of the proof of Lemma~\ref{l:leftrightexp}. Note that this inversion is simpler than for the Wiener-Hopf star factorisations, as one of the factors of the star product \eqref{eq:identity} is known and there is no need to extract $\pm$ parts of symbols. We start by defining the auxiliary functions that play the role of those in definition~\ref{d:psi}.
\begin{definition}\label{d:R}
Let $f$ be a function of two arguments. For $m\in\mathbb{N}$ we define the functions $R_m^{(\zeta)}[f]$ of $f$ and derivatives by the recursive equation
\begin{equation}\label{eq:RaK}
\begin{aligned}
R^{(\zeta)}_0[f]:&\vec \varphi \mapsto \frac{1}{\zeta-f(\vec\varphi)}\\
R^{(\zeta)}_{m}[f]:&\vec \varphi \mapsto \frac{1}{\zeta-f(\vec\varphi)}\sum_{j=0}^{m-1}\binom{2m}{2j}(-1)^{m-j}\left.\Bigl(\frac{\partial_{\varphi_1}\partial_{\varphi_2'}-\partial_{\varphi_2}\partial_{\varphi_1'}}{2 }\Bigr)^{2(m-j)} f(\vec\varphi')R^{(\zeta)}_j[f](\vec\varphi)\right|_{\vec\varphi'=\vec\varphi}
\end{aligned}
\end{equation}
\end{definition}
\begin{lemma}\label{lemma star inverse}
Let $f:\mathbb R \times \strp[\rho] \to \mathbb{C}$ 
be of class $C^{2(k+1)}$, for some $k\in\mathbb{N}_0$, with respect to the first argument so that $\partial_1^jf$, for $j=0,1,\ldots,2(k+1)$, are analytic in the second argument.
Let $\zeta$ be a complex number $\notin \mathrm{ran} f$ and $\Omega$ be a compact subset of $\mathbb R$. 
If $\zeta-f_{s\nu}$, with $f_{s\nu}: (x,p)\mapsto f(xs\nu,p)$ and $s\in[0,1]$, parametrises a star-Toeplitz flow of star-invertible symbols in $V_\rho$ connecting $\zeta-a=[\zeta-f_\nu]$ with the Toeplitz symbol $e^{i p}\mapsto \zeta-f(0,p)$, then
\begin{equation}
(\zeta-a)^{-1_\star}=[R^{(\zeta,f)}_\nu]
\end{equation}
where
$R^{(\zeta,f)}_\nu:(x,p)\mapsto (R^{(\zeta,f)}(\nu))(x\nu,p)$ and  $R^{(\zeta,f)}(\nu):\mathbb R\times (\mathbb R/(2\pi\mathbb Z))\rightarrow\mathbb C$ exhibits the  asymptotic behaviour
\begin{equation}
   R^{(\zeta,f)}(\nu)=\sum_{j=0}^{k}\frac{R_j^{(\zeta)}[f]}{(2j)!}\nu^{2j}+ O(\nu^{2k+2}) 
\end{equation}
as $\nu\rightarrow 0$, uniformly in $\Omega \times (\mathbb R/(2\pi\mathbb Z))$. 
In particular, we have
\begin{equation}\label{eq:expinverse}
R^{(\zeta,f)}_\nu=
\begin{cases}
\frac{1}{\zeta-f_\nu}+O(\nu^2)&k=0\\
\frac{1}{\zeta-f_\nu}+\frac{1}{4}\frac{(\partial_{1,2}f_\nu)^2-\partial_{1}^2f_\nu\partial_{2}^2f_\nu}{(\zeta-f_\nu)^3}-\frac{1}{4}\frac{\partial_2^2 f_\nu (\partial_1 f_\nu)^2+\partial_1^2 f_\nu (\partial_2 f_\nu)^2-2\partial_1 f_\nu\partial_2 f_\nu \partial_{1,2}f_\nu}{(\zeta-f_\nu)^4}+O(\nu^{4})&k=1
\end{cases}
\end{equation}
\end{lemma}
\begin{proofline}
The proof is analogous to that of Lemma~\ref{l:leftrightexp}. Just note that, in this case, the expansion of the Moyal product  can be restricted to even powers (even $j$ in \eqref{eq:truncationstar}) as $\zeta-a$ star commutes with its star inverse. 
\end{proofline}

\begin{quotation}
\begin{remark}
By construction, if the conditions in Lemma~\ref{lemma star inverse} hold for $k=\infty$,  then 
\begin{equation}\label{eq:Rasympt}
   R^{(\zeta,f)}(\nu)\sim \sum_{j=0}^{\infty}\frac{R_j^{(\zeta)}[f]}{(2j)!}\nu^{2j}\, ,
\end{equation}
asymptotically as $\nu\rightarrow 0$.
\end{remark}
\end{quotation}

\phantom{a}

\begin{corollary}\label{c:starfun}
Let $f$ be a function on $\mathbb R\times \strp[\rho]$ that is of class $C^{2(k+1)}$, for some $k\in\mathbb{N}_0$, so that $\partial_1^j f$, for $j=0,1,\ldots,2(k+1)$ are analytic in the second argument. Let
$f_{s\nu}: (x,p)\mapsto f(xs\nu,p)$,  with $s\in[0,1]$,  parametrise a star-Toeplitz flow in $\anb$, and the symbol $a$ be represented by $f_\nu$. If $F:\mathbb C\rightarrow\mathbb C$ is a function analytic inside a closed simple curve $C$ of the complex plane strictly surrounding the spectrum of $L(a')$ for any $a'$ in the orbit of the flow, then for any compact $\Omega\in\mathbb R$ the star function $F_\star(a)$ is represented by a function $ F_\star[f_\nu]$ of $f_\nu$ and derivatives with  the asymptotic behaviour
\begin{equation}
F_\star[f_\nu]=\begin{cases}
    F(f_\nu)+O(\nu^2)&k=0\\
    F(f_\nu)+F''(f_\nu)\frac{(\partial_{1,2}f_\nu)^2-\partial_{1}^2f_\nu\partial_{2}^2f_\nu}{8}\\
    \qquad\qquad -F'''(f_\nu)\frac{\partial_2^2 f_\nu (\partial_1 f_\nu)^2+\partial_1^2 f_\nu (\partial_2 f_\nu)^2-2\partial_1 f_\nu\partial_2 f_\nu \partial_{1,2}f_\nu}{24}+O(\nu^{4})&k=1\, ,
    \end{cases}
\end{equation}
as $\nu\rightarrow 0$, uniformly in $\tfrac{\Omega}{\nu}\times (\mathbb R/(2\pi\mathbb Z))$.
\end{corollary}
\begin{proofline}
This follows from plugging \eqref{eq:expinverse} into the Cauchy integral formula \eqref{star function definition}.
\end{proofline}
\begin{proposition}\label{p:asymplogstar}
With the same notations of Corollary~\ref{c:starfun}, let $f:\mathbb R \times S_{\rho}^{(2\pi)}\rightarrow\mathbb C$ be 
smooth with respect to the first argument and analytic with respect to the second one,
% either of the arguments, 
the flow be with zero star winding number, and $F$ be the logarithm. Then, the asymptotic expansion induced by \eqref{eq:Rasympt} can be written as 
\begin{equation}\label{eq:asymplogstar}
\log_\star a\sim \left[g+\sum_{j=1}^{\infty}\frac{L_j[g]}{(2j)!}\right]\, ,\qquad\text{with}\quad L_j[g]=\frac{\nu^{2j} }{2\pi i}\oint_C \log \zeta \ R_j^{(\zeta)}[e^g] d\zeta\, ,
\end{equation}
$g=\log f_\nu$,  and 
\begin{enumerate}
\item $L_j[g]$ is a multivariate polynomial of the derivatives of $g$  of total degree $3j$.
\item in each monomial of $L_j[g]$, the derivative with respect to each argument appears exactly $2j$ times.
\end{enumerate}
\end{proposition}
\begin{proofline}
The second property is true by construction:  $L_j[g]$ is proportional to $R_j^{(\zeta)}[f]$, which is written in terms of $2j$ derivatives with respect to the first argument and $2j$ derivatives with respect to the second one. The first property is a direct consequence of the following observations
\begin{itemize}
\item[-] for every $j>0$, $R_j^{(\zeta)}[f]$ is a linear combination of terms of the form $\frac{\overbrace{(\partial_{\dots} f)\cdots (\partial_{\dots} f)}^{n-1\text{  factors}}}{(\zeta-f)^{n}}$, with $n> 1$ and $\partial_{\dots} f$ standing for $\partial_1^{\ell_1}\partial_2^{\ell_2}f$ for some $\ell_1,\ell_2\geq 0$;
\item[-] the contour integral in $\zeta$ of each term is a multivariate polynomial in the derivatives of $g$;
\item[-] the monomial with the largest degree corresponds to the term in which all the derivatives to $R^{(\zeta)}_j[g]$ in the second line of \eqref{eq:RaK} are applied to the poles $\frac{1}{(\zeta-f)^{n}}$.
\end{itemize}
By solving the recursive equation induced by \eqref{eq:RaK} to the total degree of the polynomial, we readily find that $L_j[g]$ has total degree $3j$. 
\end{proofline}
\begin{quotation}
\begin{remark}
If $f$ is of class $C^4$, we have 
\begin{equation}\label{eq:logstara4}
\log_\star a=\left[g-\frac{2\partial_1 g\partial_2 g \partial_{1,2} g+3(\partial_{1,2} g)^2-(\partial_2 g)^2 \partial_{1}^2 g-(\partial_1 g)^2\partial_{2}^2 g-3\partial^2_1 g\partial^2_2 g}{24}+O(\nu^{4})\right]\, .
\end{equation}
\end{remark}
\end{quotation}

\subsection{Asymptotic expansion---proof of Theorem~\ref{t:main}}

In this section we put all pieces together and work out the first terms of the asymptotic expansion of \eqref{eq:Ansatz_st1} in the limit of small $\nu$. 
We start by showing that there is a gauge in which \eqref{eq:gauge} holds. 
\begin{definition}[Bulk function]
We say that an infinitely differentiable function $f$ is a bulk function in its domain if both the function and its derivatives vanish approaching the boundaries of the domain. 
\end{definition}
\begin{definition}[Boundary functional]
We say that  a functional $F[g]$ of an infinitely differentiable function $g$ is a boundary functional if, for every bulk function $\phi$,  
\begin{equation}
\lim_{\epsilon\rightarrow 0}\frac{F[g+\epsilon \phi]-F[g]}{\epsilon}=0
\end{equation}
\end{definition}
\begin{lemma}\label{l:gauge}
Let $g$ be infinitely differentiable in $(\frac{1}{2},n+\frac{1}{2})\times (\mathbb R/(2\pi \mathbb Z))$ and $F[g]$ be a functional of the form
\begin{equation}\label{eq:polynomial}
F[g]=\frac{1}{2\pi}\int_{\frac{1}{2}}^{n+\frac{1}{2}}\int_{-\pi}^\pi \left(\prod_{i=1}^k\partial_{1}^{c_i}\partial_{2}^{d_i}g(x,p)\right) dp dx
\end{equation}
with $c_i,d_i\in\mathbb N_0$. Let $\{g_m\}$ be a sequence of bulk functions converging to $g$ almost everywhere in $(\frac{1}{2},n+\frac{1}{2})\times (\mathbb R/(2\pi \mathbb Z))$. Let $G[g]$ be the functional
\begin{equation}\label{eq:polynomialbulk}
G[g]=\int_{-\infty}^0 \bar \delta_g  F[e^s g](\{g_m\}) ds
\end{equation}
where
\begin{equation}
\bar \delta_g  F[g]:\{\phi_m\}\mapsto \lim_{m\rightarrow\infty}\lim_{\epsilon\rightarrow 0}\frac{ F[g+\epsilon \phi_m]- F[g]}{\epsilon}\, .
\end{equation}
Then
\begin{enumerate}
\item for any sequence $\{\phi_m\}$ of bulk functions 
converging almost everywhere to some infinitely differentiable function $\phi$, $\bar \delta_g  F[g](\{\phi_m\})$ depends on the sequences only through $\phi$: $\bar \delta_g  F[g](\{\phi_m\})\equiv \delta_g  F[g](\phi)$
\item the functional $G[g]$ is given by
\begin{equation}\label{eq:Gproof}
G[g]=\frac{1}{2\pi}\int_{\frac{1}{2}}^{n+\frac{1}{2}}\int_{-\pi}^\pi g(x,p) \frac{1}{k}\sum_{j=1}^k(-\partial_{1})^{c_j}(-\partial_{2})^{d_j}\left(\prod_{i\neq j}^{k}\partial_{1}^{c_i}\partial_{2}^{d_i}g(x,p)\right) dp dx
\end{equation}
\item for any bulk function $\phi$,
$
 \delta_g  G[g](\phi)= \delta_g  F[g](\phi)
$
\end{enumerate}
\end{lemma}
\begin{proofline}
A proof is reported in  Appendix~\ref{a:extra}.
\end{proofline}

\begin{corollary}\label{c:gaugefixing}
A functional $F[g]$ as in \eqref{eq:polynomial} can be decomposed in the sum of a bulk contribution, captured by $G[g]$ and given by  \eqref{eq:Gproof}, and boundary contributions written in terms of $g$ and derivatives evaluated at $\frac{1}{2}$ and $n+\frac{1}{2}$. Specifically, we have 
\begin{equation}\label{eq:boundaryFG}
%\textcolor{red}{W[g]:=}
\Delta[g]:=F[g]-G[g]=\left.\frac{1}{2\pi}\int_{-\pi}^\pi \frac{1}{k}\sum_{j=1}^k \sum_{\eta=0}^{c_j-1}\partial_1^\eta\partial_{2}^{d_j} g(x,p) (-\partial_{1})^{c_j-\eta-1}\left(\prod_{i\neq j}^{k}\partial_{1}^{c_i}\partial_{2}^{d_i}g(x,p)\right) dp \right|_{x=\frac{1}{2}}^{x=n+\frac{1}{2}}
\end{equation}
\end{corollary}
\begin{proofline}
The lemma shows that $F[g]-G[g]$ is a boundary contribution, hence the main statement. To compute the boundary contributions, it is  sufficient to integrate \eqref{eq:Gproof} by parts $c_j$ times  with respect to $x$ and $d_j$ times with respect to $p$. The boundary terms of the integrations by parts with respect to $p$ vanish by periodicity, whereas the other integrations by parts give \eqref{eq:boundaryFG}.
\end{proofline}

\begin{lemma}[Euler-Maclaurin] \label{l:EM}Let $g$ be a function of class $C^{2k}$, then we have
\begin{multline}
\sum_{j=1}^{n} g(j \nu )=\frac{1}{\nu}\int_{\frac{\nu}{2}}^{(n+\frac{1}{2})\nu}\left[ g(y)-\sum_{j=1}^{k}\frac{\nu^{2j}}{(2j)!}(1-2^{1-2j})B_{2j}(1)g^{(2j)}(y)\right]\, dy\\
-\frac{\nu^{2k}}{(2k)!}\int_{-\frac{1}{2}}^{\frac{1}{2}} B_{2k}(t-[t])\left\{\sum_{j=1}^{n}g^{(2k)}(\nu (j-t))\right\}d t
\end{multline}
where $B_{2j}$ are the Bernoulli polynomials.
\end{lemma}
\begin{quotation}
\begin{remark}
Let $g$ be an infinitely differentiable function, then the following asymptotic formula holds
\begin{equation}\label{eq:EMformula}
\sum_{j=1}^{n} g(j \nu )\sim \frac{1}{\nu}\int_{\frac{\nu}{2}}^{(n+\frac{1}{2})\nu}\left[ g(y)-\sum_{j=1}^{\infty}\frac{\nu^{2j}}{(2j)!}(1-2^{1-2j})B_{2j}(1)g^{(2j)}(y)\right]\, dy
\end{equation}
as $\nu\rightarrow 0$.
\end{remark}
\end{quotation}

\phantom{a}
\begin{lemma}
Let the symbol be $a=[f_\nu]$ with  $f_\nu(x,p)=f(x\nu,p)$ and  $f$ be a function on $\mathbb R\times \strp[\rho]$. Suppose that $a\in\anb$ has zero star winding number and denote $b_\pm^{R/L}=\log_\star a_\pm^{R/L}$. We define the functions $g^{R/L}_\nu(x,p)=(g^{R/L}(\nu))(x\nu,p)$ and $\tilde h^{R/L}_\nu(x,p)=(\tilde h^{R/L}(\nu))(x\nu,p)$ in such a way  that
\begin{equation}
\begin{gathered}
b_+^{R/L}+b_-^{R/L}=[g^{R/L}_\nu]\\
(\Phi_{BCH}(b_-^R,b_+^R))^++(\Phi_{BCH}(-b_+^R,-b_-^R))^-=[\tilde h^R_\nu]\\
-(\Phi_{BCH}(-b_-^L,-b_+^L))^+-(\Phi_{BCH}(b_+^L,b_-^L))^-=[\tilde h^L_\nu]\, .
\end{gathered}
\end{equation}
Suppose $g^{R/L}(\nu)$ and $\tilde h^{R/L}(\nu)$ be $C^{2k}$ in their first argument and that the functions and their first $2k$ derivatives (with respect to the first argument) are analytic with respect to the second argument. Then, the boundary terms $E_x^{R/L}$ defined in \eqref{eq:ERL} can be approximated with an error $O(\nu^{2k})$ by the partial sum of the first $2k-2$ terms of the following asymptotic series as $\nu\rightarrow 0$
\begin{equation}\label{eq:ERLasympt}
E_x^{R/L}\sim \frac{1}{4\pi}\int_{-\pi}^\pi   \left.\frac{ \sin(\nu \frac{\partial_{y,p'}}{2})}{\sin(\nu \frac{i \partial_y}{2})}(g^{R/L}(\nu))(y,p)(\tilde h^{R/L}(\nu))(y,p')\right|_{y=x\nu\atop p'=p}\, dp\, .
\end{equation}
In particular, for $k=2$ we have
\begin{equation}\label{eq:ERLasympt1}
E_x^{R/L}=
-\frac{i}{2}(g^{R/L}(\nu)\partial_2\tilde h^{R/L}(\nu))_0+\frac{i \nu^2}{48}\partial_1^2\left( g^{R/L}(\nu)(\partial_2+\partial_2^3)\tilde h^{R/L}(\nu)\right)_0+O(\nu^4)
\end{equation}
as $\nu\rightarrow 0$. 
\end{lemma}

\begin{proofline}

Formally, the result is obtained by rewriting Eq.~\eqref{eq:ERL} as
\begin{equation}
\begin{split}
& E_x^{R/L} =\\ &\frac{1}{4\pi^2}\sum_{m=1}^\infty\sum_{j=1}^m \int_{-\pi}^{\pi}\int_{-\pi}^{\pi}\Big[ g_\nu^{R/L}(x+j-\tfrac{m+1}{2},p)\tilde h_\nu^{R/L}(x+j-\tfrac{m+1}{2},p')(e^{im(p-p')}-e^{im(p'-p)})\Big]dp dp'  ,
\end{split}
\end{equation}
we can then Taylor expand around $x$
\begin{equation}
    g_\nu^{R/L}(x+j-\tfrac{m+1}{2},p)\tilde h_\nu^{R/L}(x+j-\tfrac{m+1}{2},p')=e^{(j-\frac{m+1}{2})\partial_x} [g^{R/L}_\nu(x,p)\tilde h_\nu^{R/L}(x,p')]
\end{equation}
and use the identity
\begin{equation}
    \sum_{j=1}^m e^{(j-\frac{m+1}{2})\partial_x}=\frac{\sinh(\frac{m}{2}\partial_x)}{\sinh(\frac{1}{2}\partial_x)} \ .
\end{equation}
The manipulations made in order to obtain the final result can be justified as in the proof of Lemma~\ref{lemma truncated moyal analytic}. 
\end{proofline}

\begin{proof}[Proof of Theorem~\ref{t:main}]
We start by proving the properties of the functions $\mathcal D^{(j)}$ and $\mathcal C^{(j)}$. First of all we note that bulk contributions come only from the sum of $\log_\star a$ in \eqref{eq:strongSzego}. We extract them through the following procedure:
\begin{itemize}
\item[-] We replace $\log_\star a$ by the asymptotic expansion shown in  \eqref{eq:asymplogstar}
\begin{equation}
\log_\star a\rightarrow g+\sum_{j=1}^{\infty}\frac{L_j[g]}{(2j)!}\, .
\end{equation}
\item[-] We apply the Euler-Maclaurin formula \eqref{eq:EMformula} to the sum and transform it into an integral.
%\item[-] 
Only the first term in the expansion captures bulk contributions, so we can focus on that term
% and replace $\log_\star a$ by the asymptotic expansion in  \eqref{eq:asymplogstar}. 
\begin{equation}\label{eq:fromsumlogtoint}
\sum_{j=1}^n((\log_\star a)_j)_0\rightarrow \frac{1}{2\pi}\int_{-\pi}^\pi \int_{\tfrac{1}{2}}^{n+\tfrac{1}{2}}\left[g(x,p)+\sum_{j=1}^{\infty}\frac{L_j[g](x,p)}{(2j)!}\right] dx dp\, .
\end{equation}
\item[-] We fix the gauge using the convention of Corollary \ref{c:gaugefixing}.

The first orders of the bulk terms exhibited in \eqref{eq:asympt2L}  are obtained by  keeping only $g$ and the term proportional to  $L_1[g]$ on the right hand side of \eqref{eq:fromsumlogtoint} and fixing the gauge of the monomials in $L_1[g]$ (i.e., the terms appearing in \eqref{eq:logstara4}), which are explicitly given by (cf.~\eqref{eq:Gproof})
\begin{equation}
\begin{aligned}
    & G[g]= \frac{1}{2\pi}\int_{\tfrac{1}{2}}^{n+\tfrac{1}{2}}\int_{-\pi}^\pi  g(x,p)\,  dp dx \\
    & G[\partial_1 g\partial_2 g \partial_{1,2} g]= \frac{1}{6\pi}\int_{\tfrac{1}{2}}^{n+\tfrac{1}{2}}\int_{-\pi}^\pi  g(x,p)\left(\partial_x^2 g(x,p) \partial_p^2 g(x,p)-(\partial_{x,p} g(x,p))^2 \right)\, dp dx \\
    & G[(\partial_{1,2} g)^2]= \frac{1}{2\pi}\int_{\tfrac{1}{2}}^{n+\tfrac{1}{2}}\int_{-\pi}^\pi  g(x,p)\partial_x^2\partial_p^2 g(x,p)\,  dp dx \\
    & G[(\partial_{2} g)^2 \partial_1^2 g]=  G[(\partial_{1} g)^2 \partial_2^2 g]=  \frac{1}{3\pi}\int_{\tfrac{1}{2}}^{n+\tfrac{1}{2}}\int_{-\pi}^\pi  g(x,p)\left((\partial_{x,p} g(x,p))^2-\partial_{x}^2 g(x,p) \partial_{p}^2 g(x,p) \right)\, dp dx\\
    & G[\partial_{1}^2 g \partial_2^2 g]=  \frac{1}{2\pi}\int_{\tfrac{1}{2}}^{n+\tfrac{1}{2}}\int_{-\pi}^\pi  g(x,p)\partial_x^2\partial^2_p g(x,p)\, dp dx
\end{aligned}
\end{equation} 
\end{itemize}
The first two properties highlighted in the theorem follow from Proposition~\ref{p:asymplogstar}. 

The boundary contributions have different origins. Some of them come from the Euler-Maclaurin formula and correspond to derivatives of  $(\log_\star a)_x(e^{ip})$ with respect to $x$, which therefore are characterised by a lower number of derivatives with respect to $p$ rather than $x$. Other boundary contributions come from the asymptotic series describing the star logarithm $b_\pm$ of the factors of the Wiener-Hopf star factorisations as well as the associated functions $\Phi_{BCH}(b_-,b_+)$ and $\Phi_{BCH}(b_+,b_-)$ emerging from the Baker-Campbell-Hausdorff formula \eqref{d:PhiBCH}. Independently of how complicated such terms are, they are always characterised by the same number of derivatives with respect to the two variables. They enter the formula, however, through $E_x^{R/L}$, which we  replace by the asymptotic series~\eqref{eq:ERLasympt} and is characterised by one extra derivative with respect to $p$. This substantiates the third property of the theorem. 

%The first orders of the bulk terms exhibited in \eqref{eq:asympt2L} readily follow from applying \eqref{eq:Gproof} to \eqref{eq:logstara4}.
Specifically, the boundary terms exhibited in \eqref{eq:asympt2B} are %instead 
obtained by 
\begin{itemize}
\item[-] taking the first correction to the Euler-Maclaurin formula ($j=1$ in \eqref{eq:EMformula}) for the sum of $\log_\star a$
%\begin{equation}
%\sum_{j=1}^n f(j)- \int_{\tfrac{1}{2}}^{n+\tfrac{1}{2}}f(x)\rightarrow -\frac{f'(x)}{24}\Big|_{x=1/2}^{x=n+\tfrac{1}{2}}\, .
%\end{equation}
\begin{multline}
\sum_{j=1}^n ((\log_\star a)_j)_0- \frac{1}{2\pi}\int_{-\pi}^\pi \int_{\tfrac{1}{2}}^{n+\tfrac{1}{2}}\left[g(x,p)+\tfrac{1}{2}L_1[g](x,p)\right] d x dp=\\
-\frac{1}{48\pi}\int_{-\pi}^\pi \partial_xg(x,p)\, dp\Big|_{x=1/2}^{x=n+\tfrac{1}{2}}+O(n\nu^4+\nu^2)\, ;
\end{multline}
\item[-] collecting the boundary terms coming from the asymptotic expansion of $\log_\star a$, which are obtained by applying \eqref{eq:boundaryFG} to \eqref{eq:logstara4}
and simplifying irrelevant total derivatives with respect to $p$
\begin{equation}
\begin{aligned}
    & \Delta[g]=\Delta[\partial_{1}^2 g \partial_2^2 g]=0 \\
    & -\frac{1}{12}\Delta[\partial_1 g\partial_2 g \partial_{1,2} g]=\left.\frac{1}{36}\frac{1}{2\pi}\int_{-\pi}^\pi g(x,p)\, \partial_xg(x,p)\,  \partial_p^2g(x,p)\, dp \right|_{x=1/2}^{x=n+\tfrac{1}{2}}\\
   & -\frac{1}{8}\Delta[(\partial_{1,2} g)^2]=\left.\frac{1}{8}\frac{1}{2\pi}\int_{-\pi}^\pi g(x,p)\,  \partial_x\partial_p^2g(x,p) \ dp \right|_{x=1/2}^{x=n+\tfrac{1}{2}}\\
  & \frac{1}{24}\Delta[(\partial_{2} g)^2 \partial_1^2 g]=\frac{1}{24}\left.\frac{1}{2\pi}\int_{-\pi}^\pi \left(\partial_x g(x,p)\, (\partial_p g(x,p))^2+\frac{2}{3}g(x,p)\, \partial_xg(x,p) \, \partial_p^2g(x,p)\right) dp \right|_{x=1/2}^{x=n+\tfrac{1}{2}}\\ 
  & \frac{1}{24}\Delta[(\partial_{1} g)^2 \partial_2^2 g]=\frac{1}{36}\left.\frac{1}{2\pi}\int_{-\pi}^\pi g(x,p)\, \partial_xg(x,p) \, \partial_p^2g(x,p) \ dp \right|_{x=1/2}^{x=n+\tfrac{1}{2}}
\end{aligned}
\end{equation} 
\item[-] truncating the asymptotic expansion of the Wiener-Hopf star factorisation to first order in $\nu$, i.e., keeping only the first line of \eqref{eq:logpHplogmH}:
\begin{equation}
\begin{aligned}
a^{R}_s= \left[e^{(g+ \frac{i}{2}  \{g^-,g^+\})^s+O(\nu^2)}\right]&&&&
a^{L}_s= \left[e^{(g- \frac{i}{2}  \{g^-,g^+\})^s+O(\nu^2)}\right]
\end{aligned}
\end{equation}
with $s=\pm$.
\item[-] approximating the star logarithm of the factors by the logarithm (since the correction is $O(\nu^2)$---cf.~\eqref{eq:asymplogstar})
\begin{equation}
\begin{aligned}
\log_\star a^{R}_s= \left[(g+ \frac{i}{2}  \{g^-,g^+\})^s+O(\nu^2)\right]&&&&
\log_\star a^{L}_s= \left[(g- \frac{i}{2}  \{g^-,g^+\})^s+O(\nu^2)\right]
\end{aligned}
\end{equation}
\item[-] truncating the expansion of $\Phi_{BCH}$ as in \eqref{eq:Phi} to $O(\nu)$, which gives
%\begin{equation}
%\Phi_{BCH}(c,d)\rightarrow  d+\frac{1}{3}i\{c,d\}_M+O(\nu^2)\, .
%\end{equation}
%Thus we have
\begin{equation}
\begin{aligned}
&\Phi_{BCH}(\log_\star a_-^R,\log_\star a_+^R)= g^++\frac{1}{6}i(\{g^-,g^+\})^+-\frac{1}{3}i(\{g^-,g^+\})^-+O(\nu^2)\\
&\Phi_{BCH}(-\log_\star a_+^R,-\log_\star a_-^R)= - g^-- \frac{1}{6}i(\{g^-,g^+\})^-+\frac{1}{3}i(\{g^-,g^+\})^++O(\nu^2)\\
&\Phi_{BCH}(\log_\star a_+^L,\log_\star a_-^L)= g^--\frac{1}{6}i(\{g^-,g^+\})^-+\frac{1}{3}i(\{g^-,g^+\})^++O(\nu^2)\\
&\Phi_{BCH}(-\log_\star a_-^L,-\log_\star a_+^L)= - g^++\frac{1}{6}i(\{g^-,g^+\})^--\frac{1}{3}i(\{g^-,g^+\})^++O(\nu^2)
\end{aligned}
\end{equation}
\item[-] taking the leading order in  the expansion~\eqref{eq:ERLasympt} of $E^{R/L}_x$
\begin{equation}
E^{R/L}_x= 
-\frac{i}{2}\frac{1}{2\pi}\int_{-\pi}^\pi g^{R/L}(x,p)\partial_p\widetilde{h^{R/L}}(x,p)\ dp+O(\nu^2)
\end{equation}
with
\begin{equation}
\begin{aligned}
g^{R/L}=&g\pm \frac{1}{2}i  \{g^-,g^+\}+O(\nu^2)\\
h^{R/L}=&g\pm \frac{1}{6}i\{g^-,g^+\}+O(\nu^2)\, .
\end{aligned}
\end{equation}
This gives
\begin{equation}
E^{R/L}_x= 
-\frac{i}{2}(g\partial_p \tilde{g})_0\pm\frac{1}{6}(\{g,\tilde g\}\partial_p \tilde g)_0+O(\nu^2)
\end{equation}
%\item[-] expanding the star logarithm to leading order (i.e., replacing it with a conventional logarithm) for the factors of the factorisation.
\end{itemize}
Collecting all the terms gives \eqref{eq:asympt2B}, i.e.,
\begin{multline}
\mathcal C_{\pm}\sim-\frac{i}{4}(g\partial_2 \tilde{g})_0\\
\pm\frac{1}{24}\left(2\partial_2 g\partial_1\tilde g\partial_2\tilde g+\partial_1 g\Bigl(- 1+3\partial_2^2 g+ 2 g \partial_2^2g+ (\partial_2 g)^2- 2(\partial_2\tilde g)^2\Bigr)\right)_0+O(\nu^2)
\end{multline}
where the equivalence is up to irrelevant total derivatives with respect to $p$ (the corner contributions \eqref{eq:Ansatz_st1} are the integrals of $\mathcal C_{\pm}$ over $p$).

\end{proof}

\section{Discussion}
We conclude by pointing out some physical applications of the formulas that we derived. We have already highlighted in Section~\ref{ss:blockToeplitz} that the determinant formula provided by Theorem~\ref{c:strong} can be used to infer an explicit representation of the constant contribution to the determinant of a block-Toeplitz matrix, which appears, for example, in the calculation of correlation functions in translationally invariant quantum spin models. Unfortunately, however, the formula relies on the knowledge of both left and right factorizations, which is often beyond reach.
In inhomoegeneous physical settings, instead,  Theorem~\ref{t:main} provides the formulas with more direct applications. The asymptotic expansion that we worked out can indeed be applied, for example, to the calculation of correlation functions and entanglement entropies in systems prepared in smooth traps both in and out of equilibrium. Arguably, some derivations could also be adapted to capture more complicated bipartitioning protocols,  generalising, for example, the results of Ref.~\cite{Bocini2023Connected} to operators with semilocal fermionic representations. 

The strong assumptions we made constitute, however, an obstacle to applications. While it is reasonable to expect that symbols with a nonzero star winding number could be easily accommodated in the theory, many interesting physical problems are characterised by singularities that would require generalisations of the Fisher-Hartwig formulas which go far beyond the plain framework of this work. In that respect, even if the course of action is not obvious, we believe that symbol regularisation could still be the key to obtain explicit expressions.

\backmatter

\bmhead{Acknowledgements}

We thank Luca Capizzi and Giuseppe Del Vecchio Del Vecchio for discussions. %\textcolor{red}{We thank Estelle Basor for correspondence and providing comments on the manuscript.}
We are grateful to Estelle Basor for her correspondence and valuable feedback on the manuscript.
This work was supported by the European Research Council under the Starting Grant No. 805252 LoCoMacro.

\bmhead{Data Availability} Data sharing is not applicable to this article as no datasets were generated or analyzed
during the study.

\bmhead{Conflict of Interest}
The authors declare that there are no conflicts of interest.

\begin{appendices}

\section{Basic results}\label{a:lemmas}

\begin{lemma}\label{lemma bound laurent coefficients}
If $a\in\anb[\rho]$ then for any $\rho_1\in(\rho,1)$ we have
\begin{equation}
    |(a_x)_j| \leq M\rho_1^{|j|} \ ,
\end{equation}
where $M=\sup\|a_x(z)|: \ x\in\frac{1}{2}\mathbb{Z}, \rho_1\leq|z|\leq \rho_1^{-1} \} $.

\end{lemma}

\begin{proofline}
This is an elementary claim. Using analyticity we can deform the integration contour from $|z|=1$ to $|z|=\rho'$ for any $\rho'\in(\rho,\rho^{-1})$ giving
\begin{equation}
    |(a_x)_j|=\left|\oint_{|z|=\rho'}a_{x}(z) z^{-j} \frac{dz}{2\pi i z}\right|\leq (\rho')^{-j} \sup_{|z|=\rho'}|a_{x}(z)| \; .
\end{equation}
The claim follows by choosing $\rho'=\rho_1$ for $j<0$ and $\rho'=\rho_1^{-1}$ for $j>0$ (and say $\rho'=1$ for $j=0$) and using the uniform boundedness in $x$ in the annulus.
\end{proofline}

\begin{lemma}\label{lemma bound norm infinite matrix}
Let $A=\left(A_{j,k}\right)_{j,k\in\mathbb{Z}}$ be an infinite matrix. The matrix defines a bounded linear operator on $\ell^2(\mathbb{Z})$ with the norm satisfying
\begin{equation}
    \| A\|\leq \sqrt{\sup\limits_{j\in\mathbb{Z}}\left(\sum_{k\in\mathbb{Z}}|A_{j,k}|\right) \sup\limits_{k'\in\mathbb{Z}}\left(\sum_{j'\in\mathbb{Z}}|A_{j',k'}|\right) },
\end{equation}
provided this bound is finite.
\end{lemma}
\begin{proofline}
Denote $K_1=\sup\limits_{j\in\mathbb{Z}}(\sum_{k\in\mathbb{Z}}|A_{j,k}|), K_2=\sup\limits_{k'\in\mathbb{Z}}(\sum_{j'\in\mathbb{Z}}|A_{j',k'}|)$. Let $v\in \ell^2(\mathbb{Z})$ be an arbitrary vector. We have
\begin{eqnarray}
    \|Av\|^2=\sum_{j\in\mathbb{Z}}|\left(Av\right)_j|^2=\sum_{j\in\mathbb{Z}} |\sum_{k\in\mathbb{Z}}A_{j,k}v_k |^2
\end{eqnarray}
Writing $A_{j,k}(a)v_k=\sqrt{A_{j,k}} \big(\sqrt{A_{j,k}}v_k\big)$, for any choice of the complex square root, and applying the Cauchy-Schwarz inequality we get
\begin{equation}
\begin{split}
    \|Av\|^2&\leq\sum_{j\in\mathbb{Z}} \sum_{k\in\mathbb{Z}}|A_{j,k}|\sum_{\ell\in\mathbb{Z}}|A_{j,\ell}||v_{\ell}|^2 \leq K_1 \sum_{j\in\mathbb{Z}}\sum_{\ell\in\mathbb{Z}}|A_{j,\ell}||v_{\ell}|^2 \\ 
     &=K_1\sum_{\ell\in\mathbb{Z}}|v_{\ell}|^2 \sum_{j\in\mathbb{Z}}|A_{j,\ell}| \leq K_1K_2 \| v\|^2< \infty \: .
\end{split}
\end{equation}
Therefore $\|A\|\leq \sqrt{K_1K_2}$.
\end{proofline}

\phantom{a}
\begin{lemma}\label{lemma bounded linear operators}
For $a\in\anb$ we have
\begin{enumerate}
    \item $L(a)$ is a bounded linear operator on $\ell^2(\mathbb{Z})$
    \item $T(a)$ is a bounded linear operator on $\ell^2(\mathbb{N})$
    \item $H(a)$ is a Hilbert-Schmidt operator
\end{enumerate}
\end{lemma}

\begin{proofline}
We have $a\in V_\rho$ for some $0<\rho<1$. By Lemma~\ref{lemma bound laurent coefficients} there is $M>0$ such that
\begin{equation}\label{step inequality L annulus}
    |L_{j,k}(a)|\leq M \rho_1^{|j-k|} \; .
\end{equation}
for any $\rho_1\in(\rho,1)$. This implies, by summing the geometric series,
\begin{equation}
   \sum_{j\in\mathbb{Z}}|L_{j,k}(a)|\leq M\frac{1+\rho_1}{1-\rho_1}, \qquad \sum_{j\in\mathbb{Z}}|L_{k,j}(a)|\leq M\frac{1+\rho_1}{1-\rho_1}\qquad \forall k\in\mathbb{Z}  .
\end{equation}
By Lemma~\ref{lemma bound norm infinite matrix} we have that $L(a)$ is bounded.
The Toeplitz part of the lemma follows from
\begin{equation}
    \|T(a)\|=\|PL(a)P \| \leq \|P\| \|L(a)\| \|P\| =\|L(a)\|.
\end{equation}
To prove the Hankel part of the lemma we notice
\begin{equation}
    |H_{j,k}(a)|=|L_{j,-k+1}(a)|\leq M \rho_1^{j+k-1}, \quad j,k\in\mathbb{N} \; .
\end{equation}
Thus, summing the geometric series, we have for the squared Hilbert-Schmidt norm
\begin{equation}
    \sum_{j,k=1}^\infty|H_{j,k}(a)|^2\leq \frac{M\rho_1}{(1-\rho_1)^2} < \infty \: .
\end{equation}
\end{proofline}

\begin{quotation}
\begin{remark}\label{remark bound L sup annulus}
In particular, %from the presented proof it follows that 
for $a\in V_\rho$ we have the bound $\|L(a)\|\leq M\frac{1+\rho_1}{1-\rho_1}$, for any $\rho_1\in (\rho,1)$, where $M=\sup\{{|a_x(z)|}: \ x\in\frac{1}{2}\mathbb{Z}, {\rho_1\leq}|z|\leq \rho_1^{{-1}} \} $.
\end{remark}
\end{quotation}

\begin{proof}[Proof of Remark~\ref{lemma star product analiticity}]
By assumption $a,b\in V_\rho$ for some $\rho\in(0,1)$. Let $\rho<\rho_1<\rho_2<1$. The assumptions imply that there is $M>0$ such that $|(a_x)_j|\leq M \rho_1^{|j|}$, $|(b_x)_j|\leq M \rho_1^{|j|}$. Therefore on the annulus  $\rho_2<|z|< \rho_2^{-1}$ we have
\begin{align}
    &|(a\star b)_x(z)|\leq \sum\limits_{m,n\in\mathbb{Z}} |z|^{m+n}|(a_{x+\frac{m}{2}})_{n}||(b_{x-\frac{n}{2}})_{m}|\leq M^2 \sum\limits_{m,n\in\mathbb{Z}} \left(|z|^m\rho_1^{|m|}\right) \left(|z|^n\rho_1^{|n|}\right)\\
    &=M^2\left(\frac{1}{1-|z|\rho_1}+\frac{1}{1-\frac{\rho_1}{|z|}}-1\right)\leq M^2\left(\frac{1+\frac{\rho_1}{\rho_2}}{1-\frac{\rho_1}{\rho_2}}\right)^2 <\infty\;  .
\end{align}
Clearly, the bound we have obtained is independent of $x$. Since $\rho_1$ and $\rho_2$ are arbitrary the double series is absolutely convergent on the whole annulus $\rho<|z|<\rho^{-1}$. Absolute convergence allows to interchange the sum over $m$ and $n$ in the definition~\eqref{eq:Moyal_star} of the star product. Making also the shift of indices $m\to m-n$, we can write the star product as Laurent series
\begin{equation}
  (a\star b)_x(z)=\sum_{m\in\mathbb{Z}} ((a\star b)_x)_m z^m \; , \quad ((a\star b)_x)_m= \sum_{n\in\mathbb{Z}}\big(a_{x+\frac{m-n}{2}}\big)_n\big(b_{x-\frac{n}{2}}\big)_{m-n} \; .
\end{equation}
The associativity of the star product can now be readily checked by comparing $((a\star b)\star c)_x)_m$ with $((a\star(b\star c))_x)_m$ and the claim $L(a\star b)=L(a)L(b)$ can be checked by comparing the matrix elements of $L(a\star b)$ with the matrix elements of $L(a)L(b)$. 

\end{proof}

% \textcolor{red}{
% \begin{lemma}
% Consider a symbol $a=[f]$ for some function $f:\mathbb R\times (\mathbb{R}/(2\pi\mathbb{Z}))\to \mathbb{C}$, defined by $f(x,p)=\sum_{\ell=M_1}^{M_2}f_\ell(x)e^{i\ell p}$, for some uniformly bounded functions $f_\ell:\mathbb{R}\to\mathbb{C}$, for $\ell=M_1,M_1+1,\ldots M_2$, where $M_1\leq M_2$ are integers. Then
% \begin{equation}
%     \|a\|_{loc}\leq \sum_{n=0}^\infty \frac{1}{2^{n+1}}\sum_{\ell=M_1}^{M_2} \sup_{x\in[-n,n]}|f_\ell(x)| \; .
% \end{equation}
% \end{lemma}
% }

\section{Additional details}\label{a:extra}

\begin{proof}[Proof of Lemma~\ref{lemma truncated moyal analytic}]
Let us denote by $\tilde f, \tilde g $ the Fourier coefficients with respect to the second argument ($\tilde f(y,m)=\frac{1}{2\pi}\int_{-\pi}^\pi f(y,p) e^{-imp} dp$). Pick some $\rho_1\in(\rho,1)$. Due to analyticity we can define the (continuous) functions $K^\Lambda_{j,k}$, $H_{j,n}^\Lambda$, for $j,n=0,1,\ldots,k$ and $\Lambda\geq 0$ by
\begin{equation}
K_{j,k}^\Lambda(y)=\max\{ \sum_{m\in\mathbb Z}|m|^j \sup_{z\in y+[-\Lambda,\Lambda]} |\partial_z^k \tilde f(z,m)|,  \sum_{m\in\mathbb Z}|m|^j \sup_{z\in y+[-\Lambda,\Lambda]} |\partial_z^k \tilde g(z,m)| \} \; ,
\end{equation}
and
\begin{equation}\label{truncation step function H}
H_{j,k}^\Lambda(y)=\rho_1^{-\Lambda}\max\{\sum_{|m|\geq \Lambda}|m|^j |\partial_y^k \tilde f(y,m)|,\sum_{|m|\geq \Lambda}|m|^j|\partial_y^k \tilde g(y,m)| \}
\end{equation}
for $y\in\mathbb{R}$. Let us take some $\Lambda\in\mathbb{N}$ and consider $(x,p)\in\mathbb{R}\times (\mathbb{R}/2\pi\mathbb{Z})$. We start from the definition~\eqref{eq:Moyal_star} of the Moyal star product, from which we get
\begin{multline}\label{eq:hnuyp0}
\left|h_\nu(x,p)-\frac{1}{(2\pi)^2}\sum_{m\in \frac{2[-\Lambda,\Lambda]}{\nu}}\sum_{n\in \frac{2[-\Lambda,\Lambda]}{\nu}} e^{i(m+n)p}\int\limits_{[-\pi,\pi]^2} f(\nu(x+\tfrac{m}{2}),q_1)g(\nu(x-\tfrac{n}{2}),q_2)   e^{-i(nq_1+mq_2)} d^2 q\right|=\\
\left|\frac{1}{(2\pi)^2}\sum_{m\notin \frac{2[-\Lambda,\Lambda]}{\nu}}\sum_{n\in \mathbb Z} e^{i(m+n)p}\int\limits_{[-\pi,\pi]^2} f(\nu(x+\tfrac{m}{2}),q_1)g(\nu(x-\tfrac{n}{2}),q_2)   e^{-i(nq_1+mq_2)} d^2 q\right.\\
\left.+\frac{1}{(2\pi)^2}\sum_{m\in\mathbb{Z}}\sum_{n\notin \frac{2[-\Lambda,\Lambda]}{\nu}} e^{i(m+n)p}\int\limits_{[-\pi,\pi]^2} f(\nu(x+\tfrac{m}{2}),q_1)g(\nu(x-\tfrac{n}{2}),q_2)   e^{-i(nq_1+mq_2)} d^2 q\right|\, .
\end{multline}
Let us bound the last two lines.
For the last line we have
\begin{multline}
\left| \sum_{m\notin \frac{2[-\Lambda,\Lambda]}{\nu}}\sum_{n\in \mathbb Z} e^{i(m+n)p} \tilde f(\nu(x+\tfrac{m}{2}),n) \tilde g(\nu(x-\tfrac{n}{2}),m)   \right|
\leq \sum_{m\notin \frac{2[-\Lambda,\Lambda]}{\nu}}\sum_{n\in \mathbb Z} | \tilde f(\nu(x+\tfrac{m}{2}),n)| |\tilde g(\nu(x-\tfrac{n}{2},m) |\\
\leq \sup_{y\in\mathbb R}\sum_{n\in \mathbb Z} |\tilde f(y,n)| \sup_{y\in\mathbb R}\sum_{m\notin \frac{2[-\Lambda,\Lambda]}{\nu}} |\tilde g(y,m)| \leq K \rho^{\frac{2\Lambda}{\nu}}\, ,
\end{multline}
where
\begin{equation}
    K=\frac{2}{1-\rho_1} \sup\limits_{(y,p)\in \mathbb R\times \strp[\rho]}{\!\!\!\!\!\!\!\!\!\! |f(y,p)|} \sup\limits_{(y,p)\in \mathbb R\times \strp[\rho]}{\!\!\!\!\!\!\!\!\!\! |g(y,p)|} ,
\end{equation}
This is also an upper bound for the other line (after having interchanged $f$ and $g$). Thus the left hand side of \eqref{eq:hnuyp0} is upper bounded by $2K \rho^{\frac{2\Lambda}{\nu}}$. 
We can now use the Taylor's theorem and expand $f$ up to order $k-1$ with the integral form of the reminder. Denoting $y=x\nu$, this gives
\begin{multline}\label{eq:hnuyp1}
\left|h_\nu(\tfrac{y}{\nu},p)-\frac{1}{(2\pi)^2}\sum_{m,n\in \frac{2[-\Lambda,\Lambda]}{\nu}} e^{i(m+n)p}
\int\limits_{[-\pi,\pi]^2} \sum_{j=0}^{k-1} \tfrac{\partial_1^jf(y,q_1)(\nu \frac{m}{2})^j e^{-i n q_1}}{j!} g(y-\tfrac{\nu n}{2},q_2) e^{-i m q_2} d^2 q\right|\leq \\
 2K\rho^{\frac{2\Lambda}{\nu}}+\left|\sum_{m,n\in \frac{2[-\Lambda,\Lambda]}{\nu}} e^{i(m+n)p}
 \int_{0}^{\frac{m}{2}\nu} \frac{(\partial_1^k f_1(y+t))_n(\frac{m}{2}\nu-t)^{k-1}}{(k-1)!}d t  (g_1(y-\tfrac{\nu n}{2}))_m \right|\leq\\
2K \rho^{\frac{2\Lambda}{\nu}}+\frac{(\frac{\nu}{2})^k}{k!}\sum_{n\in \mathbb Z}
\sup_{z\in y+[-\Lambda,\Lambda]} |\partial_z^k \tilde f(z,n)| \sum_{m\in \mathbb Z}|m|^k\sup_{z\in y+[-\Lambda,\Lambda]} |\partial_z^k \tilde g(z,m)|\leq 2K\rho^{\frac{2\Lambda}{\nu}}+\frac{(\frac{\nu}{2})^k}{k!}K_{0,k}^{\Lambda}(y)K_{k,0}^{\Lambda}(y)\, .
\end{multline}
By Taylor expanding also $g$ up to order $k-1-j$ we get
\begin{multline}\label{eq:hnuyp2}
\Bigl|h_\nu(\tfrac{y}{\nu},p)\\
-\frac{1}{(2\pi)^2}\sum_{m,n\in \frac{2[-\Lambda,\Lambda]}{\nu}} e^{i(m+n)p}
\int\limits_{[-\pi,\pi]^2} \sum_{j=0}^{k-1} \tfrac{\partial_1^jf(y,q_1)(\nu \frac{m}{2})^j e^{-i n q_1}}{j!} \sum_{j'=0}^{k-1-j} \tfrac{\partial_1^{j'}g(y,q_2)(-\nu \frac{n}{2})^{j'} e^{-i m q_2}}{j'!} d^2 q\Bigr|\leq \\
2K\rho^{\frac{2\Lambda}{\nu}}+ \left(\frac{\nu}{2}\right)^k\sum_{j=0}^{k}\frac{K_{k-j,j}^{\Lambda}(y)K_{j,k-j}^{\Lambda}(y)}{j!(k-j)!}
\end{multline}
The  expression that approximates $h_\nu(\tfrac{y}{\nu},p)$ can be manipulated by moving the sums in front of everything else and integrating by parts $j'$ times in $q_1$ and $j$ times in $q$, each time integrating out  the phases $e^{-i nq_1}$ and $e^{-i m q_2}$. This results in
\begin{multline}
\Bigl|h_\nu(\tfrac{y}{\nu},p)\\
-\sum_{m,n\in\frac{2[-\Lambda,\Lambda]}{\nu}} \sum_{j=0}^{k-1}\sum_{j'=0}^{k-1-j}\int\limits_{[-\pi,\pi]^2} \left(-i\frac{\nu}{2}\partial_{y_1}\partial_{q_2}\right)^j \left(i\frac{\nu}{2}\partial_{y_2}\partial_{q_1}\right)^{j'} \left[f(y_1,q_1)f(y_2,q_2)e^{-i(nq_1+mq_2)}\right]\Bigr|_{y_1=y_2=y}\leq\\
2K\rho^{\frac{2\Lambda}{\nu}}+ \left(\frac{\nu}{2}\right)^k\sum_{j=0}^{k}\frac{K_{k-j,j}^{\Lambda}(y)K_{j,k-j}^{\Lambda}(y)}{j!(k-j)!}\, .
\end{multline}
Finally, we use eq.~\eqref{truncation step function H} to bound the discrepancy when extending the sums over $m$ and $n$ from $2[-\Lambda,\Lambda]/\nu$ to $\mathbb Z$, and recognize 
\begin{equation}
\begin{split}
   & \sum_{m,n\in\mathbb{Z}}e^{i(m+n)p} \sum_{j=0}^{k-1}\sum_{j'=0}^{k-1-j}\int\limits_{[-\pi,\pi]^2} \left(-i\frac{\nu}{2}\partial_{y_1}\partial_{q_2}\right)^j \left(i\frac{\nu}{2}\partial_{y_2}\partial_{q_1}\right)^{j'} \left[f(y_1,q_1)f(y_2,q_2)e^{-i(nq_1+mq_2)}\right]\Bigr|_{y_1=y_2=y}=\\
   & =h(\nu)(y,p) \ ,
    \end{split}
\end{equation}
which gives the bound
\begin{multline}\label{eq:upperboundprod}
\left|h_\nu(x,p)-(h(\nu))(x\nu,p)\right|\leq \\
\left(\frac{\nu}{2}\right)^k\sum_{j=0}^{k}\frac{K_{k-j,j}^{\Lambda}(x\nu)K_{j,k-j}^{\Lambda}(x\nu)}{j!(k-j)!}+\rho^{\frac{2\Lambda}{\nu}}\left[2K+3\sum_{j=0}^{k-1}\sum_{j'=0}^{k-1-j}\frac{1}{j'!j!}\left(\frac{\nu}{2}\right)^{j+j'}H_{j,j'}^{\tfrac{2\Lambda}{\nu}}(x\nu)K^0_{j',j}(x\nu)\right]
\end{multline}
for any $x\in\mathbb{R}$ and $\nu>0$. Now consider a compact set $\Omega \subset \mathbb{R}$. We have that $K_{j,j'}^\Lambda(x\nu)$ and $H^{2\Lambda/\nu}_{j,j'}(x\nu)$ are bounded uniformly in $x\in\Omega/\nu$, by $\sup_{y\in\Omega} K_{j,j'}^\Lambda(y)$ and $\sup_{y\in\Omega} H^{2\Lambda/\nu}_{j,j'}(y)$ respectively. Since also $\lim_{\nu\to 0}\sup_{y\in\Omega} H^{2\Lambda/\nu}_{j,j'}(y)=0$, we have $h_\nu(x,p)-(h(\nu))(x\nu,p)=O(\nu^k)$ as $\nu\to 0$, uniformly in $x\in\Omega/\nu$.

\end{proof}

\begin{proof}[Proof of Lemma~\ref{l:leftrightexp}]
By assumption all the symbols in the orbit of the flow have left and right factorisations with bounded star logarithm. Let us then impose the Ansatz
\begin{equation}
     f^{R/L}_\pm(x,p)=\sum_{j=0}^{k-1} f^{R/L\pm }_{j}(x\nu,p) \frac{\nu^j}{j!}+O(\nu^{k})
\end{equation}
for some functions $f^{R/L\pm}_{j}(y,p)$, $j=0,1,\ldots,k-1$ and assume that $f^{R/L\pm}_{j}(y,p)$ are of class $C^{k-j}$ in the first argument so that $k-j$ derivatives are analytic with respect to the second argument.
This assumption will be shown to hold by construction. 
We are in a position to apply Lemma~\ref{lemma truncated moyal analytic}, which gives 
\begin{multline}\label{eq:fpolnu}
     f(x\nu,p) =\\
     \sum_{m=0}^k \frac{\nu^m}{m!} \sum_{j_1=0}^m\sum_{j_2=0}^{m-j_1} \binom{m}{j_1, j_2} \left.\left(i\sigma_H\tfrac{\partial_{p_1}\partial_{y_2}-\partial_{y_1}\partial_{p_2}}{2 } \right)^{m-j_1-j_2} f_{j_1}^{H-}(y_1,p_1)f^{H+}_{j_2}(y_2,p_2)\right|_{ y_1=y_2=x\nu \atop p_1=p_2=p}\\
     +O(\nu^{k+1})
\end{multline}
uniformly in $x\nu\in\Omega$ and $p\in\mathbb R/(2\pi\mathbb Z) $, where $H\in\{L,R\}$ and $\sigma_R=-\sigma_L=1$. Treating $\nu$ and $y=x\nu$ as independent variables and imposing that the coefficients of the asymptotic polynomial in $\nu$ shown in \eqref{eq:fpolnu} are zero, we obtain the recurrence equation
\begin{equation}\label{step equation moyal truncation}
    \delta_{m,0} f(y,p)=\sum_{j_1=0}^m\sum_{j_2=0}^{m-j_1} \binom{m}{j_1, j_2} \left.\left(i\sigma_H\frac{\partial_{p_1}\partial_{y_2}-\partial_{y_1}\partial_{p_2}}{2 } \right)^{m-j_1-j_2} f_{j_1}^{H-}(y_1,p_1)f_{j_2}^{H+}(y_2,p_2)\right|_{ y_1=y_2=y \atop p_1=p_2=p}
\end{equation}
For $m=0$, it reads $f=f_{0}^{H-} f_{0}^{H+}$ and is solved by
\begin{equation}
    f_{0}^{H\pm}= e^{(\log f)^\pm}
\end{equation}
For $m=1,2,\ldots,k-1$, we can rewrite Eq.~\eqref{step equation moyal truncation} as
\begin{equation}
    0=f_{m}^{H-}f_{0}^{H+}+f_{0}^{H-}f_{m}^{H+}- f\Omega_m^{H} \ ,
\end{equation}
where we have defined the functions
\begin{equation}
   \Omega_m^{H}(y,p)=-\frac{1}{f}\sum_{j_1=0}^{m-1}\sum_{j_2=0 \atop j_2\neq m}^{m-j_1}\binom{m}{j_1, j_2} \left.\left(i\sigma_H\frac{\partial_{p_1}\partial_{y_2}-\partial_{y_1}\partial_{p_2}}{2 } \right)^{m-j_1-j_2} f_{j_1}^{H-}(y_1,p_1)f_{j_2}^{H+}(y_2,p_2)\right|_{ y_1=y_2=y \atop p_1=p_2=p}\, .
\end{equation}
Dividing by $f=f_{0}^{H-}f_{0}^{H+}$ we get
\begin{equation}
    e^{-(\log f)^-}f_{m}^{H-}+e^{-(\log f)^+}f_{m}^{H+}= \Omega_m^{H} \ ,
\end{equation}
which is solved by
\begin{equation}\label{eq:fmHpm}
    f_{m}^{H\pm}=e^{(\log f)^\pm}\left(\Omega_m^{H}\right)^{\pm} , \quad m=1,2,\ldots, k-1
\end{equation}
We can then rewrite $\Omega_m^{H}(y,p)$ as
\begin{multline}
   \Omega_m^{H}(y,p)=-e^{-(\log f)^-(y_1,p_1)-(\log f)^+(y_2,p_2)}\sum_{j_1=0}^{m-1}\sum_{j_2=0 \atop j_2\neq m}^{m-j_1}\binom{m}{j_1, j_2}
  \left(i\sigma_H\frac{\partial_{p_1}\partial_{y_2}-\partial_{y_1}\partial_{p_2}}{2 } \right)^{m-j_1-j_2} \\\left. e^{(\log f)^-(y_1,p_1)}\Omega_{j_1}^{H-}(y_1,p_1)e^{(\log f)^+(y_2,p_2)}\Omega_{j_2}^{H+}(y_2,p_2)\right|_{ y_1=y_2=y \atop p_1=p_2=p}=\\
  -\sum_{j_1=0}^{m-1}\sum_{j_2=0 \atop j_2\neq m}^{m-j_1}\binom{m}{j_1, j_2}
  \Bigl(i\tfrac{\sigma_H}{2}\bigl(e^{-(\log f)^-(y_1,p_1)}\partial_{p_1}e^{(\log f)^-(y_1,p_1)}e^{-(\log f)^+(y_2,p_2)}\partial_{y_2}e^{(\log f)^+(y_2,p_2)}\\
  -e^{-(\log f)^-(y_1,p_1)}\partial_{y_1}e^{(\log f)^-(y_1,p_1)}e^{-(\log f)^+(y_2,p_2)}\partial_{p_2}e^{(\log f)^+(y_2,p_2)} \bigr)\Bigr)^{m-j_1-j_2} \\\left. \Omega_{j_1}^{H-}(y_1,p_1)\Omega_{j_2}^{H+}(y_2,p_2)\right|_{ y_1=y_2=y \atop p_1=p_2=p}\, ,
\end{multline}
which, using $e^{-g} \partial e^g=\partial+(\partial g)$, reduces to $\psi_m^H[\nabla \log f](y,p)$, as defined in \eqref{eq:inversion}, and hence \eqref{eq:fmHpm} implies \eqref{eq:leftrightsmooth}. 
Finally, $f^{R/L\pm }_{j}(x\nu,p) $ as defined are $C^{k-j}$ in the first argument by construction, and their $k-j$ derivatives are analytic in the second argument. 
We have thus shown that the symbols $d^{R/L}_\pm$ represented by the functions $(x,p)\mapsto \sum_{j=0}^{k-1} f^{R/L\pm }_{j}(x\nu,p) \frac{\nu^j}{j!}$ satisfy $d^L_+\star d_-^L=a+O(\nu^{k})$, $d^R_-\star d_+^R=a+O(\nu^{k})$. Lemma~\ref{lemma approximation left right} implies then $a_\pm^{R/L}=d_{\pm}^{R/L}+O(\nu^{k})$.
\end{proof}

\begin{proof}[Proof of Lemma~\ref{l:gauge}]
The first property follows from the chain of identities
\begin{multline}
\lim_{m\rightarrow\infty}\lim_{\epsilon\rightarrow 0}\frac{1}{2\pi}\int_{\frac{1}{2}}^{n+\frac{1}{2}}\int_{-\pi}^\pi \frac{\prod_{i=1}^k\partial_{1}^{c_i}\partial_{2}^{d_i}(g(x,p)+\epsilon \phi_m(x,p))-\prod_{i=1}^k\partial_{1}^{c_i}\partial_{2}^{d_i} (g(x,p))}{\epsilon} dp dx=\\
\lim_{m\rightarrow\infty}\frac{1}{2\pi}\int_{\frac{1}{2}}^{n+\frac{1}{2}}\int_{-\pi}^\pi \sum_{j=1}^k\partial_{1}^{c_j}\partial_{2}^{d_j}\phi_m(x,p)\prod_{i\neq j}^{k}\partial_{1}^{c_i}\partial_{2}^{d_i}g(x,p) dp dx=\\
\lim_{m\rightarrow\infty}\frac{1}{2\pi}\int_{\frac{1}{2}}^{n+\frac{1}{2}}\int_{-\pi}^\pi \sum_{j=1}^k\phi_m(x,p)(-\partial_{1})^{c_j}(-\partial_{2})^{d_j}\left(\prod_{i\neq j}^{k}\partial_{1}^{c_i}\partial_{2}^{d_i}g(x,p)\right) dp dx=\\
\frac{1}{2\pi}\int_{\frac{1}{2}}^{n+\frac{1}{2}}\int_{-\pi}^\pi \sum_{j=1}^k\phi(x,p)(-\partial_{1})^{c_j}(-\partial_{2})^{d_j}\left(\prod_{i\neq j}^{k}\partial_{1}^{c_i}\partial_{2}^{d_i}g(x,p)\right) dp dx
\end{multline}
with the identification of $\delta_g F[g](\phi)$ with the last line. For the second identity, we plug $\delta_g F[g](\phi)$ in \eqref{eq:polynomialbulk} and integrate over $s$, which gives \eqref{eq:Gproof}.
We can now compute the variation under a perturbation with a bulk function $\phi$
\begin{multline}
\delta_g G[g](\phi)=\frac{1}{2\pi}\int_{\frac{1}{2}}^{n+\frac{1}{2}}\int_{-\pi}^\pi \phi(x,p) \frac{1}{k}\sum_{j=1}^k(-\partial_{1})^{c_j}(-\partial_{2})^{d_j}\left(\prod_{i\neq j}^{k}\partial_{1}^{c_i}\partial_{2}^{d_i}g(x,p)\right) dp dx\\
+\frac{1}{2\pi}\int_{\frac{1}{2}}^{n+\frac{1}{2}}\int_{-\pi}^\pi g(x,p) \frac{1}{k}\sum_{j=1}^k(-\partial_{1})^{c_j}(-\partial_{2})^{d_j}\left(\sum_{l\neq j}^k \partial_{1}^{c_l}\partial_{2}^{d_l}\phi(x,p) \prod_{i\neq j,l}^{k}\partial_{1}^{c_i}\partial_{2}^{d_i}g(x,p)\right) dp dx
\end{multline}
If we integrate the integral in the second line by parts $c_j$ times  with respect to $x$ and $d_j$ times with respect to $p$, the boundary terms vanish because the derivatives of the bulk function are bulk functions. Thus we have
\begin{multline}
\delta_g G[g](\phi)=\frac{1}{2\pi}\int_{\frac{1}{2}}^{n+\frac{1}{2}}\int_{-\pi}^\pi \phi(x,p) \frac{1}{k}\sum_{j=1}^k(-\partial_{1})^{c_j}(-\partial_{2})^{d_j}\left(\prod_{i\neq j}^{k}\partial_{1}^{c_i}\partial_{2}^{d_i}g(x,p)\right) dp dx\\
+\frac{1}{2\pi}\int_{\frac{1}{2}}^{n+\frac{1}{2}}\int_{-\pi}^\pi \frac{1}{k}\sum_{j=1}^k\partial_{1}^{c_j}\partial_{2}^{d_j} g(x,p)  \sum_{l\neq j}^k \partial_{1}^{c_l}\partial_{2}^{d_l}\phi(x,p) \prod_{i\neq j,l}^{k}\partial_{1}^{c_i}\partial_{2}^{d_i}g(x,p) dp dx=\\
\frac{1}{2\pi}\int_{\frac{1}{2}}^{n+\frac{1}{2}}\int_{-\pi}^\pi \phi(x,p) \frac{1}{k}\sum_{j=1}^k(-\partial_{1})^{c_j}(-\partial_{2})^{d_j}\left(\prod_{i\neq j}^{k}\partial_{1}^{c_i}\partial_{2}^{d_i}g(x,p)\right) dp dx\\
+\frac{1}{2\pi}\int_{\frac{1}{2}}^{n+\frac{1}{2}}\int_{-\pi}^\pi  \phi(x,p)\frac{1}{k}\sum_{j=1}^k\sum_{l\neq j}^k (-\partial_{1})^{c_l}(-\partial_{2})^{d_l}\left( \prod_{i\neq l}^{k}\partial_{1}^{c_i}\partial_{2}^{d_i}g(x,p) \right)dp dx=\\
\frac{1}{2\pi}\int_{\frac{1}{2}}^{n+\frac{1}{2}}\int_{-\pi}^\pi \phi(x,p) \sum_{j=1}^k(-\partial_{1})^{c_j}(-\partial_{2})^{d_j}\left(\prod_{i\neq j}^{k}\partial_{1}^{c_i}\partial_{2}^{d_i}g(x,p)\right) dp dx\, ,
\end{multline}
quod erat demonstrandum.
\end{proof}

\begin{lemma}\label{t:entire}
Consider two functions $f,g:\mathbb C\times\strp[\rho]\rightarrow\mathbb C$ that are entire with respect to their first argument, analytic with respect to the second argument and satisfy
\begin{equation}
    \left| \frac{\partial^n}{\partial x^n} f(x,p)\right|\leq M ^n , \quad  \left| \frac{\partial^n}{\partial x^n} g(x,p)\right|\leq M ^n \qquad  \quad \forall x\in\mathbb{R}, \ \forall p\in \Sigma_{\rho}^{(2\pi)}, \ \forall n\in\mathbb{N}
\end{equation}
for some $M>0$. 
If $M<2|\log\rho|$ then  $[f]\star[g]=[h]$, with the function $h:\mathbb R\times\strp[\rho]\to \mathbb{C}$ defined as
    \begin{equation}
    \begin{split}\label{eq:Moyal_starD}
h(x,p)&= \left. e^{i\frac{\partial_{p}\partial_{y}-\partial_{q} \partial_{ x}}{2}} f(x,p)g(y,q)\right|_{q=p\atop y=x}\\
&=\sum_{m=0}^\infty\sum_{n=0}^\infty \frac{i^{m-n}}{2^{m+n}m!n!}\frac{\partial^m}{\partial p^m} \frac{\partial^m}{\partial y^m}\frac{\partial^n}{\partial q^n}\frac{\partial^n}{\partial x^n}\left[f(x,p)g(y,q) \right]_{q=p\atop y=x}\  ,\quad  x\in\mathbb{R}, \  p\in \Sigma_{\rho}^{(2\pi)}\, ,
\end{split}
\end{equation}
where the double series is absolutely convergent. 
\end{lemma}
\begin{proofline}
We start from the right hand side of Eq.~\eqref{eq:Moyal_starD} and show that it equals the star product. By assumption we have
\begin{equation}\label{step inequality derivative}
  \left|\frac{1}{2\pi}\int\limits_{-\pi}^\pi e^{-i j q}\frac{\partial^n}{\partial x^n}f(x,q) dq\right|= \left|\frac{1}{2\pi i}\oint_{|z|=\rho_1^{-\mathrm{sgn}(j)}}z^{-j-1} \frac{\partial^n}{\partial x^n}f(x,p)|_{e^{ip}=z}dz\right|  \leq M^n\rho_1^{|j|}
\end{equation}
for any $\rho_1\in(\rho,1)$.
By  Fourier expanding  $f$ and $g$ we obtain terms of the form
 \begin{equation}
     \frac{\partial^m}{\partial p^m}\frac{\partial^n}{\partial x^n}\sum_{j\in\mathbb{Z}} e^{ijp} \frac{1}{2\pi}\int_{-\pi}^\pi
 f(x,q_1)e^{-i q_1j} dq_1= \sum_{j\in\mathbb{Z}} (ij)^m e^{ijp} \frac{1}{2\pi}\int_{-\pi}^\pi
 \frac{\partial^n}{\partial x^n}f(x,q_1)e^{-i q_1j} dq_1 \ ,
 \end{equation}
where that the derivatives and the series can be interchanged because the resulting series are uniformly convergent (this can be seen, e.g., using the bound~\eqref{step inequality derivative} to set a Weierstrass M-test).
We are thus left with
 \begin{equation}\label{step moyal before intarechanging the order of summation}
 \begin{split}
     &\left. e^{i\frac{\partial_{p}\partial_{y}-\partial_{q} \partial_{ x}}{2}}f(x,p)g(y,q)\right|_{q=p\atop y=x}\\
& = \sum_{m=0}^\infty\sum_{n=0}^\infty\sum_{j\in\mathbb{Z}}\sum_{k\in\mathbb{Z}} \frac{ (-j)^m k^n }{m!n! 2^{m+n}}e^{i(j+k)p} \frac{1}{(2\pi)^2}\!\!\!\int\limits_{[-\pi,\pi]^2} \frac{\partial^n}{\partial x^n}f(x,q_1) \frac{\partial^m}{\partial x^m}g(x,q_2)e^{-i(q_1 j+q_2k)} d^2q
    \end{split}
 \end{equation}
This series is absolutely convergent, indeed, using \eqref{step inequality derivative}, we have
\begin{align}
& \sum_{m=0}^\infty\sum_{n=0}^\infty\sum_{j\in\mathbb{Z}}\sum_{k\in\mathbb{Z}} \frac{ 1 }{m!n!}\left|\frac{j}{2}\right|^{m}\left|\frac{k}{2}\right|^{n}\left|\ \frac{1}{(2\pi)^2}\int\limits_{[-\pi,\pi]^2} \frac{\partial^n}{\partial x^n}f(x,q_1) \frac{\partial^m}{\partial x^m}g(x,q_2) d^2q\right| \\
& \leq \sum_{m=0}^\infty\sum_{n=0}^\infty\sum_{j\in\mathbb{Z}}\sum_{k\in\mathbb{Z}} \frac{ 1 }{m!n!}\left|\frac{j}{2}\right|^{m}\left|\frac{k}{2}\right|^{n} \rho_1^{|j|+|k|} M^{m+n}\\
& = \sum_{j\in\mathbb{Z}}\sum_{k\in\mathbb{Z}} \exp\left(\frac{M}{2}|j|+\frac{M}{2}|k|\right)\rho_1^{|j|+|k|},
\end{align}
which converges for any $\rho_1$ such that $M/2+\log\rho_1 <0$. We can therefore interchange the order of summation in \eqref{step moyal before intarechanging the order of summation}, which gives
\begin{multline}
\left. e^{i\frac{\partial_{p}\partial_{y}-\partial_{q} \partial_{ x}}{2}}f(x,p)g(y,q)\right|_{q=p\atop y=x}=\\
\sum_{j,k\in\mathbb{Z}}e^{i(j+k)p} \frac{1}{(2\pi)^2}\int\limits_{[-\pi,\pi]^2} e^{-i(q_1 j+q_2k)}  \sum_{n=0}^\infty \frac{(\frac{k}{2})^n}{n!} \frac{\partial^n}{\partial x^n}f(x,q_1)\sum_{m=0}^\infty \frac{(-\frac{j}{2})^m}{m!} \frac{\partial^m}{\partial x^m}g(x,q_2) d^2q=\\
\sum_{j,k\in\mathbb{Z}}e^{i(j+k)p} \frac{1}{(2\pi)^2}\int\limits_{[-\pi,\pi]^2} e^{-i(q_1 j+q_2k)}   f(x+\tfrac{k}{2},q_1)g(x-\tfrac{j}{2},q_2)d^2q=([f]\star [g])_x(e^{ip}),
 \end{multline}
 where in the third equality we used that the function is entire and hence equal to its Taylor series (uniformly convergent in all variables),  and in the fourth equality we recognised the star product~\eqref{eq:Moyal_star}.
 
\end{proofline}

\end{appendices}

\bibliography{sn-bibliography}% common bib file

\end{document}